\documentclass[12pt]{article}
\usepackage[round,colon,authoryear]{natbib}
\usepackage{bibentry}
\RequirePackage[OT1]{fontenc}
\RequirePackage{amsmath,amsthm,amsfonts,amssymb,bbm} 
\RequirePackage[colorlinks]{hyperref}
\RequirePackage{epsfig, graphicx, color}

\usepackage{enumitem}   

\allowdisplaybreaks

\usepackage{fullpage}
\usepackage{graphics,graphicx}
\usepackage{subcaption}
\usepackage{booktabs}

\usepackage[titletoc,title]{appendix}  

\usepackage{caption}
\usepackage{subcaption}

\usepackage{multirow}
\usepackage{graphicx}

\usepackage[table,xcdraw,dvipsnames]{xcolor}   
\usepackage{hyperref}
\newcommand\myshade{85}
\colorlet{mylinkcolor}{YellowOrange}
\colorlet{mycitecolor}{Aquamarine}
\colorlet{myurlcolor}{violet}

\hypersetup{
    linkcolor  = mylinkcolor!\myshade!black,
    citecolor  = mycitecolor!\myshade!black,
    urlcolor   = myurlcolor!\myshade!black,
    colorlinks = true,
}

\usepackage[linesnumbered,ruled,vlined]{algorithm2e}

\SetCommentSty{mycommfont}

\SetKwInput{KwInput}{Input}                
\SetKwInput{KwOutput}{Output}              

\newtheorem{theorem}{Theorem}
\newtheorem{lemma}{Lemma}
\newtheorem{corollary}{Corollary}
\newtheorem{assump}{Assumption}

\newtheorem{remarkx}{Remark}

\renewcommand{\hat}{\widehat}
\renewcommand{\tilde}{\widetilde}

\newcommand{\bfm}[1]{\ensuremath{\mathbf{#1}}}

\def\bzero{\bfm 0}
\def\bone{\bfm 1}

   \def\bA{\bfm A}  
\def\bb{\bfm b}   \def\bB{\bfm B}  
     
   \def\bD{\bfm D}  
\def\be{\bfm e}   \def\bE{\bfm E}  \def\EE{\mathbb{E}}
  \def\bF{\bfm F}  
\def\bg{\bfm g}   \def\bG{\bfm G}  
\def\bh{\bfm h}   \def\bH{\bfm H}  
   \def\bI{\bfm I}

   \def\bM{\bfm M}  
     
   \def\bO{\bfm O}  \def\OO{\mathbb{O}}
   \def\bP{\bfm P}  \def\PP{\mathbb{P}}
   \def\bQ{\bfm Q}  
   \def\bR{\bfm R}  \def\RR{\mathbb{R}}
\def\bs{\bfm s}   \def\bS{\bfm S}  
\def\bt{\bfm t}     
\def\bu{\bfm u}   \def\bU{\bfm U}  
\def\bv{\bfm v}     
\def\bw{\bfm w}   \def\bW{\bfm W}  
\def\bx{\bfm x}   \def\bX{\bfm X}  
\def\by{\bfm y}   \def\bY{\bfm Y}  
   \def\bZ{\bfm Z}

\def\calA{{\cal  A}} 
\def\calB{{\cal  B}} 
\def\calC{{\cal  C}} 
 
\def\calE{{\cal  E}} 
\def\calF{{\cal  F}} \def\cF{{\cal  F}}
\def\calG{{\cal  G}}

\def\calM{{\cal  M}} 
\def\calN{{\cal  N}} 
 
\def\calP{{\cal  P}}

\def\calS{{\cal  S}} 
 
\def\calU{{\cal  U}}

\def\calX{{\cal  X}} 
\def\calY{{\cal  Y}}

\newcommand{\bfsym}[1]{\ensuremath{\boldsymbol{#1}}}

\def\bgamma{\bfsym {\gamma}}             \def\bGamma{\bfsym {\Gamma}}

             \def\bSigma{\bfsym \Sigma}
\def\blambda {\bfsym {\lambda}}        \def\bLambda {\bfsym {\Lambda}}

\def\bzeta{\bfsym {\zeta}}
          
\def\bphi{\bfsym {\phi}}          \def\bPhi{\bfsym {\Phi}}

\providecommand{\abs}[1]{\left\lvert#1\right\rvert}
\providecommand{\norm}[1]{\left\lVert#1\right\rVert}

\providecommand{\paran}[1]{\left( #1 \right)}
\providecommand{\brackets}[1]{\left[ #1 \right]}
\providecommand{\braces}[1]{\left\{ #1 \right\}}

\DeclareMathOperator{\rank}{rank}
\DeclareMathOperator*{\argmax}{arg\,max}

\def\bLam{\bLambda}
\def\eps{\varepsilon}

\newcommand{\E}[1]{{\mathbb{E}} \left[ #1 \right]}

\providecommand{\defeq}{\triangleq}

\def\vec{\mbox{vec}}

\def\diag{{\rm diag}}
\def\hooi{\textsf{\tiny HOOI}}

\newcommand{\pkg}[1]{{\normalfont\fontseries{b}\selectfont #1}}
\let\proglang=\textsf

\usepackage{setspace}
\usepackage{sectsty}
\usepackage{etoolbox}
\BeforeBeginEnvironment{equation*}{\vspace*{-.6\baselineskip}}
    \AfterEndEnvironment{equation*}{\noindent\ignorespaces}
    
\BeforeBeginEnvironment{equation}{\vspace*{-.6\baselineskip}}
    \AfterEndEnvironment{equation}{\noindent\ignorespaces}

\BeforeBeginEnvironment{eqnarray}{\vspace*{-.6\baselineskip}}
    \AfterEndEnvironment{eqnarray}{\noindent\ignorespaces}
\BeforeBeginEnvironment{eqnarray*}{\vspace*{-.6\baselineskip}}
    \AfterEndEnvironment{eqnarray*}{\noindent\ignorespaces}

\BeforeBeginEnvironment{align}{\vspace*{-.6\baselineskip}}
    \AfterEndEnvironment{align}{\noindent\ignorespaces}
\BeforeBeginEnvironment{align*}{\vspace*{-.6\baselineskip}}
    \AfterEndEnvironment{align*}{\noindent\ignorespaces}

\BeforeBeginEnvironment{gather}{\vspace*{-.6\baselineskip}}
    \AfterEndEnvironment{gather}{\noindent\ignorespaces}
\BeforeBeginEnvironment{gather*}{\vspace*{-.6\baselineskip}}
    \AfterEndEnvironment{gather*}{\noindent\ignorespaces}

\usepackage{comment}

\newcommand{\msd}[2]{\renewcommand{\arraystretch}{0.8} \begin{tabular}[t]{@{}c@{}}#1\\ {\footnotesize (#2)}\end{tabular}}

\usepackage{xr}
\makeatletter
\newcommand*{\addFileDependency}[1]{
	\typeout{(#1)}
	\@addtofilelist{#1}
	\IfFileExists{#1}{}{\typeout{No file #1.}}
}\makeatother

\newcommand*{\myexternaldocument}[1]{%
	\externaldocument{#1}%
	\addFileDependency{#1.tex}%
	\addFileDependency{#1.aux}%
}

\myexternaldocument{stefa-suppl}

\begin{document}

\title{Semiparametric Tensor Factor Analysis by \\ Iteratively Projected SVD}

\author{Elynn Y. Chen$^1$, Dong Xia$^2$, Chencheng Cai$^3$ and Jianqing Fan$^{4, 5}$\\ \normalsize
$^1$New York University, USA, \\ \normalsize
$^2$Hong Kong University of Science and Technology, China\\ \normalsize
$^3$Washington State University, USA, \\ \normalsize
$^4$ Fudan University, China; \quad
$^5$  Princeton University, USA}

\date{}


\footnotetext[1]{
Fan is the corresponding author. Email: jqfan@princeton.edu.  
$^{1,2}$ Equal contribution. }
   
\maketitle

\begin{abstract}
This paper introduces a general framework of Semiparametric TEnsor Factor Analysis (STEFA) that focuses on the methodology and theory of low-rank tensor decomposition with auxiliary covariates.
STEFA models extend tensor factor models by incorporating auxiliary covariates in the loading matrices.
We propose an algorithm of Iteratively Projected SVD (IP-SVD) for the semiparametric estimation.
It iteratively projects tensor data onto the  linear space spanned by the basis functions of covariates and applies SVD on matricized tensors over each mode.
We establish the convergence rates of the loading matrices and the core tensor factor.
The theoretical results only require a sub-exponential noise distribution, which is weaker than the assumption of sub-Gaussian tail of noise in the literature.
Compared with the Tucker decomposition, IP-SVD yields more accurate estimators with a faster convergence rate.
Besides estimation, we propose several prediction methods with new covariates based on the STEFA model.
On both synthetic and real tensor data, we demonstrate the efficacy of the STEFA model and the IP-SVD algorithm on both the estimation and prediction tasks.
\end{abstract}

\section{Introduction}

Nowadays large-scale datasets in the format of matrices and tensors (or multi-dimensional arrays) routinely arise in a wide range of applications.
The low-rank structure, among other specific geometric configurations, is of paramount importance to enable statistically and computationally efficient analysis of such datasets.
The low-rank {\em tensor factor models} assume the following {\em noisy} Tucker decomposition:
\begin{equation} \label{eqn:tensor-factor-model-0}
\calY = \calF \times_1 \bA_1 \times_2 \cdots \times_M \bA_M + \calE,
\end{equation}
where $\calY$ is the $M$-th order tensor observation of dimension $I_1 \times \cdots \times I_M$, the \textit{latent tensor factor} $\calF$ is of dimension $R_1 \times \cdots \times R_M$, the \textit{loading} matrix $\bA_m$ is of dimension $I_m \times R_m$ with $R_m \ll I_m$ for each $m\in [M]$, and the noise $\calE$ is an $M$-th order tensor with the same dimension of $\calY$.
Tucker decomposition is a widely used form of tensor decomposition  \citep{kolda2009tensor, de2000multilinear} and  has been studied from different angles in mathematics, statistics and computer science.
Particularly, the statistical and computational properties of the decomposition have been analyzed in \cite{zhang2018tensor,richard2014statistical,allen2012sparse,allen2012regularized,wang2017tensor,zhang2019cross} under the general setting
and in \cite{zhang2019optimal} under the sparsity setting where parts of the loading matrices $\braces{\bA_m: m\in[M]}$ contain row-wise sparsity structures.

Tucker decomposition to the tensor factor model is similar to singular value decomposition (SVD) to the classic vector factor model, the latter being one of the most useful tools for modeling low-rank structures in biology, psychometrics, economics, business and so on.
Theoretical analyses on multivariate factor models assume i.i.d Gaussian noise at early stages \citep{anderson1956statistical,anderson1962introduction} and later  allow for variable-wise and sample-wise correlations \citep{bai2002determining,bai2003inferential,bai2012statistical}.
Chapter 11 of \cite{fan2020statistical} and the references therein provide a thorough review of recent advances and applications of multivariate factor models.
For 2nd-order tensor (or matrix) data, \cite{wang2019factor, chen2019constrained, chen2020multivariate} consider the matrix factor model which is a special case of \eqref{eqn:tensor-factor-model-0} with $M=2$ and propose estimation procedures based on the second moments.
Later, \cite{chen2019factor} extends the idea to the model \eqref{eqn:tensor-factor-model-0} with arbitrary $M$ by using the mode-wise auto-covariance matrices.

While the vanilla tensor factor model \eqref{eqn:tensor-factor-model-0} is neat and fundamental, it cannot incorporate any additional information that may be relevant.
Nowadays, the boom of data science has brought together informative covariates from different domains and multiple sources, in addition to the tensor observation $\calY$.
For example, the gene expression measurements from breast tumors can be cast in a tensor format, and the relevant covariates of the cancer subtypes are usually viewed as a partial driver of the underlying patterns of genetic variation among breast cancer tumors \citep{schadt2005integrative,li2016supervised}.
In restaurant recommendation system, online review sites like Yelp have access to shopping histories and friendship networks of customers, as well as the cuisine and ratings of restaurants \citep{acar2011all}.
The covariate-assisted factor models have been explored for vector and matrix observations \citep{connor2007semiparametric,connor2012efficient,fan2016projected, mao2019matrix}.
Their results show that sharing relevant covariate information across datasets leads to not only a more accurate estimation but also a better interpretation.

Inspired by those prior arts, we introduce a new modeling framework -- Semi-parametric TEnser Factor Analysis (STEFA) model -- to leverage the auxiliary information provided by mode-wise covariates.
STEFA captures practically important situations in which the observed tensor $\calY$ has an intrinsic low rank structure and the structure in $m$-th mode is {\em partially} explainable by some relevant covariate $\bX_m$.
The model is semi-parametric in the sense that it still allows covariate-free low-rank factors as in \eqref{eqn:tensor-factor-model-0}.
In the special case when $\bX_m$'s are unavailable, STEFA reduces to the classical tensor factor model  \eqref{eqn:tensor-factor-model-0}.
As to be shown in Section~\ref{sec:appl}, with auxiliary covariates, our STEFA model can outperform the vanilla tensor factor model in many scenarios. 
The auxiliary information of $\bX_m$ not only improves the performances of estimating latent factors but also enables prediction on new input covariates, which is an essential difference between our proposed framework and the existing tensor decomposition literature \citep{richard2014statistical,zhang2018tensor,zhang2019optimal,cai2019nonconvex,sun2017provable,wang2018learning,zhou2021partially}.
Indeed, unlike those tensor SVD or PCA models where estimating the latent factors usually only acts as the proxy of dimension reduction, STEFA utilizes those estimators for prediction with new observed covariates.
Another popular way of incorporating auxiliary covariates information is to couple tensors and matrix covariates together for joint factorization \citep{acar2011all,song2019tensor}.
Such method assumes that the covariate matrix and tensor share the same loading matrix along one mode.
Our method is different in that auxiliary covariates can partially predict loading matrices through nonparametric function approximation.
\cite{hao2019sparse} also used additive model in nonparametric tensor regression. 
But those authors dealt with tensor predictors and scalar responses, rather than a tensor of responses.

On the methodological aspect, we propose a computationally efficient algorithm, called Iteratively Projected SVD (IP-SVD), to estimate both the covariate-relevant loadings and covariate-independent loadings in STEFA.
As shown in Section~\ref{sec:thm}, a typical projected PCA method from \cite{fan2016projected}, while computationally fast, is generally  sub-optimal because it ignores multi-dimensional tensor structures.
The IP-SVD yields more accurate estimators of both the latent factors and loadings by adding a simple iterative projection after the initialization by projected PCA.
On the other hand, the IP-SVD can be viewed as an alternating minimization algorithm which solves a constrained tensor factorization  program where the low-rank factors are constrained to a certain functional space.
The dimension of this functional space, based on the order of sieve approximation, can be  significantly smaller than the ambient dimension which makes IP-SVD faster than the standard High-Order Orthogonal Iteration (HOOI) for solving the vanilla Tucker decomposition.
As a result, IP-SVD requires also weaker signal-to-noise ratio conditions for convergence in general.

Theoretically, we discovered interesting properties of STEFA that are different from those of the vanilla tensor factor model \eqref{eqn:tensor-factor-model-0}.
As proved in  \cite{richard2014statistical, zhang2018tensor}, the HOOI algorithm achieves statistically optimal convergence rates for model \eqref{eqn:tensor-factor-model-0} as long as the signal-to-noise ratio ${\rm SNR} \gtrsim (I_1I_2\cdots I_M)^{1/4}$ where the formal definition of SNR is deferred to Section~\ref{sec:thm}.
However,  due to the constraint of a low-dimensional (compared with $I_m$) functional space, the SNR condition required by IP-SVD in STEFA is ${\rm SNR} \gtrsim  (J_1J_2\cdots J_M)^{1/4}$
{ where $J_m$ is the number of basis function used in functional approximation and can be much smaller than $I_m$.}
Note that this weaker SNR condition is sufficient even for estimating the covariate-independent components.
Surprisingly, it shows that covariate information is not only beneficial to estimating the covariate-relevant components but also to the covariate-independent components.
Concerning the statistical convergence rates of IP-SVD, there are two terms which comprise of a parametric rate and a non-parametric rate.
By choosing a suitable order for sieve approximation, we can obtain a typical semi-parametric convergence rate for STEFA which fills a void of understanding non-parametric ingredients of tensor factor models. 
On the technical front, investigating the theoretical properties of STEFA is challenging due to the iterative nature of the estimation procedure, which involves both a parametric and non-parametric component.
Furthermore, our theoretical results only require a sub-exponential tail on the noise, which is weaker than the Gaussian or sub-Gaussian distributions of noise in all these prior works.
This technical improvement may be of independent interests.

\paragraph{Notation and organization.}
The following notations are used throughout the paper. 
We use lowercase letter $x$, boldface letter $\bx$, boldface capital letter $\bX$, and calligraphic letter $\calX$ to represent scalar, vector, matrix and tensor, respectively.
We denote $[N]=\braces{1, \ldots, N}$ for a positive integer $N$.
For any matrix $\bX$,  we use $\bx_{i \cdot}$, $\bx_{\cdot j}$, and $x_{ij}$ to refer to its $i$-th row, $j$-th column, and $ij$-th entry, respectively.
All vectors are column vectors and row vectors are written as $\bx^\top$.
The set of $N\times K$ orthonormal matrices is defined as $\OO^{N \times K}$. 
We denote $\sigma_i(\bX)$ as the $i$-th largest singular value of $\bX$,  $\|\bX\|$ as the spectral norm of $\bX$, i.e., $\|\bX\| = \sigma_1(\bX)$, and $\|\bX\|_{\rm F}$ as the Frobenius norm of $\bX$.
In addition, we frequently use the projection matrices $\bP_X = \bX \paran{\bX^\top\bX}^{-1} \bX^\top$ and $\bP_X^\perp = \bI - \bP_X$ where $\paran{\bX^\top\bX}^{-1}$ denotes the Moore-Penrose generalized inverse.

The rest of this paper is organized as follows.
Section \ref{sec:model} introduces the STEFA model and a set of identification conditions.
Section \ref{sec:est} proposes the IP-SVD algorithm to estimate the STEFA model and considers prediction with new covariates.
Section \ref{sec:thm} establishes theoretical properties of the estimators.
Section \ref{sec:simu} studies the finite sample performance via simulations.
Section \ref{sec:appl} presents empirical studies of two real data sets.
All proofs and technique lemmas are relegated to the supplementary material. 

\section{STEFA: Semi-parametric TEnsor FActor model}    \label{sec:model}

In this section, we introduce the Semi-parametric TEnsor FActor (STEFA) model. 
We present it with third-order tensors ($M=3$) to simply notation while the properties hold for general $M$.
More information of tensor algebra can be found in \cite{kolda2009tensor}. 

\subsection{Tensor factor model}  \label{sec:TFM}

For a tensor $\calS \in \RR^{I_1\times I_2\times I_3}$, the mode-1 slices of $\calS$ are matrices $\bS_{i_1 : :}\in\RR^{I_2\times I_3}$ for any $i_1\in[I_1]$ and the mode-1 fibers of $\calS$ are vectors $\bs_{: i_2 i_3} \in \RR^{I_1}$ for any $i_2 \in [I_2]$ and $i_3 \in [I_3]$.
We define its mode-1 matricization as a $I_1 \times I_2I_3$ matrix $\calM_1(\calS)$ such that
$
\brackets{ \calM_1(\calS)}_{i_1, i_2 + (i_3-1)I_2} = s_{i_1 i_2 i_3}, \text{for all } i_1 \in [I_1], i_2 \in [I_2], \text{ and } i_3 \in [I_3].
$
In other words, matrix $\calM_1(\calS)$ consists of all mode-1 fibers of $\calS$ as columns.
For a tensor $\calF \in \RR^{R_1 \times R_2 \times R_3}$ and a matrix $\bA_1 \in \RR^{I_1 \times R_1}$, the \textit{mode-$1$ product} is a mapping defined as $\times_1: \RR^{R_1 \times R_2 \times R_3}\times \RR^{I_1\times R_1} \mapsto \RR^{I_1\times R_2\times R_3}$ as
$
\calF \times_1 \bA_1 = \big[ \sum_{r_1=1}^{R_1} a_{i_1 r_1} f_{r_1 r_2 r_3} \big]_{i_1 \in [I_1], r_2 \in [R_2], r_3 \in [R_3]}
$.
In a similar fashion, we can define fibers, mode matricization, and mode product for mode-2 and mode-3, respectively.

The widely used Tucker ranks (or multilinear ranks) of a tensor $\calS$ is defined by the triplet $\rank(\calS) := (R_1, R_2, R_3)$ where 
$
R_m=  \rank(\calM_j(\calS))
$
for modes $m=1,2,3$.
The Tucker rank $(R_1, R_2, R_3)$ is closely associated with the Tucker decomposition.
If a tensor $\calS$ has an exact tensor rank $(R_1, R_2, R_3)$, then there exists a core tensor $\calF \in \RR^{R_1 \times R_2 \times R_3}$ such that $\calS$ has a {\em Tucker decomposition} $\calS = \calF \times_1 \bA_1 \times_2 \bA_2 \times_3 \bA_3$ where $\bA_m\in \RR^{I_m \times R_m}$, $m\in[3]$, are orthonormal matrices of the left singular vectors of $\calM_m(\calS)$ respectively.

Given a tensor observation $\calY  \in \RR^{I_1 \times I_2 \times I_3}$, a {\em tensor factor model} assumes that
\begin{equation} \label{eqn:tensor-factor-model}
\calY = \calS + \calE = \calF \times_1 \bA_1 \times_2 \bA_2 \times_3 \bA_3 + \calE,
\end{equation}
where the \textit{latent tensor factor} $\calF$ is of dimension $R_1 \times R_2 \times R_3$, the \textit{loading}  matrices $\bA_m \in \RR^{I_m \times R_m}$ are unknown deterministic parameters, and $\calE$ is the noise tensor.
The low-rank structure is captured by the assumption of $R_m \ll I_m$ along the $m$-th mode.
Model \eqref{eqn:tensor-factor-model} encompasses the vector and the matrix factor models as special sub-cases:
the vector factor model \citep{fan2020robust} corresponds to the special case of $\calY=\bA_1\calF + \calE$ where $\calY$, $\calE\in\RR^{I_1}$ and $\calF\in\RR^{R_1}$ are all vectors (i.e. 1st-order tensor).
The matrix factor model \citep{wang2019factor,chen2019constrained,chen2020statistical} corresponds to the special case of $\calY=\bA_1\calF\bA_2^\top + \calE$ where $\calY$, $\calE\in\RR^{I_1\times I_2}$ and $\calF\in\RR^{R_1\times R_2}$ are all matrices (i.e. 2nd-order tensors).


All the components on the right hand side of model \eqref{eqn:tensor-factor-model} are not directly observable, thus the tuples 
$\paran{\calF\times_1\bH_1^{-1}\times_2\bH_2^{-1}\times_3\bH_3^{-1}, \bA_1\bH_1, \bA_2\bH_2, \bA_3\bH_3}$
and 
$\paran{\calF, \bA_1, \bA_2, \bA_3}$ are indistinguishable for any invertible matrix $\bH_m\in\RR^{R_m\times R_m}$, $m\in[3]$.
This is a common issue with latent models since they can only be identified up to the columns space of $\bA_m$ \citep{bai2003inferential,zhang2018tensor,fan2020robust}. 
To identify one representative matrix of the column space $\bA_m$, we restrict our solution to the one that satisfies Assumption \ref{assum:tfm-ic-1}.
Lemma \ref{lem:tfm-ic} confirms the validity of Assumption~\ref{assum:tfm-ic-1} as an identification condition for model \eqref{eqn:tensor-factor-model}.

\begin{assump}[Tensor Factor Model Identification Condition]~ \label{assum:tfm-ic-1}
    We restrict our estimation targets to the loading matrices and core tensor that satisfy 
    (i) $\bA_m^\top\bA_m/I_m =\bI_{R_m}$  for all $m\in[M]$ where $\bI_{R_m}$ is an $R_m\times R_m$ identity matrix; and
    (ii) $\mathcal M_m(\mathcal F)\mathcal M_m(\mathcal F)^\top$ is a diagonal matrix with non-zero decreasing singular values for all $m$.
\end{assump}

\begin{lemma}\label{lem:tfm-ic}
    Given an $\calS\in\RR^{I_1\times\cdots\times I_M}$ with Tucker ranks $(R_1, \cdots, R_M)$ and $\calM_m(\calS)\calM_m(\calS)^{\top}$ having distinct non-zero  singular values for all $m$, then there exist unique\footnote{Note that uniqueness is up to column-wise signs of $\bA_m$'s.} $\bA_1,\cdots, \bA_M$ and $\calF$ satisfying Assumption~\ref{assum:tfm-ic-1} so that $\calS=\calF\times_1 \bA_1\times_2\cdots\times_M \bA_M$.
\end{lemma}


Model \eqref{eqn:tensor-factor-model} can be estimated by solving the optimization program
\begin{equation} \label{eqn:tucker-opt}
\min_{\calF, \bA_1,\bA_2,\bA_3}\norm{\calY-\calF\times_1 \bA_1\times_2\bA_2\times_3 \bA_3}^2_F,
\end{equation}
under the constraints in Assumption \ref{assum:tfm-ic-1}.
It is highly non-convex and computationally NP-hard. 
The higher order orthogonal iteration (HOOI) algorithm \citep{de2000best} solves \eqref{eqn:tucker-opt} by alternating minimization along the direction of $\bA_m$. 
Given an initial guess of $\{\hat \bA_m\}_{m\geq 2}$, the algorithm update $\hat \bA_1$ to be the maximizing value  $\hat\bA_1=\sqrt{I_1}\cdot {\rm SVD}_{R_1}\big(\calM_1(\calY)(\hat\bA_2\otimes \hat \bA_3 )\big)$ where ${\rm SVD}_r(\cdot)$ returns top-$r$ left singular vectors of a given matrix. 
Then, the algorithm proceeds to iteratively updating $\hat \bA_m$ while fixing the other $\hat\bA_{j}, j \neq m$ until some stopping criterion is satisfied. 
The performance of HOOI usually relies on the initial input of $\{\hat \bA_m\}_{m\in [M]}$.

One way to measure the importance of each factor dimension along a mode is through the mode-wise percentage explained variance. 
Suppose we are interested in the relative importance of mode-$1$ factors, the total variance along mode-$1$ can be calculated by 
$\sigma_1^2 = {\rm Tr}\paran{\calM_1(\calY)\calM_1(\calY)^\top/(I_2I_3)}$ and variances of the $R_1$ factors of mode-$1$ are the diagonal elements in the covariance matrix
$\bSigma_{F,1} = \calM_1(\calF)\calM_1(\calF)^\top/(R_2R_3)$. 
Then the mode-$1$ percentage explained variances for each of the $R_1$ factors corresponds to each element in ${\rm diag}\paran{\bSigma_{F,1}}/\sigma_1^2$, where ${\rm diag}(\cdot)$ extracts $R_1$ diagonal elements from matrix $\bSigma_{F,1}$. 

\subsection{Semiparametric tensor factor model}\label{sec:stefa}

We now generalize the classic tensor factor model to integrate mode-wise auxiliary covariates.
For any $i_1\in[I_1]$, let $\bx_{1, i_1} = [x_{1, i_1 1}, \cdots, x_{1, i_1 D_1}]^\top$ be a $D_1$-dimensional vector of covariates associated with the $i_1$-th entry along mode 1.
We assume that the mode-1 loading coefficient $a_{1,i_1 r_1}$ can be (partially) explained by $\bx_{1, i_1}$ such that
\begin{equation*}
a_{1, i_1 r_1} = g_{1, r_1}\paran{\bx_{1, i_1}} + \gamma_{1, i_1 r_1}, \quad i_1 \in [I_1], r_1 \in [R_1],
\end{equation*}
where $g_{1, r_1}: \RR^{D_1}\mapsto\RR$ is a function and $\gamma_{1, i_1 r_1}$ is the part that \textit{cannot} be explained by the covariates.
Under this assumption, the entries in the $i_1$-th mode-1 slice, $i_1\in[I_1]$, can be written as
\begin{equation} \label{eqn:tucker-semi-entry}
y_{i_1 i_2 i_3} = \sum_{r_1 = 1}^{R_1} \sum_{r_2 = 1}^{R_2} \sum_{r_3 = 1}^{R_3} \paran{g_{1,r_1}\paran{\bx_{1,i_1}} + \gamma_{1,i_1 r_1}} a_{2, i_2 r_2} a_{3, i_3 r_3} f_{r_1 r_2 r_3} + \varepsilon_{i_1 i_2 i_3},
\end{equation}
for all $i_2 \in [I_2]$ and $i_3 \in [I_3]$.
Let $\bX_1$ be a $I_1 \times D_1$ matrix taking $\bx_{1,i_1}^\top$ as rows, $\bG_1 \paran{\bX_1}$ be the $I_1 \times R_1$ matrix with its $i_1$-th row being $[g_{1,1}\paran{\bx_{1,i_1}},\cdots, g_{1,R_1}(\bx_{1,i_1})]$, and $\bGamma_1$ be the $I_1 \times R_1$ matrix of $[\gamma_{1, i_1 r_1}]$, we can write compactly $\bA_1=\bG_1(\bX_1)+\bGamma_1$ and 
\begin{equation} \label{eqn:tucker-covari-1}
\calY = \calF \times_1 \paran{\bG_1 \paran{\bX_1} + \bGamma_1} \times_2 \bA_2 \times_3 \bA_3 + \calE.
\end{equation}
This semi-parametric configuration is easily extendable to all modes of $\calY$.  
If any mode-$m$ loading entries $a_{m, i_m r_m}$ can be partially explained by a $D_m$-dimensional vector $\bx_{m, i_m}$, i.e. $a_{m, i_m r_m} = g_{m, r_m}\paran{\bx_{m, i_m}} + \gamma_{m, i_m r_m}$, then we have
\begin{equation} \label{eqn:tucker-covari-3}
\calY = \calF \times_1 \paran{\bG_1 \paran{\bX_1} + \bGamma_1} \times_2 \paran{\bG_2 \paran{\bX_2} + \bGamma_2} \times_3 \paran{\bG_3 \paran{\bX_3} + \bGamma_3} + \calE,
\end{equation}
where $\bX_m$ is a $I_m \times D_m$ matrix taking $\bx_{m,i_m}^\top$ as rows, $\bG_m\paran{\bX_m}$ be the $I_m \times R_m$ matrix with its $i_m$-th row being $[g_{m,1}\paran{\bx_{m,i_m}},\cdots, g_{m,R_m}(\bx_{m,i_m})]$, and $\bGamma_m$ be the $I_m \times R_m$ matrix of $[\gamma_{m, i_m r_m}]$.
We refer to \eqref{eqn:tucker-covari-3} as the Semiparametric TEnsor FActor (STEFA) Model.
An an illustration of model \eqref{eqn:tucker-covari-1} is presented in Figure \ref{fig:stefa-tensor}. 
\begin{figure}[htpb!]
    \centering
    \includegraphics[width=.8\textwidth]{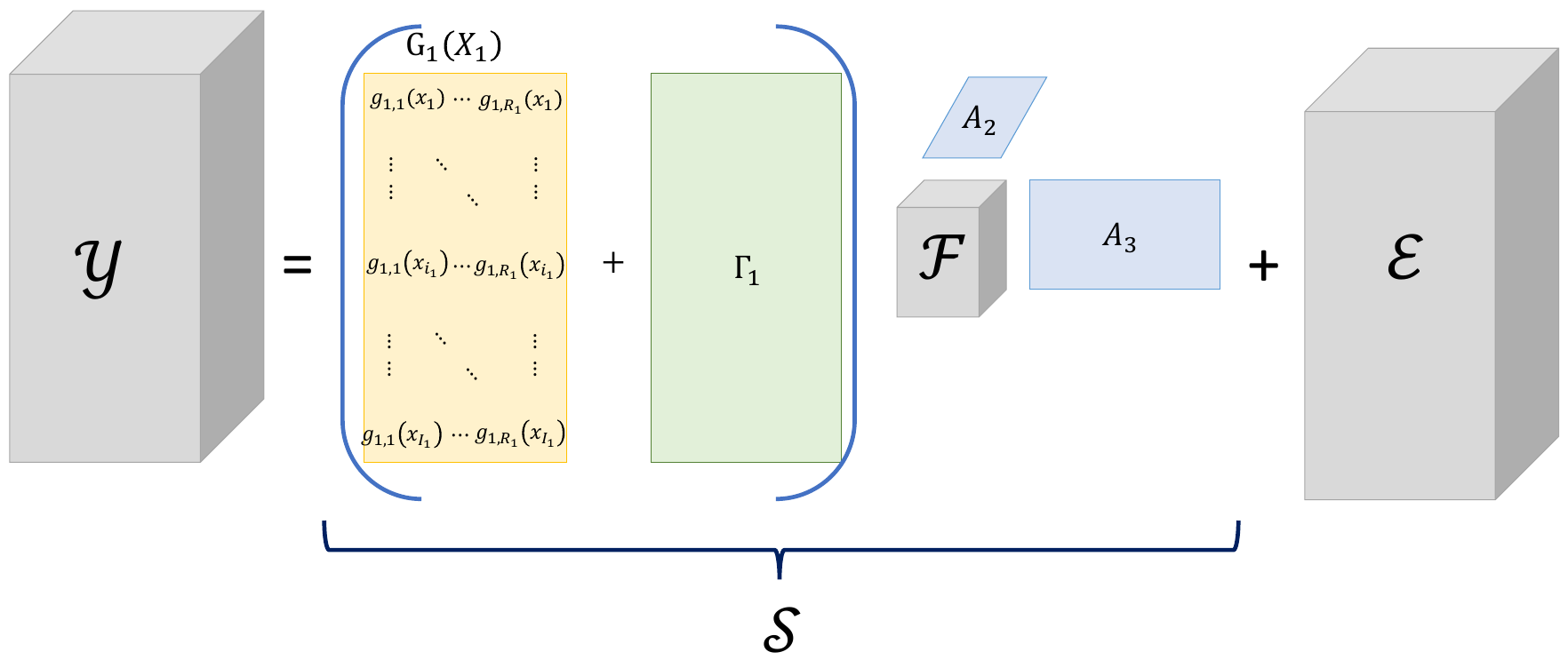}
    \caption{An illustration of the STEFA model \eqref{eqn:tucker-covari-1}.}
    \label{fig:stefa-tensor}
\end{figure}
When mode $m$ has no covariates, we take $\bG_m(\bX_m)={\bf 0}$.
If, additionally, mode $m$ has no factor structure, we take $\bA_m = \bI_{R_m}$ -- the identity matrix.
If all modes have no covariates, then STEFA reduces to the classical tensor factor model (\ref{eqn:tensor-factor-model}).
STEFA is a generalization of the  semi-parametric vector factor model \citep{fan2016projected} to the tensor data. 
But it is more complex in computation and theoretical analysis. 

\begin{remarkx}[Multivariate functional SVD]
    In the field of functional data analysis, 
    researchers have studied multidimensional functional SVD \citep{silverman1996smoothed,huang2009analysis} and functional PCA \citep{zhou2014principal,wang2017regularized}.  
    Specifically, the two-way functional SVD views each entry $y_{i_1 i_2}$ of the data matrix $\bY \in \RR^{I_1 \times I_2}$ as the evaluation of an underlying function $y(\cdot, \cdot)$ on a rectangular grid of sampling pints $\bx_{1,i_1}$ and $\bx_{2,i_2}$, that is,  $y_{i_1, i_2} := y(\bx_{1, i_1}, \bx_{2,i_2}) := \sum_{r=1}^{R} \sigma_r g_{1,r}(\bx_{1,i_1}) g_{1,r}(\bx_{2,i_2})$. 
    
    Let $\bG_1(\bX_1) \in \RR^{I_1 \times R}$ be the matrix that contains $g_{1,r}(\bx_{1, i_1})$ as its $(i_1, r)$-th element, $\bG_2(\bX_2) \in \RR^{I_2 \times R}$ be the matrix that contains $g_{2,r}(\bx_{2, i_2})$ as its $(i_2, r)$-th element, and $\bD$ represent the diagonal matrix $\diag (\sigma_1, \cdots, \sigma_R)$. 
    Under the functional SVD assumption, the data matrix has the following low-rank structure: 
     \begin{equation}   \label{eqn:func-svd}
        \bY = \bG_1(\bX_1) \, \bD \, \bG_2(\bX_2)^\top = \bD \times_1  \bG_1(\bX_1) \times \bG_2(\bX_2), 
    \end{equation}
    which is equivalent to a special case of the STEFA model where $M = 2$, core tensor $\calF\in \RR^{R\times R}$ is diagonal, $\bGamma_1 \equiv \bzero$ and $\bGamma_2 \equiv \bzero$. 
   
    The estimation method for function SVD are mostly based on regularized SVD which imposes the smoothness constraint on columns of $\bG_1(\bX_1)$ and $\bG_2(\bX_2)$. 
    For the STEFA model, we do not impose such constraints and our projection-based algorithm also estimate the covariate independent component $\bGamma_m$ that cannot be explained by the covariate. 
    
     In fact, model \eqref{eqn:func-svd} can be extended to the higher-order setting with $M \ge 3$, which can be viewed as a functional CP tensor decomposition \citep{kolda2009tensor} and is an interesting topic for future research. 
\end{remarkx}

\begin{remarkx}[Tensor response regression]
    The STEFA model is related to a list of tensor response regression models \citep{raskutti2019convex} with a low-rank coefficient tensor.
    Notably, \cite{sun2017store} and \cite{zhou2021partially} consider a model where response tensors $\calY_t \in \RR^{I_1 \times \cdots \times I_{M-1}}$ are related to a $D_M$-dimensional vector of covariate $\bx_t$ through 
    \begin{equation} \label{eqn:trr}
        \calY_t = \calB \times_{M} \bx_t  + \calE_t, 
    \end{equation} 
    where $\calB$ is a $I_1 \times \cdots \times I_{M-1} \times D_M$ unknown parameter tensor of interest, and the noise tensor $\calE_t$ has i.i.d.~standard Gaussian entries. 
    Model \eqref{eqn:trr} can be rearranged to a similar form as the STEFA model. 
    Specifically, we stack the tensor response $\calY_t$ along a new $M$-th order and get a new $(\bI_1 \times \cdots \times \bI_{M-1} \times T)$ tensor $\calY$. 
    We also stack the vector covariate $\bx_t$ together and get a new $(T \times D_M)$ matrix $\bX_M$. 
    Then, model \eqref{eqn:trr} can be rewritten as
    \begin{equation} \label{eqn:trr-mat}
        \calY = \calB \times_{M} \bX_M  + \calE_t. 
    \end{equation} 
    For high-dimensional data, the sparse or low-rank structure is assumed on the coefficient tensor $\calB$ to facilitate estimation. 
    For example, \cite{sun2017store} and \cite{zhou2021partially} assume that $\calB$ admits a rank-$R$ CP decomposition structure. 
    Alternatively, $\calB$ can be assumed to admit a rank-$(R_1, \cdots, R_M)$ Tucker decomposition structure \citep{raskutti2019convex} denoted by $\calB := \calF \times_1 \bA_1 \times_2 \cdots \times_{M-1} \bA_{M-1} \times_{M} \bB_M$, where $\calF$ is a $R_1 \times \cdots \times R_M$ tensor, $\bA_m$ are $I_m \times R_m$ matrices for $m\in[M-1]$ and $\bB_M$ is a $D_M \times R_M$ matrix.
    Under such Tucker low-rankness, model \eqref{eqn:trr-mat} can be further rewritten as
     \begin{equation*} 
        \calY = \calF \times_1 \bA_1 \times_2 \cdots \times_{M-1} \bA_{M-1} \times_{M} (\bX_M \bB_M)  + \calE_t. 
    \end{equation*} 
    which has the same form as a restricted STEFA model with $\bA_M = \bX_M \bB_M$ being exact linear and non-existence of the covariate-independent component $\bGamma_M$. 
\end{remarkx}

\begin{remarkx}[Multiple-mode-covariate tensor regression]
The multiple-mode-covariate (MMC) tensor regression  \citep{hu2021generalized} with identity link function writes
    \begin{equation} \label{eqn:xu-2019}
        \calY = \calB \times_1 \bX_1 \times_2 \bX_2 \times_3 \bX_3 + \calE,
    \end{equation} 
    where $\bX_m$ is the observable $I_m \times D_m$ covariate matrix and $\calB$ is a low-rank regression coefficient tensor. 
    The MMC tensor regression model is a parametric model while the the STEFA model is semi-parametric. 
    The STEFA model is to the MMC tensor regression as the projected PCA is to the reduce-rank regression. 
    
    If we wish to make the parametric assumption that the true loading function $g_{m, r_m}(\cdot)$ is linear and no covariate-independent component, i.e. $\gamma_{m,i_m,r_m} = 0$ in \eqref{eqn:tucker-semi-entry}, the STEFA model can be rewritten in the same form as \eqref{eqn:xu-2019}. 
    Specifically, the loading can be explicitly written as $\bA_m = \bX_m\bB_m$ 
    where $\bB_m \in \RR^{D_m \times R_m}$. 
    The STEFA model can be rewritten as
     \begin{equation}  \label{eqn:linear-loading}
         \calY = \calF \times_1 (\bX_1\bB_1) \times_2 (\bX_2\bB_2) \times_3 (\bX_3\bB_3) + \calE
          = \calB \times_1 \bX_1 \times_2 \bX_2 \times_3 \bX_3 + \calE, 
     \end{equation}
     where $\calB = \calF \times_1 \bB_1 \times_2 \bB_2 \times \bB_3$.
     Otherwise, the STEFA model is very different from the MCC tensor regression since it allows any smooth function $g_{m, r_m}(\cdot)$ and the existence of the covariate-independent component $\gamma_{m,i_m,r_m}$. 
     Generally, the advantages of the STEFA model are its non-parametric modeling on the covariates as well as its weak technical assumptions.
\end{remarkx}

\subsubsection{Identifiability conditions for STEFA}
Similar to the tensor factor model \eqref{eqn:tensor-factor-model}, the identifiability is also an issue for STEFA.
Note that the factor loading $\bA_m$ in STEFA consists of two components $\bG_m(\bX_m)$ and $\bGamma_m$.
A naive generalization of Assumption~\ref{assum:tfm-ic-1} requires that
\begin{equation*}
 \bI_{R_m}=\bA_m^{\top}\bA_m=(\bG_m(\bX_m)+\bGamma_m)^{\top}(\bG_m(\bX_m)+\bGamma_m)=\bG_m(\bX_m)^{\top}\bG_m(\bX_m)+\bGamma_m^{\top}\bGamma_m,
\end{equation*}
where we assume that $\bGamma_m^{\top}\bG_m(\bX_m)={\bf 0}$. While the above identification is theoretically valid, such a condition imposes a constraint jointly for both the parametric and non-parametric components and introduces unnecessary difficulty into the estimating procedures.
Instead, we propose the following identification condition for STEFA.

\begin{assump}[STEFA Identification Condition]~ \label{assum:stefa-ic-1}
    We restrict our estimation targets to the loading matrices and core tensor that satisfy 
    \begin{enumerate}[label=(\roman*)]
        \item $\bG_m^{\top}(\bX_m)\bG_m(\bX_m)/I_m=\bI_{R_m}$ and $\bG_m^{\top}(\bX_m)\bGamma_m=\bzero$ for all $m\in[M]$.
        \item $\mathcal M_m(\mathcal F)\mathcal M_m(\mathcal F)^\top$ is a diagonal matrix with non-zero decreasing singular values for all $m\in[M]$.
    \end{enumerate}
\end{assump}
Note that the identification condition $\bG_m^{\top}(\bX_m)\bG_m(\bX_m)/I_m=\bI_{R_m}$ can be replaced with $\bGamma_m^{\top}\bGamma_m/I_m=\bI_{R_m}$. We choose the first equation just for simplicity because our method starts with estimating the non-parametric component $\bG_m(\bX_m)$. However, if some mode $m$ has no covariate information, then we have to replace the identification condition with $\bGamma_m^{\top}\bGamma_m/I_m=\bI_{R_m}$.
Also note that $\bG_m \paran{\bX_m}$ is the $I_m \times R_m$ matrix of $[g_{m,r_m}\paran{\bx_{m,i_m}}]_{i_m,r_m}$, thus the identification condition is defined with respect to matrix $\bG_m \paran{\bX_m}$ with a fixed $I_m$, not on the functional form of $g_{m, r_m}\paran{\bx_{m,i_m}}$.
Alternatively, one can consider a functional version of identification conditions on $g_{m,r_m}\paran{\bx_{m,i_m}}$ defined on a Hilbert space consisting of all the square integrable functions.
But the intricate combination of functional space and tensor structure renders the problem even more difficult and thus will not be pursued here.

\section{Estimation} \label{sec:est}

In this section, we present a computationally efficient Iteratively Projected SVD (IP-SVD) algorithm to estimate the  STEFA model.
Given the identification condition (Assumption~\ref{assum:stefa-ic-1}), we start with estimating the non-parametric component $\bG_m(\bX_m)$.

\subsection{Sieve approximation and basis projection}
Our primary ingredient of estimating $\bG_m(\bX_m)$ is the sieve approximation which is a classical method in non-parametric statistics \citep{chen2007large}. At this moment, we assume that the latent dimensions $R_1$, $R_2$ and $R_3$ are known. In Section \ref{sec:estimate-tucker-ranks}, we will discuss a method to consistently estimate $R_1$, $R_2$ and $R_3$ when they are unknown.

Sieve approximation relies on a set of basis functions.
Take mode 1 for illustration.
We denote $\braces{\phi_{1,j_1}(\cdot)}_{j_1\in[J_1]}$ as a set of basis functions on $\{f:\RR^{D_1} \rightarrow \RR^{I_1}\}$,
which spans a complete  space for $\braces{g_{1,  r_1}\paran{\cdot}}_{r_1\in[R_1]}$ . 
Some widely-used basis functions are  B-spline, Fourier series, wavelets, and polynomial series \citep[Section 2.3]{chen2007large}. 
We let $\bPhi_1(\bX_1)$ be the $I_1 \times {J_1}$ matrix whose $(i_1, j_1)$-th element is $\phi_{1,j_1}(\bx_{1, i_1})$. 
We denote the $J_1 \times R_1$ matrix of sieve coefficients as $\bB_1 = \brackets{\bb_{1,1}, \cdots, \bb_{1,R_1}}$, and the $I_1 \times R_1$ residual matrix as $\bR_1(\bX_1)$, consisting of approximation errors.  
Then, in the matrix form, we have $\bG_1(\bX_1) = \bPhi_1(\bX_1)\bB_1 + \bR_1(\bX_1)$, where $\bPhi_1(\bX_1)$ can be constructed from covariates and $\bR_1(\bX_1)$ shall be small for a large enough $J_1$.
To this end, the factor loading $\bA_1$ can be written as $\bA_1=\bPhi_1(\bX_1)\bB_1+\bR_1(\bX_1)+\bGamma_1$ (illustrated in the big parentheses in Figure \ref{fig:stefa-approx}.)
\begin{figure}[htpb!]
    \centering
    \includegraphics[width=.8\textwidth]{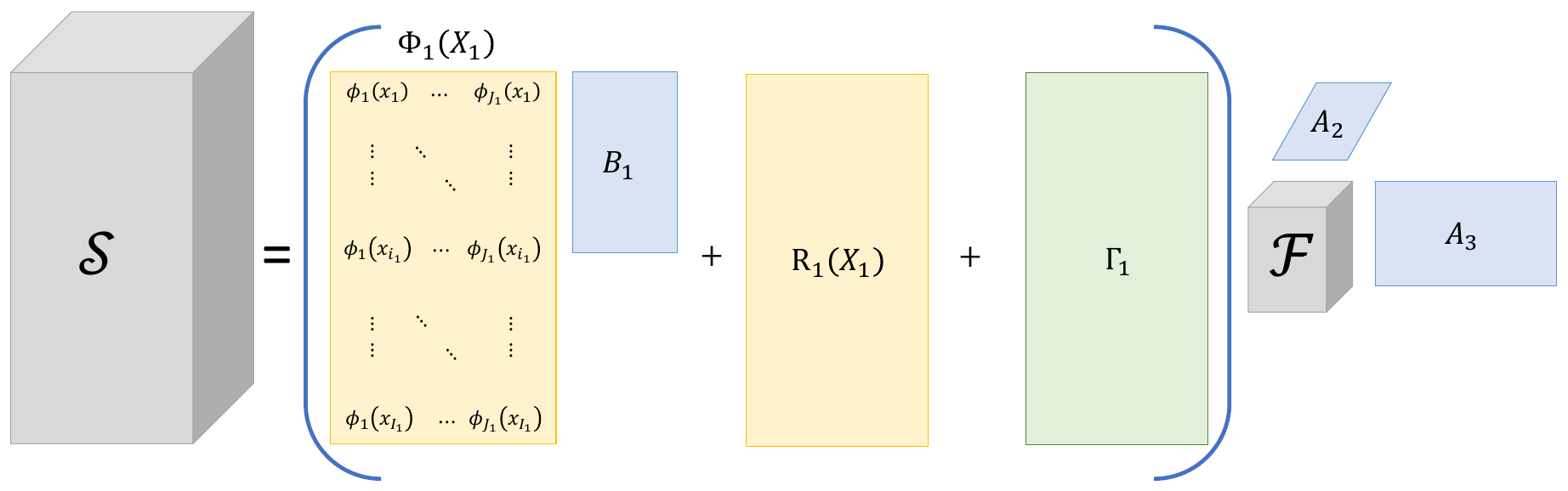}
    \caption{An illustration of sieve approximation in signal of the STEFA model. 
    The first loading matrices $\bA_1$ is decomposed into three part: the sieve approximation $\bPhi_1(\bX_1)\bB_1$, the sieve residual $\bR_1(\bX_1)$ and the covariate independent component $\bGamma_1$}
    \label{fig:stefa-approx}
\end{figure}
Generalizing to other modes $m \in [M]$, we can define similar terms and write that
\begin{equation}\label{eqn:G-reg}
\bG_m(\bX_m) = \bPhi_m(\bX_m)\bB_m + \bR_m(\bX_m).
\end{equation}
Then, a general STEFA can be re-formulated as
\begin{equation*}
\calY=\calF\times_1 \big( \bPhi_1(\bX_1)\bB_1 + \bR_1(\bX_1)+\bGamma_1\big)\times_2\cdots\times_M \big( \bPhi_M(\bX_M)\bB_M+ \bR_M(\bX_M)+\bGamma_M\big)+\calE.
\end{equation*}
In practice, to nonparametrically estimate $g_{m,r_m}(\bx_{m,i_m})$ without suffering from the curse of dimensionality when the dimension of $\bx_{m,i_m}$ is large, we can assume $g_{m, r_m}(\bx_{m,i_m})$ to be structured.   
A popular example of this kind is the additive model: for each $r_m \in [R_m]$, there are $D_m$ univariate functions $\braces{g_{m,r_m d_m}(\cdot)}_{d_m=1}^{D_m}$ such that
\begin{equation} \label{eqn:g-func-addants-vec}
g_{m,r_m}(\bx_{m,i_m}) = \sum_{d_m=1}^{D_m} g_{m,r_m d_m}(x_{m, i_m d_m}).
\end{equation}
Each one dimensional additive component $g_{m,r_m,d_m}(x_{i_m d_m})$ can be estimated without curse of dimensionality by the sieve approximation or other more complex functions.
Possible data-driven methods to estimate $J$'s are discussed in Appendix \ref{append:n-basis} in the supplemental material.

\subsection{Iteratively projected SVD}

We propose an iteratively projected SVD (IP-SVD) algorithm \footnote{A Python library of IP-SVD is available at \url{https://github.com/ElynnCC/STEFA-Code.git}. }
 to estimate the right hand side of the STEFA model  \eqref{eqn:tucker-covari-3}  from tensor $\calY$ and
matrices of covariate $\bX_m$ for $m\in [M]$.
For ease of notations, we write $\bG_m$ and $\bPhi_m$ instead of $\bG_m(\bX_m)$ and $\bPhi_m(\bX_m)$ and
define $\bP_m = \bPhi_m\cdot \left(\bPhi_m^\top\bPhi_m\right)^{-1}\bPhi_m^\top$ as the $I_m \times I_m$ projection matrix onto the sieve spaces spanned by the basis functions of $\bX_m$ of all $m\in[M]$.
Algorithm  \ref{alg-IP-SVD} summarizes the whole procedure. 
For ease of presentation, it is presented for the third order tensor or $M = 3$. 
But it is representative for the general $M$ setting. 
The outputs are estimators of the tensor factor $\hat\calF$, covariate-relevant loadings $\hat\bG_m$, sieve coefficient matrices $\hat\bB_m$, full loading matrices $\hat\bA_m$ and covariate-independent loadings $\hat\bGamma_m$ for all $m \in [M]$.

The algorithm is divided into two major blocks. 
The first block consists of the first four steps, namely projected spectral initialization, projected power iteration, projection estimate for the tensor factor and orthogonal calibration.
Together, they estimate $\hat\calF$ and $\hat\bG_m$ through an iterative procedure. 
The first step of projected spectral initialization utilizes the fact that the column space of each loading $\bG_m$ is mainly a subspace of the basis projection $\bP_m$ by sieve approximation. 
It obtains a preliminary estimator for $\bG_m$ for each $m\in[3]$ via sieve projection, matricization and singular value decomposition (SVD), specified in equation \eqref{eqn:initialization}. 
This step, in spirit, is similar to the projected PCA in \citep{fan2016projected}.
This initial estimator $\tilde \bG_m^{(0)}$ acts as a good starting point, but is {\em sub-optimal} in general.
In the second step of projected power iteration, we apply power iterations to refine the initialization.
Given rudimentary estimators $\tilde \bG_2^{(t-1)}$ and $\tilde \bG_3^{(t-1)}$, we further denoise $\calY$ by the mode-2 and 3 projections: $\calY\times_2 \tilde \bG_2^{(t-1)\top} \times_3 \tilde \bG_3^{(t-1)\top}$.
This refinement can significantly reduce the amplitude of noise while reserving the mode-1 singular subspace.
Iteratively for $t=1,\cdots,t_{\max}$, we obtain an updated estimator $\tilde\bG_m^{(t)} $ for each $m\in[3]$ according to \eqref{eqn:projected-power-iter}. 
This projected power iteration algorithm is a modification of the classical HOOI algorithm \citep{de2000best}.
The additional projection $\bP_m$ restricts the solution to be a linear function of sieve basis functions.
Empirically, the projected version of HOOI in this step converges very fast within a few iterations.
The output of this step is the final estimators $\tilde \bG_m = \tilde \bG_m^{(t_{\max})}$ for $m\in [M]$.
In the third step, $\hat\calF$ is estimated via least squares, which amounts to the projection in equation \eqref{eqn:proj-F}. 
The fourth step fixes a numerical solution of tensor factor and loadings that satisfy Assumption~\ref{assum:stefa-ic-1} by orthogonal calibration. 
The orthogonal rotation matrices are calculated by \eqref{eqn:rotation} and the ultimate estimator is given by equation \eqref{eqn:FandGm-final}. 

The second main block of the algorithm takes care of the estimation of the sieve coefficient matrices $\hat\bB_m$, full loading matrices $\hat \bA_m$, and the covariate-independent loading matrices $\hat\bGamma_m$ for $m \in [M]$ in the fifth and sixth steps, respectively. 
The sieve coefficients $\bB_m$ is useful for prediction on new covariates.
After obtaining $\hat\bG_m(\bX_m)$, sieve coefficients can be estimated following the standard sieve approximation procedure.
Indeed, we estimate $\hat\bB_m$ by equation \eqref{eqn:coeff-B}. 
Then the mode-m loading function $\bg_m(\bx)=(g_{m,1}(\bx),\cdots, g_{m,R_m}(\bx))$ can be estimated by $\hat\bg(\bx)=\bPhi(\bx)\hat\bB_m$ for any $\bx$ in the domain of mode-$m$ covariates.
Further, with the estimated $\hat\bG_m$ and tensor factor $\hat\calF$, we estimate $\hat \bA_m$ by regression in \eqref{eqn:full-loading} and $ \hat\bGamma_m$ by projecting $\hat \bA_m$ on the orthogonal column space of $\bPhi_m$ in \eqref{eqn:resid-loading}. 

The above procedure only involves matrix product and matrix SVD, which computes fast. 
Without loss of generality, assume $I_1\geq \cdots\geq I_M$ and $R_1\geq \cdots\geq R_M$. The major computation load comes from the first three three steps. 
Specifically, the projected spectral initialization in the first step requires $O(I_1^2I_2\cdots I_M)$ flops; each iteration in the second step requires $O(I_1\cdots I_M R_1\cdots R_{M-1})$ flops; and the third step requires $O(I_1\cdots I_M R_1\cdots R_M)$ flops.  
Distributed or parallel computing can be employed to speed up the computation \citep{de2014distributed,baskaran2017memory}. 

\begin{algorithm}[htpb!]
    \SetKwInOut{Input}{Input}
    \SetKwInOut{Output}{Output}
    \Input{Tensor $\calY \in \RR^{I_1 \times I_2 \times I_3}$, 
        matrices of covariate $\bX_m$ whose rows are $\bx_{m, i_m}$,
        ranks $R_m$, and
        sets of basis functions $\braces{\phi_{m,j_m}(\cdot)}_{j_m\in[J_m]}$ for $m \in [3]$.  
        }
    \Output{$\hat\calF$, $\hat\bG_m$, $\hat\bB_m$, $\hat\bA_m$ and $\hat\bGamma_m$ for $m \in [3]$. }
    
    For each $m \in [3]$, calculate the projection matrices
    $\bP_m = \bPhi_m\cdot \left(\bPhi_m^\top\bPhi_m\right)^{-1}\bPhi_m^\top$,
    where $\bPhi_m(\bX_m)$ is the $I_m \times {J_m}$ matrix whose $(i_m, j_m)$-th element is $\phi_{m,j_m}(\bx_{m, i_m})$. 
    
    \tcc{1st step: Projected spectral initialization.}
    Let $t = 0$ and calculate 
    \begin{equation} \label{eqn:initialization}
          \tilde\calY = \calY \times_1 \bP_1 \times_2 \bP_2 \times_3 \bP_3 \quad {\rm and}\quad
          \tilde \bG_m^{(0)} = \sqrt{I_m}\cdot{\rm SVD}_{R_m}(\calM_m(\tilde\calY)).
    \end{equation}
    
    \tcc{2nd step: Projected power iterations.}
    \For{$t=1,\dots, t_{max}$}{
        Calculate 
         \begin{equation} \label{eqn:projected-power-iter}
            \begin{split}
                \tilde\bG_1^{(t)} &= \sqrt{I_1}\cdot{\rm SVD}_{R_1}\paran{ \bP_1\cdot\calM_1\paran{\calY \times_2 \tilde\bG_2^{(t-1)\top} \times_3 \tilde\bG_3^{(t-1)\top}}},\\
                \tilde \bG_2^{(t)} &=\sqrt{I_2}\cdot {\rm SVD}_{R_2}\paran{\bP_2\cdot\calM_2\paran{\calY\times_1 \tilde\bG_1^{(t)\top} \times_3 \tilde\bG_3^{(t-1)\top}} },\\
                \tilde \bG_3^{(t)} &= \sqrt{I_3}\cdot{\rm SVD}_{R_3}\paran{\bP_3\cdot\calM_3\paran{\calY\times_1 \tilde \bG_1^{(t)\top} \times_2 \tilde \bG_2^{(t)\top}}}.
            \end{split}
        \end{equation}
    }  
   
   \tcc{3rd step: Projection estimate for tensor factor.}   
   Calculate , with $ \tilde \bG_j = \tilde \bG_j ^{t_{max}} (j=1, 2, 3)$,
    \begin{equation} \label{eqn:proj-F}
        \tilde \calF = (I_1I_2I_3)^{-1}\cdot \calY \times_1 \tilde \bG_1^\top \times_2 \tilde \bG_2^\top \times_3 \tilde \bG_3^\top.
    \end{equation}   

   \tcc{4th step: Orthogonal calibration.}
   Calculate 
   \begin{equation} \label{eqn:rotation}
       \hat \bO_m={\rm SVD}_{R_m}\big(\calM_m(\tilde\calF)\calM_m(\tilde\calF)^{\top}\big), \text{for each } m\in [3].
   \end{equation}
   
   Calculate the ultimate estimator by
   \begin{equation} \label{eqn:FandGm-final}
       \hat\calF=\tilde\calF\times_1\hat \bO_1^{\top}\times_2 \hat\bO_2^{\top} \times_3 \hat \bO_3^{\top}\quad {\rm and}\quad 
       \hat\bG_m=\tilde\bG_m\hat\bO_m, \text{for each } m\in [3].
   \end{equation}
 
 \tcc{5th step: Covariate sieve coefficient matrices.}
 Calculate 
 \begin{equation} \label{eqn:coeff-B}
     \hat\bB_m = \brackets{\bPhi_m^\top\bPhi_m}^{-1}\bPhi_m^\top \hat\bG_m.
 \end{equation}

  \tcc{6th step: Full and covariate-independent loading matrices.}
  Calculate 
  \begin{equation*}
      \hat\bQ_m = \calM_m\paran{\hat\calF \times_{j\neq m} (\hat\bG_j/\sqrt{I_j})},\quad \tilde\calY_m= \calY\times_{j\neq m} \bP_j. 
  \end{equation*}

  Calculate the full loading matrices by
  \begin{equation} \label{eqn:full-loading}
  \hat\bA_m = \calM_m\paran{\tilde\calY_m}\hat\bQ_m^\top \paran{\hat\bQ_m\hat\bQ_m^\top}^{-1}/\sqrt{I_m^-}
  \end{equation}

  Calculate the covariate-independent loading matrices by
  \begin{equation} \label{eqn:resid-loading}
      \hat\bGamma_m 
      = (\bI-\bP_m)\widehat\bA_m
  \end{equation}
   \vspace{-3ex}
   \caption{ Iteratively Projected SVD (IP-SVD) } 
   \label{alg-IP-SVD}
\end{algorithm}

\subsection{Estimating the Tucker ranks} \label{sec:estimate-tucker-ranks}

In this section, we discuss the problem of estimating the Tucker ranks $(R_1, R_2, R_3)$ when they are unknown.
Given $\calY = \calF\times_1\bA_1\times_2\bA_2\times_3\bA_3 + \calE$ with the identifiable condition in Assumption~\ref{assum:tfm-ic-1}, the mode-$1$ matricization of $\calY$ is
\begin{equation}
\calM_1(\calY) = \bA_1\calM_1(\calF)(\bA_2\otimes \bA_3)^{\top} + \calM_1(\calE).\label{eq:matriciation-factor-model}
\end{equation}
The first term in \eqref{eq:matriciation-factor-model} is of rank $R_1$ when $\bA_1\in\mathbb R^{I_1\times R_1}$ and $R_1 \ll I_1$. The second term in \eqref{eq:matriciation-factor-model} is a $I_1\times I_2I_3$ noise matrix with i.i.d entries. Viewing $\calM_1(\calF)(\bA_2\otimes \bA_3)^{\top}$ as a whole, equation~\eqref{eq:matriciation-factor-model} is a factor model and $R_1$ is the corresponding unknown number of factors to be determined.
There exist many approaches in consistently estimating the number of factors from the model \eqref{eq:matriciation-factor-model}. In particular, \cite{lam2012factor, ahn2013eigenvalue, fan2016projected} proposed to estimate number of factors by selecting the largest eigenvalue ratio of $\calM_1(\calY)[\calM_1(\calY)]^\top$.
Due to the noise term in \eqref{eq:matriciation-factor-model}, \cite{fan2016projected} pointed out it is better to work on the projected version of $\calM_1(\calY)$.

Suppose $\tilde \calY = \calY\times_1\bP_1\times_2\bP_2\times_3\bP_3$ is the projected version of $\calY$. Then with Assumption~\ref{assum:stefa-ic-1}, $\mathbb E[\calM_1(\tilde \calY)[\calM_1(\tilde\calY)]^\top] = I_2I_3\bG_1\calM_1(\calF)[\calM_1(\calF)]^\top\bG_1^\top + \mathbb E[\bP_1\calM_1(\calE)(\bP_2\otimes \bP_3)[\calM_1(\calE)]^\top\bP_1^\top]=I_2I_3\bG_1\calM_1(\calF)[\calM_1(\calF)]^\top\bG_1^\top + \sigma_\epsilon^2R_2R_3\bP_1\bone_{I_3\times I_3} \bP_1^\top$ has the same spectrum structure as $\mathbb E[\calM_1(\calY)[\calM_1(\calY)]^\top]$ but with a reduced noise term. Here $\sigma^2_\epsilon$ denotes the variance of the entries of $\calE$ and $\bone_{I_3\times I_3}$ is the $I_3\times I_3$ matrix with all entries equal to one.
Denote by $\lambda_k(\calM_m(\calY)[\calM_m(\calY)]^\top)$ the $k$-th largest eigenvalue of the mode-$m$ matricization of the projected tensor. The eigenvalue ratio estimator of $R_m$ is defined as
\begin{equation}\label{eq:hat_Rm}
\hat R_m = \argmax_{1\leqslant k\leqslant k_{max}}\dfrac{\lambda_k(\calM_m(\tilde\calY)[\calM_m(\tilde\calY)]^\top)}{\lambda_{k+1}(\calM_m(\tilde\calY)[\calM_m(\tilde\calY)]^\top)}
\end{equation}
where $k_{max}$ is an upper bound on the number of factors, such as the nearest integer of $\min\left\{I_m, \prod_{n\neq m}I_n\right\}/2$, say.

The theoretical foundation for this estimator is {\it partially}  provided in \cite{fan2016projected}. 
Specifically, {\it for each mode $m$}, as long as there exists an $\alpha\in(0, 1]$ such that all the $R_m$ eigenvalues of $\left(\prod_{n\neq m}I_n^{1-\alpha}\right)\calM_m(\calF)[\calM_m(\calF)]^\top$ are bounded between two positive constants $c_{min}$ and $c_{max}$. 
The consistency of $\hat R_m$ is provided,  in terms of $\PP[\hat R_m = R_m]\rightarrow 1$, under suitable conditions (e.g., sub-Gaussian noise and $J_m=o(I_m^{1/2})$ ). 
However, while $\hat R_m$ works reasonably well in simulation studies,  it may be {\it statistically sub-optimal} for STEFA because the multi-way tensor structure is under-exploited, i.e., the low-dimensional tensor-product structure of row space of $\calM_m(\tilde \calY)$ is ignored. 
A statistically more efficient approach is to also estimate $\hat R_1, \cdots, \hat R_M$ iteratively.
The idea is similar to the iterative procedure to estimate the loadings in the tensor factor models instead of the single-step estimation as in the vector factor models \citep{fan2016projected}. 

Specifically, we note that $\tilde\calY$ in \eqref{eq:hat_Rm} is a one-time projection onto the column space of $\bP_1,\cdots,\bP_M$. 
To make the estimation of $\hat R_1, \cdots, \hat R_M$ iterative, $\tilde\calY$ in \eqref{eq:hat_Rm} should be replaced by a projection of the tensor onto the column spaces of $\tilde{\bG}_1^{(t)},\cdots,\tilde{\bG}_{m-1}^{(t)}, \tilde{\bG}_{m+1}^{(t-1)}, \cdots, \tilde{\bG}_M^{(t-1)}$ when estimating $\hat R_{m}$ for $m=1,\cdots,M$.  
However, in this case, establishing the consistency theory jointly for all $\hat R_m$, i.e., $\PP[\cap_{m=1}^M \{\hat R_m=R_m\}]$ can be more challenging than that in the PCA setting \citep{fan2016projected},  due to the interplay between $R_m$'s and dependence among $\hat R_m$'s.
We leave the theoretical investigations of $\hat R_m$ for future work and suggest interested readers to refer to a very recent work \citep{han2022rank} on the rank determination for tensor factor model.

\subsection{Prediction} \label{sec:pred}

The STEFA model can be applied to predict unobserved outcomes from the available data.
We illustrate the procedure of prediction along the first mode under model \eqref{eqn:tucker-covari-1}.
Prediction along other modes can be done in a similar fashion.
The task here is to predict a new $I_1^{new}\times I_2\times I_3$ tensor $\calY^{new}$ with new covariate matrix $\bX_1^{new}$ whose rows are $I_1^{new}$ new covariate $\{\bx_{1, i_1}^{new} \}_{i_1\in[I_1^{new}]}$ along mode 1. 
Under the STEFA model \eqref{eqn:tucker-covari-1}, the tensor observation $\calY$ assumes the following structure
\begin{equation*} 
\calY = \underbrace{\calF \times_1 \bPhi_1(\bX_1)\bB_1 \times_2 \bA_2 \times_3 \bA_3}_{\text{sieve signal}} + \underbrace{\calF \times_1 \bLambda_1 \times_2 \bA_2 \times_3 \bA_3}_{\text{residual signal}}+ \calE,
\end{equation*}
where $\bPhi_1(\bX_1)\bB_1$ is the part explained by the sieve approximation of $\bX_1$ and $\bLambda_1 = \bR_1(\bX_1) + \bGamma_1$ contains the sieve residual and the orthogonal part.
In Section \ref{sec:est}, we obtain estimators $~\hat{\cdot}~$ for the unknowns on the right hand side.
Note that $\bLambda_1$ can be estimated as a whole whereas its component $\bR_1(\bX_1)$ and $\bGamma_1$ are not separable.
With new observation $\bX_1^{new}$, we estimate the sieve signal using
\[
\hat\calS^{new}_{sieve} = \hat\calF \times_1 \bPhi_1(\bX^{new})\,\hat\bB_1 \times_2 \hat\bA_2 \times_3 \hat\bA_3.
\]
For the residual part, we use the simple kernel smoothing over mode-1 using $\bX_1$ and $\bX_1^{new}$.
Specifically,
we have the residual signal estimator $\hat\calS_{resid} = \hat\calF\times_1 \hat\bLam_1 \times_2 \hat\bA_2 \times_3 \hat\bA_3$.
Define the kernel weight matrix $\bW \in \RR^{I_1^{new} \times I_1}$ with entry
\[
w_{ij} = \frac{ K_h(dist(\bx_{1,i}^{new},\bx_{1,j})) }{ \sum_{j = 1}^{I_1} K_h(dist(\bx_{1,i}^{new},\bx_{1,j})) }, \quad i \in [I_1^{new}] \text{ and } j \in [I_1].
\]
where $K_h(\cdot)$ is the kernel function, $dist(\cdot, \cdot)$ is a pre-defined distance function such as the Euclidean distance, and $\bx_{1,i}$ is the $i$-th row of $\bX_1$.
We estimate the new residual signal by
\begin{equation} \label{eqn:kernel-smoothing}
\hat\calS_{resid}^{new} = \hat\calS_{resid} \times_1 \bW,
\end{equation}
the derivation of which is given in Section B of the supplementary material. 
Finally, our prediction for new entries corresponding to new covariate matrix $\bX_1^{new}$ is given by
\begin{equation}  \label{eqn:pred-stefa}
\hat\calY^{new} = \hat\calS^{new}_{sieve} + \hat\calS_{resid}^{new}.
\end{equation}


\begin{remarkx}
    The identification condition Assumption \ref{assum:stefa-ic-1} is not restrictive in the sense that it is only used to help us separate the loadings and the factor, that is, fix a numerical solution corresponding to a specific linear transformation among multiple equivalent ones. 
    The signal part $\calS$ will not be affected by the specific linear transformation and thus the identification Assumption  \ref{assum:stefa-ic-1} will not affect the prediction.  
    Suppose the true decomposition of the signal part $\mathring{\calS} = \mathring{\calS}_{sieve} +  \mathring{\calS}_{resid}$ is 
   \begin{equation*}
   \calS_{sieve} = \mathring{\calF} \times_1 \bPhi_1(\bX_1) \mathring{\bB}_1 \times_2 \mathring{\bA}_2 \times_3 \mathring{\bA}_3, \quad\text{and}\quad
   \calS_{resid} = \mathring{\calF} \times_1 \mathring{\bLam}_1 \times_2 \mathring{\bA}_2 \times_3 \mathring{\bA}_3, 
    \end{equation*}
   where $\mathring{\calF}$, $\mathring{\bB}_1$, $\mathring{\bA}_2$, and $ \mathring{\bA}_3$ are the true components. 
   Our estimation targets are restricted by Assumption \ref{assum:tfm-ic-1} and \ref{assum:stefa-ic-1} on observed discrete rows of $\bX_1$ and they are linear transformations of their true counterparts. 
   That is, $\bB_1 := \mathring{\bB}_1 \bH_1$, $\bA_2 := \mathring{\bA}_2 \bH_2$,  $\bA_3 = \mathring{\bA}_3 \bH_3$, and $\calF = \mathring{\calF} \times_1  \bH_1^{-1}  \times_2 \bH_2^{-1} \times_3 \bH_3^{-1}$ for some invertible matrices $\bH_1$, $\bH_2$, and $\bH_3$. 
  Algorithm \ref{alg-IP-SVD} outputs one specific solution
  $\hat\calF, \hat\bB_1,  \hat\bGamma_1, \hat\bA_2, \hat\bA_3$ such that Assumption  \ref{assum:stefa-ic-1} is satisfied on the observed $\bX_1$ for $ \hat\bB_1$ and $\hat\bGamma_1$, and Assumption \ref{assum:tfm-ic-1} is satisfied for $\hat\bA_2$ and $\hat\bA_3$. 
  
  In Section \ref{sec:thm}, our theoretical results show that the estimators output by the Algorithm \ref{alg-IP-SVD} is close to the estimation targets that satisfy Assumption \ref{assum:tfm-ic-1} and \ref{assum:stefa-ic-1}. 
  As a result, 
  $\hat\bB_1 \approx \mathring{\bB}_1 \bH_1$, 
  $\hat\bA_2 \approx \mathring{\bA}_2 \bH_2$, 
  $\hat\bA_3 \approx \mathring{\bA}_3 \bH_3$, and 
  $\hat\calF \approx \mathring{\calF} \times_1 \bH_1^{-1} \times_2 \bH_2^{-1} \times_3 \bH_3^{-1}$. 
  For a new observation $\bX^{new}_1$, we have 
  \begin{align*}
      \hat \calS_{sieve}^{new} & = \hat \calF \times_1 \bPhi_1(\bX^{new}_1) \hat\bB_1 \times_2 \hat \bA_2 \times_3 \hat\bA_3 \\
      & \approx  (\mathring{\calF} \times_1  \bH_1^{-1}  \times_2 \bH_2^{-1} \times_3 \bH_3^{-1}) \times_1 ( \bPhi_1(\bX^{new}_1) \mathring{\bB}_1 \bH_1)  \times_2 ( \mathring{\bA}_2 \bH_2) \times_3 ( \mathring{\bA}_3 \bH_3) \\
      & = \mathring{\calF} \times_1 \bPhi_1(\bX^{new}_1) \mathring{\bB_1} \bH_1 \bH_1^{-1} 
      \times_2 (\mathring{\bA_2}\bH_2\bH_2^{-1}) 
      \times_3 (\mathring{\bA_3}\bH_3\bH_3^{-1}) \\
      & = \mathring{\calF} \times_1 \bPhi_1(\bX^{new}_1) \mathring{\bB}_1 \times_2 \mathring{\bA}_2 \times_3 \mathring{\bA}_3 \\
      & = \mathring{\calS}_{sieve}^{new}.
  \end{align*} 
  Here, the linear transformations $\bH_1$, $\bH_2$ and $\bH_3$ will depend on $\bX_1$.  
  But the key point here is that the respective $\bH_1$ and $\bH_1^{-1}$ transformation of $\mathring{\calF}$ and $\mathring{\bB_1}$ will canceled out and the signal part as a whole will not be affected by any specific $\bH_1$ or $\bX_1$. 
\end{remarkx}

\begin{remarkx}[Comparison to the MMC tensor regression]
    IP-SVD aims to estimate both the covariate-explainable and covariate-orthogonal components in the STEFA model while the objective of the MMC tensor regression \citep{xu2019generalized,hu2021generalized} is to estimate the reduced-rank coefficients in a tensor regression with observed independent variables. 
    For prediction, the covariate-explainable component in the STEFA model is predicted by Sieve approximation and the covariate-orthogonal component is predicted by kernel approximation, which are very different from the regression-based prediction in  \cite{xu2019generalized}. 
    We report a simulation in Appendix section \ref{sec:compare-xu-2019} to show cases when STEFA performance better in prediction.
    
\end{remarkx}

\section{Theoretical Results} \label{sec:thm}

In this section, we establish the statistical properties of the estimators in Algorithm \ref{alg-IP-SVD} assuming data is generated from model (\ref{eqn:tucker-covari-3}).
Lemma \ref{lem:init} and Corollary \ref{cor:optimalJ} provide error bounds of the column spaces spanned by $\tilde\bG_m^{(t)}$ for $0 \le t \le t_{max}$, which concerns with the estimation errors of the projected spectral initialization and the projected power iterations in Algorithm \ref{alg-IP-SVD}. 
Theorem \ref{thm:F} provides the estimation errors of the final estimators of $\hat\bG_m$ and $\hat\calF$ from their respective estimation targets $\bG_m$ and $\calF$ that satisfy identification condition Assumption \ref{assum:stefa-ic-1}. 
Theorem \ref{thm:Gamma} provides the errors for the covariate-independent loadings $\hat\bGamma_m$. 
We provide discussions after each theorem, revealing some interesting observations in the interaction of parametric, non-parametric estimations, and iterative tensor projection.

We impose two assumptions, respectively, on the smoothness of the loading functions and on tail behavior of the noise. 
The smoothness assumption is standard in the non-parametric literature, while the tail condition is weaker than what is usually assumed in the tensor decomposition literature.

\begin{assump}[Smooth loading functions] \label{assump:sieve}
    We assume that, for all tensor modes $m \in [M]$, 
    \begin{enumerate}[label=(\roman*)]
        \item 
        The loading functions $g_{m,r_m}(\bx_m)$, $\bx_m \in \calX_m \in \RR^{D_m}$ belong to a H{\"o}lder class $\calA^{\tau}_c(\calX_m)$ ($\tau$-smooth) defined by
        \[
        \calA^{\tau}_c(\calX_m) = \left \{g \in \calC^q(\calX_m): \underset{[\eta]\le q}{\sup} \; \underset{\bx \in \calX_m}{\sup} \left| D^{\eta}\, g(\bx) \right| \le c, \text{ and }
        \underset{[\eta]=q}{\sup} \; \underset{\bu, \bv \in \calX_m}{\sup} \frac{\left| D^{\eta}\, g(\bu) - D^{\eta}\, g(\bv) \right|}{\norm{\bu - \bv}^{\beta}_2}  \le c   \right \},
        \]
        for some positive number $c$, where $ \tau=q+\beta$ is assumed $\tau\geq 2$. Here, $\calC^q(\calX_m)$ is the space of all $q$-times continuously differentiable real-value functions on $\calX_m$.
        The differential operator $D^{\eta}$ is defined as $D^{\eta} = \frac{\partial^{[\eta]}}{\partial x_1^{\eta_1}\cdots\partial x_{d_m}^{\eta_{D_m}}}$ and $[\eta] = \eta_1 + \cdots + \eta_{D_m}$ for non-negative integers $\eta_1,\cdots,\eta_{D_m}$.
        
        \item The sieve coefficients $\bb_{m,r_m}= [b_{m,r_m,1} \; b_{m,r_m,2} \; \cdots \; b_{m,r_m,J_m}]^{\top}$ for all $1\leq r_m\leq R_m$, satisfy, as $J_m\to\infty$,
        $$
        \sup_{\bx\in \mathcal{X}_m} \Big|g_{m,r_m}(\bx)- \sum\nolimits_{j=1}^{J_m}b_{m,r_m,j}\phi_j(\bx)\Big|^2=O(J_m^{-\tau})
        $$
        where $\braces{\phi_j(\cdot)}_{j=1}^{J_m}$ is a set of basis functions, and $J_m$ is the sieve dimension.
    \end{enumerate}
\end{assump}

Assumption~\ref{assump:sieve}  imposes mild conditions on loading functions so that their sieve approximation errors are well controlled.
It is satisfied if the loading functions $g_{m,r_m}(\bx_m)$, $m\in[M]$, belong to the H{\"o}lder class \citep{tsybakov2008introduction}.
The basis functions that satisfy Assumption~\ref{assump:sieve} include polynomial, wavelet basis, and B-splines \citep{chen2007large}. 
To nonparametrically estimate  $g_{m,r_m}(\bx_m)$ without the curse of dimensionality when $\bx_m$ is multivariate, we could impose certain low-dimensional structure on  $g_{m,r_m}(\bx_m)$, such as an additive structure used in \cite{fan2016projected}. 
To emphasis the main theoretical founding, we use $\tau$ as a given parameter in the following theorems and avoid dissecting it from the perspective of nonparametric estimation.  

\begin{assump}[Sub-exponential noise] \label{assump:noise} 
Each entry $\eps_{\omega}$ of the noise tensor $\calE$ are i.i.d. sub-exponential random variables with $\EE(\eps_{\omega})= 0$  
and $\EE\exp(\eps_{\omega}/K_0)\leq e$ for some constant $K_0=O(1)$, for all $\omega\in[I_1]\times [I_2]\times [I_3]$.
\end{assump}
The independence condition in Assumption~\ref{assump:noise} is standard for the statistical analysis of tensor factor model \citep{richard2014statistical, zhang2018tensor, xia2019sup,han2020tensor} and tensor time series \citep{chen2019factor,han2020tensor}.
However, all these prior works assume the Gaussian or sub-Gaussian distributions of noise.
Our Assumption~\ref{assump:noise} is weaker, which requires only a sub-exponential tail on the noise. Note that Assumption~\ref{assump:noise} implies that ${\rm Var}(\eps_\omega)=O(1)$. 


We first present the estimation errors related to the iterates of covariate-relevant loadings $\tilde\bG_m^{(t)}$ for $0 \le t \le t_{max}$, which correspond to the rates of convergence of the eigen-space spanned by the columns of $\bG_m^{(t)}$.
For a clear presentation, the theorems are presented for the case of $M=3$. 
The results can be easily extended to higher order tensors with $M > 3$. 
 Recall that we write $I_m$'s for the tensor dimensions,  $J_m$'s for the sieve dimensions of covariate-relevant component,  and $R_m$'s for the Tucker ranks of covariate-independent component. 
We also assume that $I_1\geq I_2\geq I_3$ and $R_1\geq R_2\geq R_3$ for brevity of notations. 
The {\it signal strength} of $\calF$ is measured by 
$\lambda_{\min}:=\underset{m\in[M]}{\min}\sigma_{R_m}\big(\calM_m(\calF)\big)$, 
which is the smallest  singular value of all the matricizations of $\calF$. 
The condition number of $\calF$ is defined as 
$\kappa_0:=\underset{m\in[M]}{\max}\ \|\calM_m(\calF)\|/\lambda_{\min}$. 
Since the noise has a bounded variance under Assumption~\ref{assump:noise}, the signal strength $\lambda_{\min}$ is regarded as the {\it signal-to-noise ratio} (SNR). See a similar definition in \cite{zhang2018tensor}. 

\begin{lemma}[Projected spectral initialization and projected power iterations] \label{lem:init}
    Suppose that Assumptions \ref{assump:sieve} and \ref{assump:noise} hold 
    under model (\ref{eqn:tucker-covari-3}),  
    the condition number $\kappa_0=O(1)$, $J_1\asymp J_2\asymp J_3$, and $R_m\leq J_m$. 
    If $\sqrt{I_1I_2I_3}\lambda_{\min}\geq C_1\Big(\kappa_0\sqrt{R_1J_1}\log^2I_1+(R_1J_1J_2J_3)^{1/4}\log^2I_1\Big)$ and $R_m J_m^{-\tau}\leq C_1^{-1}$ for some large enough absolute constant $C_1>0$.  Then it holds with probability at least $1-7I_1^{-2}$  that
    \begin{align}\label{eq:init_err}
        I_m^{-1}\norm{ \tilde\bG_m^{(0)} \tilde \bG_m^{(0)\top}-\bG_m\bG_m^{\top}  }_{\rm F}
        \leq C_4\left(\frac{\sqrt{R_1J_1}\log^2I_1}{\lambda_{\min}\sqrt{I_1I_2I_3}}+\frac{\sqrt{R_1J_1J_2J_3}\log^4I_1}{\lambda_{\min}^2I_1I_2I_3}+\sqrt{R_m} J_m^{-\tau/2}\right)
    \end{align}
    for some absolute constant $C_4>0$. 
    Moreover, for all  $t=1,\cdots,t_{\max}$, it holds with probability at least $1-48I_1^{-2}$,
    \begin{align} \label{eq:pi_rate}
        \max_{m} \; I_m^{-1} \norm{ \tilde \bG_m^{(t)}\tilde \bG_m^{(t)\top}-\bG_m\bG_m^{\top} }_{\rm F} &\leq \frac{1}{2}\cdot\max_{m} \; I_m^{-1} \norm{\tilde \bG_m^{(t-1)}\tilde \bG_m^{(t-1)\top}-\bG_m\bG_m^{\top} }_{\rm F}\\
        +&2\sqrt{R_1}J_1^{-\tau/2}+C_4'\frac{\sqrt{J_1R_1+R_1R_2R_3}\log^{2}I_1}{\lambda_{\min}\sqrt{I_1I_2I_3}},\nonumber
    \end{align}
    where $C_4'>0$ is an absolute constant. 
    Therefore, after $t_{\max}=O(\log(\lambda_{\min}\sqrt{I_1I_2I_3/J_1})+\tau\cdot \log (J_1)+1)$ iterations, it holds with probability at least $1-48I_1^{-2}$.
    \begin{align}\label{eq:Gtmax-G}
        \max_m \; I_m^{-1} \norm{ \tilde \bG_m^{(t_{\max})}\tilde \bG_m^{(t_{\max})\top}-\bG_m\bG_m^{\top} }_{\rm F}
        \leq C_5'\frac{\sqrt{J_1R_1+R_1R_2R_3}\log^{2}I_1}{\lambda_{\min}\sqrt{I_1I_2I_3}}+2\sqrt{R_1}J_1^{-\tau/2},
    \end{align}
    where $C_5'>0$ is an absolute constant. 
\end{lemma}
Recall that $\tau$ characterizes the smoothness of the covariate-relevant loading functions.  As shown in Lemma~\ref{lem:init},  if $\tau$ is larger,  the estimation error decreases. 
The projected initialization $\tilde\bG_m^{(0)}$ in Algorithm \ref{alg-IP-SVD} is obtained by the projected PCA \citep{fan2016projected}. 
By Lemma~\ref{lem:init}, a warm initialization satisfying $I_m^{-1} \norm{\tilde \bG_m^{(0)}\tilde \bG_m^{(0)\top}-\bG_m\bG_m^{\top} } \leq 1/2$ is guaranteed as long as $\sqrt{I_1I_2I_3}\lambda_{\min}\geq  C_4'(J_1J_2J_3)^{1/4}\log^2 I_1+C_5'\sqrt{J_1}\log^2I_1$ for $R_1=O(1)$ and some absolute constants $C_4'$ and $C_5'$. 
Compared with the vanilla spectral initialization \citep{zhang2018tensor,xia2019sup, richard2014statistical} that requires $\sqrt{I_1I_2I_3}\lambda_{\min}\gg (I_1I_2I_3)^{1/4}$ and sub-Gaussian noise, our projected spectral initialization requires a substantially weaker condition on the signal strength when $J_m\ll I_m$. 
The logarithmic factors in Lemma~\ref{lem:init} emerge from the sub-exponential tail of noise distribution, which has never been studied in existing literature.
Moreover, the initialization error \eqref{eq:init_err} has two leading terms.
When the signal strength $\lambda_{\min}$ is only medium strong, that is, $\sqrt{I_1I_2I_3}\lambda_{\min}\gg (J_1J_2J_3)^{1/4}$ but $\sqrt{I_1I_2I_3}\lambda_{\min}\ll J_1$, the second term in \eqref{eq:init_err} dominates and the initialization error is at the order of $(R_1J_1J_2J_3)^{1/2}\log^4(I_1)/(\lambda_{\min}^2I_1I_2I_3)$.

The initialization obtained by projected PCA \citep{fan2016projected} is sub-optimal for tensor data and the IP-SVD refines it by projected power iteration. 
Equation \eqref{eq:pi_rate} shows that the error is decreasing after each mid-step projected power iteration. 
In the end, the error \eqref{eq:Gtmax-G} of the final estimator is at a smaller order of $(J_1R_1)^{1/2}\log^{2}(I_1)/\big(\lambda_{\min}(I_1I_2I_3)^{1/2}\big)$.

The estimation error in \eqref{eq:Gtmax-G} of the final estimator is a mixture of two terms. 
The first term can be viewed as a parametric rate and is related to the model complexity in approximating $\bG_m$ by the column space of $\bPhi(\bX_m)$. 
Similar to usual parametric settings, the dimension of the $J_1 \times R_1$ parameter matrix $\bB_1$ appears in the numerator of this first term. 
An interesting fact is that the parametric estimation error decreases when the signal strength $\lambda_{\min}$ increases, and increases when the Sieve dimension $J_m$ increases.  

The second term in the estimation error of \eqref{eq:Gtmax-G} can be viewed as a non-parametric rate and is related to functional approximation errors which relies crucially on the Sieve dimension. 
This rate is unaffected when signal strength $\lambda_{\min}$ changes, but decreases when the Sieve dimension $J_m$ increases. 
So there is a trade-off in choosing Sieve dimension in order to balance the parametric and non-parametric rates. 
The following result establishes the estimation error with the optimally-chosen Sieve dimension $J_m$. 

\begin{corollary}\label{cor:optimalJ}
Under the conditions of Lemma~\ref{lem:init} and $J_1=\lceil C_6\big(\log^{2}(I_1)/(\lambda_{\min}\sqrt{I_1I_2I_3})\big)^{-2/(\tau+1)}\rceil$, it holds that, for some absolute constants $C_6, C_7,C_8>0$, with probability at least $1-48I_1^{-2}$,
\begin{align*}
\max_m\; I_m^{-1} \norm{ \widetilde{\bG}_m^{(t_{\max})}\widetilde{\bG}_m^{(t_{\max})\top}-\bG_m\bG_m^{\top} }_{\rm F} 
\leq C_7\sqrt{R_1}\Big(\frac{\log^{2}I_1}{\lambda_{\min}\sqrt{I_1I_2I_3}}\Big)^{\frac{\tau}{\tau+1}}+C_8\frac{\sqrt{R_1R_2R_3}\log^{2}I_1}{\lambda_{\min}\sqrt{I_1I_2I_3}}.
\end{align*}
\end{corollary}

The first rate in Corollary~\ref{cor:optimalJ} dominates whenever $R_2, R_3=O(1)$.  
This rate is very typical in non-parametric regression \citep{chen2007large, tsybakov2008introduction} and it shows that the estimation error of $\widetilde{\bG}_m^{(t_{\max})}$ decreases when the true loading functions are smoother in Assumption~\ref{assump:sieve}.

Till now, we have shown that the space spanned by the columns of the loadings $\bG_m$ can be consistently estimated.  
Next, we show that the columns of $\bG_m$ and tensor factor $\calF$ can be determined up to a sign for the restricted estimation targets $\bG_m$ and $\calF$ that satisfy the identification condition Assumption \ref{assum:stefa-ic-1}. 
The concept of the eigengap of tensor $\calF$ is needed before we present those results. 
Here, we define 
\begin{equation*}
{\rm Egap}(\calF)=\min_{1\leq m\leq M}\Big\{\min_{1\leq j\leq R_m}\sigma_j\big(\calM_m(\calF)\big)-\sigma_{j+1}\big(\calM_m(\calF)\big)\Big\},
\end{equation*}
where we denote $\sigma_{R_m+1}\big(\calM_m(\calF)\big)=0$. 
Intuitively, ${\rm Egap}(\calF)$ represents the smallest gap of singular values of $\calM_m(\calF)$ for all $m\in[M]$.
The eigengap condition on ${\rm Egap}(\calF)$ is imposed to ensure that the order of singular values will not be violated by small perturbations. 

\begin{theorem}[Covariate-relevant loadings and tensor factor] \label{thm:F}
    Suppose that the signal strength satisfies $\sqrt{I_1I_2I_3}\lambda_{\min}\geq C_0\Big(\kappa_0\sqrt{R_1J_1}\log^2I_1+(R_1J_1J_2J_3)^{1/4}\log^2I_1\Big)$ under model (\ref{eqn:tucker-covari-3}), the conditions of Lemma~\ref{lem:init} and
    \begin{equation*}
        {\rm Egap}(\calF)\geq C_1 \sqrt{J_1R_1^2+R_1^2R_2R_3}\log^{2}(I_1)/\sqrt{I_1I_2I_3}+C_2\lambda_{\min}R_1J_1^{-\tau/2}
    \end{equation*}
    hold for some absolute constants $C_0, C_1, C_2>0$.
    Let $\hat\calF$ and $\hat \bG_m$ be the estimators after orthogonality calibration \eqref{eqn:FandGm-final}. 
    Then there exist diagonal matrices $\{\bS_m\}_{m\in[3]}$ whose diagonal entries are either $-1$ or $+1$ such that, with probability at least $1-49I_1^{-2}$,
    \begin{equation*}
        \max_{m\in[3]}\; I_m^{-1/2} \norm{\hat \bG_m-\bG_m\bS_m}_{\rm F}
        \leq C_7\frac{\sqrt{J_1R_1+R_1R_2R_3}\log^{2}I_1}{\lambda_{\min}\sqrt{I_1I_2I_3}}+C_8\sqrt{R_1} J_1^{-\tau/2}
    \end{equation*}
    and
    \begin{equation*}
    \|\hat\calF-\calF\times_1\bS_1\times_2\bS_2\times_3 \bS_3\|_{\rm F}\leq C_7'\frac{\sqrt{J_1R_1+R_1R_2R_3}\log^{2}I_1}{\sqrt{I_1I_2I_3}}+C_8'\lambda_{\min}\sqrt{R_1} J_1^{-\tau/2}
    \end{equation*}
where $C_7, C_8, C_7', C_8'>0$ are absolute constants.
\end{theorem}

Here the columns of factor loadings $\bG_m$ can be determined up to a sign which is common in matrix singular value decomposition. 
Similarly to Corollary~\ref{cor:optimalJ}, if we choose $J_1\asymp \lceil \big(\log^{2}(I_1)/(\lambda_{\min}\sqrt{I_1I_2I_3})\big)^{-2/(\tau+1)}\rceil$, Theorem~\ref{thm:F} implies that 
 \begin{equation*}
\|\hat\calF-\calF\times_1\bS_1\times_2\bS_2\times_3 \bS_3\|_{\rm F}\leq C_7'\sqrt{R_1}\lambda_{\min}^{\frac{1}{\tau+1}}\cdot\Big(\frac{\log^{2}I_1}{\sqrt{I_1I_2I_3}}\Big)^{\frac{\tau}{\tau+1}}+C_8\Big(\frac{R_1R_2R_3\log^4I_1}{I_1I_2I_3}\Big)^{1/2}.
 \end{equation*}
The second term on the right hand side is negligible if $\lambda_{\min}\sqrt{I_1I_2I_3}\geq C_8'(R_2R_3)^{\frac{\tau+1}{2}}\log^{2}I_1$. 
So the first term dominates when $R_1=O(1)$. 
Moreover, the first term decreases when $\tau$ increases implying that the core tensor can be more accurately estimated if the loading functions are smoother.

Finally, we bound the estimation error for the covariate-independent components.
\begin{theorem}[Covariate-independent loadings]\label{thm:Gamma}
    Suppose the conditions of Theorem~\ref{thm:F} hold under model (\ref{eqn:tucker-covari-3}).  Then, for all $m=1,2,3$, it holds with probability at least $1-50I_1^{-2}$ that 
    \begin{align*}
    \norm{\hat\bGamma_m-\bGamma_m\bS_m}_{\rm F} 
    \leq & \quad C_8\|\bGamma_m\|\cdot\Big( \frac{\sqrt{J_1R_1+R_1R_2R_3}\log^{2}I_1}{\lambda_{\min}\sqrt{I_1I_2I_3}}+\sqrt{R_1}J_1^{-\tau/2}\Big) \\
   &  +C_9\frac{\sqrt{R_1I_1+J_1R_1^2+R_1^2R_2R_3}\log^{3/2} I_1}{\lambda_{\min}\sqrt{I_2I_3}}
    \end{align*}
    where $\bS_m$ is defined as in Theorem~\ref{thm:F} and $C_8,C_9>0$ are some absolute constants. 
    By choosing  $J_1\asymp \lceil \big(\log^{2}(I_1)/(\lambda_{\min}\sqrt{I_1I_2I_3})\big)^{-2/(\tau+1)}\rceil$, we get, with probability at least $1-50I_1^{-2}$, that
    \begin{align}\label{eq:hatGamma-Gamma}
    \norm{\hat\bGamma_m-\bGamma_m\bS_m}_{\rm F}\
    \leq & \quad C_8'\|\bGamma_m\|\cdot\bigg(\sqrt{R_1}\Big(\frac{\log^{2}I_1}{\lambda_{\min}\sqrt{I_1I_2I_3}}\Big)^{\frac{\tau}{\tau+1}}+\frac{\sqrt{R_1R_2R_3}\log^{2}I_1}{\lambda_{\min}\sqrt{I_1I_2I_3}}\bigg) \\
     & +C_9'\frac{\sqrt{R_1I_1+J_1R_1^2+R_1^2R_2R_3}\log^{3/2} I_1}{\lambda_{\min}\sqrt{I_2I_3}}\notag
    \end{align}
for some absolute constants $C_8', C_9'>0$.
\end{theorem}

The error bound \eqref{eq:hatGamma-Gamma} involves two terms. 
The second term is similar to (except a logarithmic factor) the typical rate of tensor factor models \citep{zhang2018tensor,richard2014statistical} if $I_1\geq J_1R_1+R_1R_2R_3$. 
However, there is a crucial difference in the STEFA model since no condition is required for $\bGamma_m$ (such as orthogonality of its columns).
The first term in \eqref{eq:hatGamma-Gamma} emerges from the estimation error of covariate-relevant component $\bG_m$. 
For ease of exposition, assume $\|\bGamma_m\|_{\rm F}\asymp R_1^{1/2}\|\bGamma_m\|$ and $R_1=O(1)$. The rate \eqref{eq:hatGamma-Gamma} yields the relative error of $\hat\bGamma_m$ as
\begin{equation}\label{eq:rel_Gamma}
\frac{\|\hat\bGamma_m-\bGamma_m\bS_m\|_{\rm F}}{\|\bGamma_m\|_{\rm F}}\leq C_8'\Big(\frac{\log^{2}I_1}{\lambda_{\min}\sqrt{I_1I_2I_3}}\Big)^{\frac{\tau}{\tau+1}}+C_9'\frac{\sqrt{I_1}\log^{3/2}I_1}{\lambda_{\min}\sqrt{I_2I_3}\|\bGamma_m\|}.
\end{equation}
Therefore, the estimator $\hat\bGamma_m$ is consistent in relative Frobenius-norm error if $\lambda_{\min}\sqrt{I_1I_2I_3}\gg \log^2 I_1$ and $\lambda_{\min}\|\bGamma_m\|(I_2I_3)^{1/2}\gg I_1^{1/2}\log^{3/2}I_1$. The former condition is mild in view of the signal strength condition in Theorem~\ref{thm:F}. The latter condition relies on the magnitude of $\|\bGamma_m\|$, and it dominates the former one if $\|\bGamma_m\|\leq I_1\log^{-1/2}I_1$. Basically, if $\|\bGamma_m\|$ becomes smaller,  a larger signal strength $\lambda_{\min}$ is required to ensure the consistency of $\hat\bGamma_m$. 

{\it Comparison with HOOI}. Ignoring the covariate information and assuming the orthogonality of the columns of $\bGamma_m$, one can apply the higher-order orthogonal iteration (HOOI) algorithm to estimate $\bGamma_m$ (additional treatments are perhaps necessary to separate $\bGamma_m$ from the covariate-relevant component $\bG_m(\bX_m)$). It is proved in \cite{zhang2018tensor} that if the SNR satisfies $\lambda_{\min}\|\bGamma_1\|\|\bGamma_2\|\|\bGamma_3\|\geq C_0(I_1^{1/2}+(I_1I_2I_3)^{1/4})$, the HOOI algorithm outputs an estimator attaining, with high probability,  a relative Frobenius-norm error rate as
\begin{equation}\label{eq:hooi}
\frac{\|\hat\bGamma_m^{\hooi}-\bGamma_m\bO_m\|_{\rm F}}{\|\bGamma_m\|_{\rm F}}\leq C_8''\frac{\sqrt{I_1}}{\lambda_{\min}\|\bGamma_1\|\|\bGamma_2\|\|\bGamma_3\|}
\end{equation}
where $\bO_m$ is an orthogonal matrix that minimizes $\|\hat\bGamma_m^{\hooi}-\bGamma_m\bO\|_{\rm F}$. 
For ease of comparison, let us further assume $\|\bGamma_1\|\asymp \|\bGamma_2\|\asymp \|\bGamma_3\|$ and $I_1\asymp I_2\asymp I_3$. 
Comparing (\ref{eq:hooi}) to (\ref{eq:rel_Gamma}), if $\|\bGamma_1\|\gg I_1^{1/2}$, i.e., the covariate-independent component has a signal strength (characterized by $\|\bGamma_1\|$)  stronger than the covariate-relevant one (that is simply $I_1^{1/2}$ by Assumption~\ref{assum:stefa-ic-1}) ,  HOOI achieves a sharper error rate than our STEFA-based estimator.  On the other hand, STEFA can outperform HOOI when $\|\bGamma_1\|\ll I_1^{1/2}$. Nonetheless, STEFA still enjoys a major advantage over HOOI by exploiting the covariate information. 
Indeed, the auxiliary covariates can potentially reduce the SNR requirement. Note that our Theorem~\ref{thm:Gamma} suggests that an SNR condition $\sqrt{I_1I_2I_3}\lambda_{\min}\geq C_0 (J_1J_2J_3)^{1/4}$ suffices to estimate the covariate-independent component, while HOOI requires an SNR condition $\|\bGamma_1\|\|\bGamma_2\|\|\bGamma_3\|\lambda_{\min}\geq C_0 (I_1I_2I_3)^{1/4}$. Therefore, if $J_1, J_2, J_3\ll I_1$ and $\|\bGamma_1\|=O(I_1^{1/2})$, STEFA requires a weaker SNR condition. 

\section{Numerical studies} \label{sec:simu}

In this section, we use Monte Carlo simulations to assess the performances of the IP-SVD algorithm on the STEFA model under different settings. 
In all examples, the observation tensor $\calY$ is generated according to model~\eqref{eqn:tucker-covari-3}, of which the dimensions of the latent tensor factor and the covariates are fixed at $R_m=R=3$ and $D_m=D=2$.
We generate the noise tensor $\calE$ with each entry $\eps_{i_1 i_2 i_3} \sim \calN(0,1)$.
The core tensor $\cF$ is obtained from the core tensor of the Tucker decomposition of a $R_1\times R_2\times R_3$ random tensor with i.i.d. $\calN(0,1)$ entries.
The core tensor is further scaled such that $\lambda_{\min} \defeq \min_{m} \sigma_{R_m}(\calM_m(\calF)) = (I_{min})^\alpha$, where $I_{min}=\min\{I_1, I_2, I_3\}$ with some desired value of $\alpha$.
This characterization of signal strength was proposed in \cite{zhang2018tensor} and we focus on the low signal-to-noise ratio regime ($\alpha \leqslant  0.5$), where HOOI is known to have unsatisfactory performance.

The explanatory variable matrix $\bX_m\in\mathbb R^{I_m\times D_m}$ is generated from independent uniform distribution $\calU(0, 1)$. 
We generate $\bG_m = \brackets{ g_{m,r_m}(\bx_{m,i_m \cdot})}_{i_m r_m}$ by:
\begin{equation} \label{eq:sim-sieve}
g_{m, r_m}(\bx_{m,i_m \cdot}) = \xi_{m, r_m, 0} + \sum_{d_m=1}^{D_m}\sum_{j=1}^{J^*}  \xi_{m, r_m,d_m, j}\kappa^{j-1}P_{j}(2x_{m,i_m d_m}  -1), 
\end{equation}
where $\xi_{m, r_m, 0}$ and $\xi_{m, r_m,d_m, j} \sim \mathcal N(0, 1)$, $J^*$ is the true number of basis functions, $\kappa\in(0, 1)$ is the decay coefficient to make sure convergence of sequences as $J^*$ increases, and $P_j(\cdot)$ is the $j$-th Legendre polynomial defined on $[-1, 1]$.
Note that $J^*$ denotes the true sieve order used in simulation and the $J$ used in IP-SVD is not necessarily same as $J^*$.
The generation of $\bGamma_m$ will be specified later in each setting. Whenever a non-zero $\bGamma_m$ is generated, we orthonormalize the columns of $\bA_m = \bG_m + \bGamma_m$ such that $\bA_m^\top\bA_m$ is an identity matrix. 
Here we abuse Assumption~\ref{assum:stefa-ic-1} a little bit in order to control the signal-to-noise ratio through the magnitude of the core matrix $\calF$. 
The orthonormalized $\bA_m$ and the original one differ by a linear transformation of columns, which does not affect the Schatter q-$\sin\theta$ distance. 

In what follows, we vary $(I_1, I_2, I_3)$, $\alpha$, $\bG_m$ and $\bGamma_m$ to investigate the effects of different tensor dimensions, signal-to-noise ratios and semi-parametric assumptions on the accuracy of estimating factor, loadings and loading functions.
For the error of estimating the loading $\bA_m$, we report the average Schatten q-$\sin\theta$ norm ($q=2$):
\begin{equation*}
\ell_2(\hat\bA_m) := \norm{\sin\Theta\paran{\hat\bA_m, \bA_m}}_2, \quad m\in[3].
\end{equation*}
For the error of estimating the loading function $g_{m,r}(\bx)$, $m\in[M]$, we report
\begin{equation*}
\ell(\hat g_{m,r}) := \frac{\int \abs{\hat g_{m,r}(\bx) - g_{m,r}^\star(\bx)}^2 d \bx}{\int \abs{g_{m,r}^\star(\bx)}^2 d \bx}, \quad r\in[3].
\end{equation*}
For the error of the estimation of a tensor $\calY$, we report the relative mean squared error 
${\rm ReMSE}_{\calY} = \frac{\norm{\hat\calY-\calY}_F}{\norm{\calY}_F}$.
For the setting where all three modes share similar properties, we only report results for the $1$-st mode for conciseness.
All results are based on $100$ replications.

\paragraph{Effect of growing dimensions and signal-to-noise ratio.}
In the first experiment, we examine the effect of growing dimension $I$ and different values of $\alpha$.
We fix $R=3$, $J = J^*=4$, $\bGamma_m = \bzero$, and set $I_1 = I_2 = I_3 = I$.
We vary $I = \braces{100, 200, 300}$ and $\alpha=\braces{0.1,0.3,0.5}$.
The mean and standard deviation of $\ell_2(\hat\bA_1)$ are presented in Table \ref{tabl:growing-I-SNR-AF}. Since $I_1=I_2=I_3$, we only report the Shatten's $q$-$\sin\Theta$-norm for $\hat\bA_1$ as similar result holds for $\hat\bA_2$ and $\hat\bA_3$.
It is clear that the IP-SVD significantly improves upon HOOI in Shatten's $q$-$\sin\Theta$-norm ($q$=2) under all settings.
While both IP-SVD and HOOI perform better when $\alpha$ increases and worse when dimension $I$ increases, the IP-SVD is more favorably affected by increased $\alpha$ and less negatively affected by increased dimension $I$.
The error in estimating $g_{m,r}(\bx)$ for the first mode $m=1$ is reported in Table \ref{tabl:growing-I-SNR-G}, where the phenomenon is the same as those for $\ell_2(\hat\bA)$.
The supplementary material \cite[Section C]{chen2020supplement} also reports the same phenomenon for the unbalanced setting where $I_1$, $I_2$, and $I_3$ are different.

\begin{table}[htpb!]
\centering
\caption{The mean and standard deviation of the the average Schatten q-$\sin\theta$ loss $\ell_2(\hat\bA_1)$ and ${\rm ReMSE}_{\calY}$, from $100$ replications, under varying dimensions and signal-to-noise ratio.}
\small
\begin{tabular}{cc|ccc|ccc|ccc}
\hline
&\multicolumn{1}{c|}{$\alpha$} & \multicolumn{3}{c|}{0.1} & \multicolumn{3}{c|}{0.3} & \multicolumn{3}{c}{0.5}\\
&\multicolumn{1}{c|}{$I$}  & 100 & 200 & 300 & 100 & 200 & 300 & 100 & 200 & 300 \\
\hline
\multirow{2}{*}{\rotatebox[origin=c]{90}{IP-SVD}} & $\ell_2(\hat\bA_1)$
& \msd{1.305}{0.138} & \msd{1.303}{0.126} & \msd{1.292}{0.169} & \msd{0.866}{0.233} & \msd{0.621}{0.205} & \msd{0.574}{0.200} &\msd{0.274}{0.068} &\msd{0.195}{0.051} &\msd{0.152}{0.038}\\
& ${\rm ReMSE}_{\calY}$
& \msd{2.471}{0.519}  & \msd{2.382}{0.519} & \msd{2.281}{0.483}  & \msd{0.934}{0.283} & \msd{0.675}{0.212} & \msd{0.588}{0.179} & \msd{0.280}{0.065} & \msd{0.195}{0.044} & \msd{0.154}{0.035}\\
\hline
\multirow{2}{*}{\rotatebox[origin=c]{90}{\parbox[c]{3em}{HOOI}}} & $\ell_2(\hat\bA_1)$
& \msd{1.707}{0.012} & \msd{1.719}{0.007} & \msd{1.724}{0.004} & \msd{1.705}{0.012} & \msd{1.719}{0.006} & \msd{1.724}{0.004} & \msd{1.581}{0.189} & \msd{1.671}{0.122}  & \msd{1.691}{0.162}\\
& ${\rm ReMSE}_{\calY}$
& \msd{7.829 }{1.632}  & \msd{10.368}{2.133} & \msd{11.999}{2.379} & \msd{3.330 }{0.665} & \msd{3.652 }{0.798} & \msd{3.987 }{0.849} & \msd{1.548 }{0.323} & \msd{1.576 }{0.278} & \msd{1.556 }{0.269}\\
\hline
\end{tabular}

\label{tabl:growing-I-SNR-AF}
\end{table}

\begin{table}[htpb!]
    \centering
    \caption{Under varying dimensions and signal-to-noise ratio, the mean and standard deviation of the function approximation loss $\ell(\hat g_{m,r})$, for model $m=1$ and $r\in[3]$, from $100$ replications. This results for modes $m=2, 3$ are similar.
    }
        \begin{tabular}{cc|ccc|ccc|ccc}
            \hline
            \multirow{2}{*}{$I$} & \multirow{2}{*}{$R$} & \multicolumn{3}{c|}{$\alpha=0.1$} & \multicolumn{3}{c|}{$\alpha=0.3$} & \multicolumn{3}{c}{$\alpha=0.5$} \\ \cline{3-11}
            &  &  $\ell(\hat g_{1,1})$ & $\ell(\hat g_{1,2})$ & $\ell(\hat g_{1,3})$ & $\ell(\hat g_{1,1})$ & $\ell(\hat g_{1,2})$ & $\ell(\hat g_{1,3})$ & $\ell(\hat g_{1,1})$ & $\ell(\hat g_{1,2})$ & $\ell(\hat g_{1,3})$ \\ \hline
            100 & 3  & \msd{1.653}{1.028} & \msd{1.699}{0.749} & \msd{1.786}{0.740} & \msd{0.745}{1.224} & \msd{1.082}{1.141} & \msd{1.295}{1.014} & \msd{0.382}{1.098} & \msd{0.429}{1.078} & \msd{0.575}{1.259}\\
            200 & 3  & \msd{1.479}{0.861} & \msd{1.715}{0.653} & \msd{1.792}{0.682} & \msd{0.524}{1.119} & \msd{0.898}{1.193} & \msd{1.016}{1.119} & \msd{0.134}{0.669} & \msd{0.127}{0.558} & \msd{0.270}{0.900}\\
            300 & 3  & \msd{1.500}{0.929} & \msd{1.781}{0.725} & \msd{1.834}{0.669} & \msd{0.410}{0.916} & \msd{0.832}{1.220} & \msd{1.063}{1.299} & \msd{0.100}{0.548} & \msd{0.190}{0.778} & \msd{0.063}{0.392}\\ \hline
        \end{tabular}
    \label{tabl:growing-I-SNR-G}
\end{table}

\paragraph{Effect of the number of fitting basis.} \label{sec:J-effect}
In this experiment, we examine the effect of different choices of the number of fitting basis $J$.
Specifically, we fix $I_1=I_2=I_3=I=200$, $R=3$ and set $\bGamma_m = 0$.
We vary SNR by changing $\alpha = 0.3, 0.5$.
The loadings are simulated according to the additive sieve structure as in \eqref{eq:sim-sieve} with fixed $J^*=16$.
However, in the estimation of $\hat\bA_m$, we use different numbers of sieve orders $J=2, 4, 8, 16$.
The mean and standard deviation of $\ell_2(\hat\bA_1)$ and ${\rm ReMSE}_{\calY}$ are reported in Table~\ref{tabl:J-effect}.

A noteworthy observation is that increasing the sieve order $J$ does not consistently enhance the performance.
For both signal-to-noise strength in Table~\ref{tabl:J-effect}, $J=16$ does not achieve the best performance among all choices of $J$, even though the data is simulated with order 16.  This reflects well the bias and variance trade-off.
On one hand, increasing sieve order $J$ enhances the capability of $\bG_m$ in capturing the parametric dependence between $\bA_m$ and $\bX_m$. On the other hand, a large order $J$ increases the Frobenius norm of the projected noise $\norm{\calE\times_1\bP_1\times_2\bP_2\times_3\bP_3}_F$, which may result in a reduced signal-to-noise ratio.
Large value of $\alpha$ is more tolerant to this signal-to-noise decrease caused by large sieve order.
As shown in Table~\ref{tabl:J-effect}, the minimum error is obtained at $J=4$ when $\alpha=0.3$, while $J=8$ is the optimal one when $\alpha = 0.5$.
These observations align with findings in the realm of semiparametric studies.
For example, extensive spline bases often exhibit overfitting tendencies and are commonly employed alongside regularization techniques \citep{carroll2006discussion}.


\begin{table}[htpb!]
\centering
\renewcommand{\arraystretch}{1.2}
\caption{The average spectral and Frobenius Schatten q-$\sin\Theta$ loss for $\hat\bA_1$ and relative mean square errors for $\calY$ under various settings.}
\begin{tabular}{c|cccc|cccc}
\hline
& \multicolumn{4}{c|}{$\ell_2(\hat\bA_1)$}& \multicolumn{4}{c}{${\rm ReMSE}_{\calY}$}\\
\hline
J & 2 & 4 & 8 & 16 & 2 & 4 & 8 & 16\\
\hline

$\alpha=0.3$ &  \msd{1.024}{0.177} & \msd{0.910}{0.195}& \msd{1.093}{0.256} & \msd{1.486}{0.173} & \msd{0.886}{0.113} & \msd{0.872}{0.182} & \msd{1.154}{0.349}& \msd{1.781}{0.432}\\
\hline
$\alpha=0.5$ & \msd{0.881}{0.153} & \msd{0.503}{0.116}& \msd{0.327}{0.057} &\msd{0.445}{0.101}  & \msd{0.720}{0.080}  & \msd{0.467}{0.073} & \msd{0.303}{0.053} & \msd{0.398}{0.102}\\
\hline
\end{tabular}%

\label{tabl:J-effect}
\end{table}

\paragraph{Effect of the covariate-orthogonal loading.} \label{sec:simu-nonparam-loading}
In this experiment, we examine the effect of the covariate-orthogonal loading part $\bGamma_m$.
To simulate nonzero $\bGamma_m$ such that $\bA_m$ satisfies the identification condition, we first generate a matrix $\bLambda_m$ with each elements drawn from independent $\calN(0,1)$, project it to the orthogonal complement of $\bG_m$ and normalize each column.
Specifically, the $r$-th column of $\bGamma_m$ is obtained as
\[
\bgamma_{m, \cdot r} = \mu\cdot (\bI - \bP_{\bG_m})\blambda_{m, \cdot r}\big/\norm{(\bI - \bP_{\bG_m})\blambda_{m, \cdot r}},\quad\text{for }r = 1,\dots, R_m,
\]
where $\bP_{\bG_m}$ is the projection matrix of $\bG_m$ and $\blambda_{m, \cdot r}$ is the $r$-th column of $\bLambda_m$.
We add a scaling factor $\mu \geqslant 0$ to controls the amplitude of the orthogonal part.
Note that $\bA_m = \bG_m + \bGamma_m$ generated in this way is not necessarily an orthogonal matrix.
So a final QR decomposition is conducted on $\bA_m$ to orthonormalize the columns of $\bA_m$.
Again, we note that we orthonormalize $\bA_m$ just in order to control the overall signal-to-noise ratio. 
In the experiments, we fix $I_1=I_2=I_3=I=200$, $R=3$ and $\alpha=0.5$ and change the values of $\mu$.
The magnitude or the Frobenious norm of $\bGamma_m$ is controlled through the coefficient $\mu$. The errors under four different choices of $\mu$'s are reported in Table~\ref{tabl:gamma-effect}.
Note that in the simulation, $\bA_m = \bG_m +\bGamma_m$ is normalized such that the signal-to-noise ratio of the tensor $\calY$ can be controlled by the core tensor $\calF$. A larger value of $\mu$ indicates a smaller norm of the projected tensor $\tilde\calY$ and results in a decreased signal-to-noise ratio in the projected model. As demonstrated in Table~\ref{tabl:gamma-effect}, the error increases as $\mu$ increases.

\begin{table}[htpb!]
\centering
\renewcommand{\arraystretch}{1.2}
\caption{The means and standard deviations of $\hat\bA_1$ and ${\rm ReMSE}_{\calY}$ under various settings.}
\begin{tabular}{c|cccc}
\hline
$\mu$ & 0 & 0.01 & 0.1 & 1.0\\
\hline
$\ell_2(\hat\bA_1)$ & \msd{0.877}{0.101} & \msd{0.851}{0.125} & \msd{0.876}{0.117} &  \msd{1.285}{0.132}\\
\hline
${\rm ReMSE}_{\calY}$ & \msd{1.043}{0.274} & \msd{0.985}{0.271} & \msd{1.031}{0.296} & \msd{2.157}{0.543}\\
\hline
\end{tabular}%

\label{tabl:gamma-effect}
\end{table}

\paragraph{Effect of underlying $g_{m,r_m}(\cdot)$.}
In this experiment, we exam the potential impact of using the additive approximation \eqref{eqn:g-func-addants-vec} of $\bG_m$.
Under the setting $I_1=I_2=I_3=200$, $R=3$, $D=2$, $\alpha=0.3$, $\bGamma_m=0$ and $J=J^*=3$, we simulate $\bG_m$ according to the additive case \eqref{eq:sim-sieve} and plot the true function $g_{1,1}^\star(\bx)$ and the estimated function $\hat g_{1, 1}(\bx)$ in Figure~\ref{fig:additive}. As the additive assumption is valid for this case, the estimated function is pretty close the true one.
\begin{figure}[htpb!]
    \centering
    \includegraphics[width=0.35\textwidth]{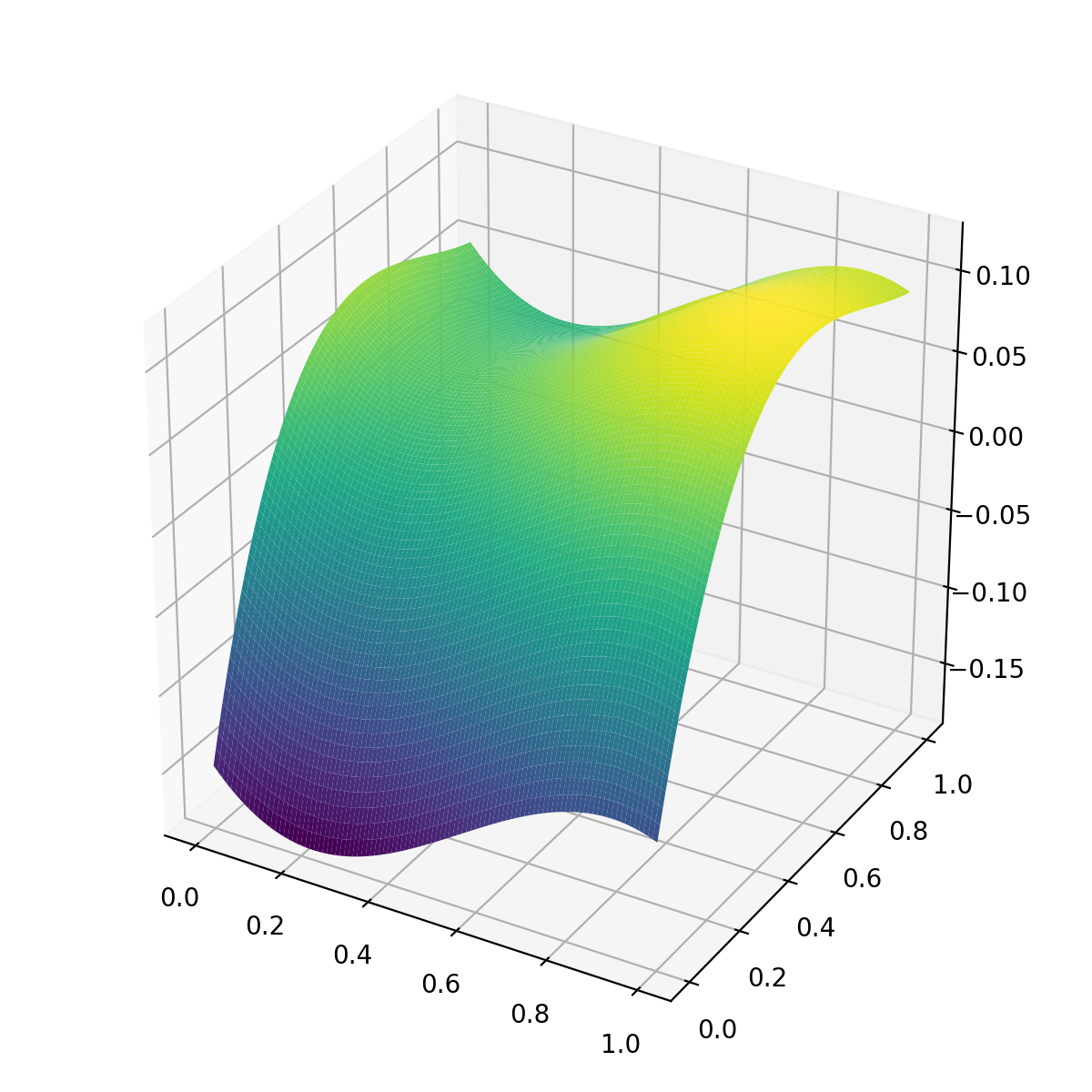}
    \hspace{2em}
    \includegraphics[width=0.35\textwidth]{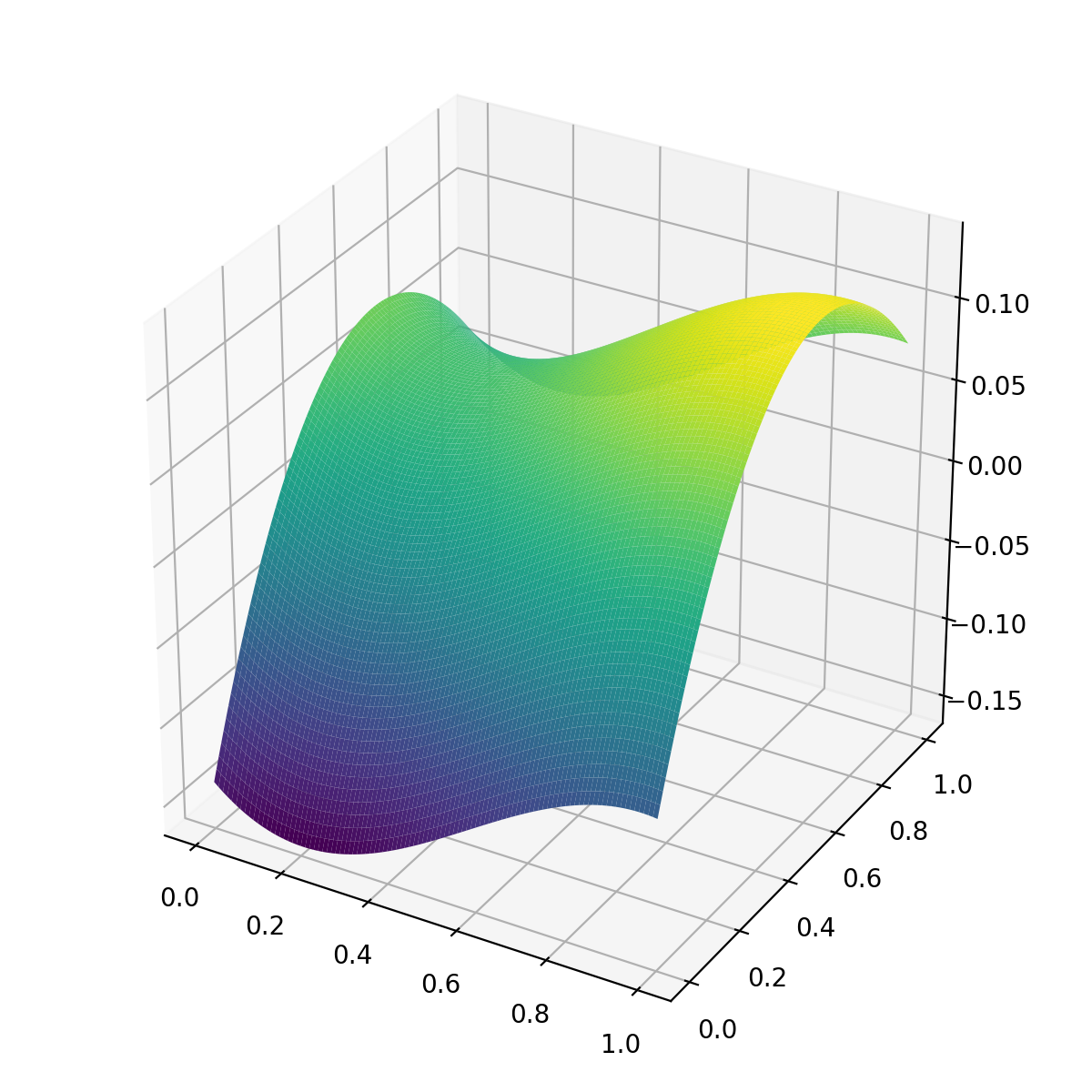}
    \caption{(Left) $g_{1, 1}^\star(\bx)$ generated under additive model. (Right) $\hat g_{1, 1}(\bx)$ estimated under additive assumption.}
    \label{fig:additive}
\end{figure}

Further, we simulate the data such that the additive assumption \eqref{eqn:g-func-addants-vec} is not valid.
Specifically, we generate $g_{m,r_m}(\bx_{m,i_m\cdot})$ in a multiplicative scheme such that
\begin{equation} \label{eqn:g_func_multiplicative}
g_{m,r_m}(\bx_{m,i_m\cdot}) = \prod_{d_m=1}^{D_m} g_{m,r_m,d_m}(x_{m,i_m d_m}).
\end{equation}
where $g_{m, r_m, d_m}$ is given by  \eqref{eq:sim-sieve}.
We conduct the IP-SVD procedure using the additive approximation \eqref{eqn:g-func-addants-vec}.
The true and estimated function of $g_{1, 1}(\bx)$ are plotted in Figure~\ref{fig:multiplicative-true} and \ref{fig:multiplicative-hat}, respectively.
The estimated function can capture some structures of the true function but misses other details as we approximate it with the additive form.
Figure~\ref{fig:multiplicative-p} depicts the projection of the true function $g_{1, 1}^\star$ to the additive sieve space used in IP-SVD.
The projection is supposed to be the best function estimate that can be obtained from the additive sieve basis.
Note that $\bA$ is identified up to an orthogonal matrix and so is $\bG$.
To address this potential problem of non-identifiability of $g_{1, 1}$, we calculate the best linear combination of $\hat g_{1, r}$, $r=1,\dots, R$, that is closest to $g_{1, 1}^*$ to mimic any potential orthogonal matrix applied to $\bG$.
The best linear combination is reported in Figure~\ref{fig:multiplicative-linear}.
As one can see, Figure~\ref{fig:multiplicative-p} and \ref{fig:multiplicative-linear} are almost identical to each other.
In conclusion, the projected Tucker under an additive basis assumption can ideally recover at most the linear (and additive) part of the true parametric component $\bG_m$.
The performance of this approximation depends on the deviation between the $\bG_m$ and its projected version $\bP_m\bG_m$.

To assess the performance of IP-SVD when the additive assumption \eqref{eqn:g-func-addants-vec} becomes invalid, we repeat the experiment in Table~\ref{tabl:growing-I-SNR-AF} with exactly the same settings except that the multiplicative scheme in \eqref{eqn:g_func_multiplicative} is used to generate $\bG_m$.
The errors in estimating $\bA_m$ and $\calY$ are reported in Table~\ref{tabl:growing-I-SNR-multiplicative}.
Comparing Table~\ref{tabl:growing-I-SNR-AF} with Table~\ref{tabl:growing-I-SNR-multiplicative}, we observe that even when the additive assumption in \eqref{eqn:g-func-addants-vec} is not valid, IP-SVD still performs better than HOOI.
But the improvement under misspecification is not as good as that under the valid additive assumption.
This shows empirically that even when the additive assumption is violated, IP-SVD in general performances better than HOOI as long as the sieve basis used in IP-SVD can partially explain the parametric part of $\bA_m$ with respect to $\bX_m$.

\begin{figure}[htpb!]
    \centering
    \begin{subfigure}[b]{0.24\textwidth}
        \centering
        \includegraphics[width=\linewidth]{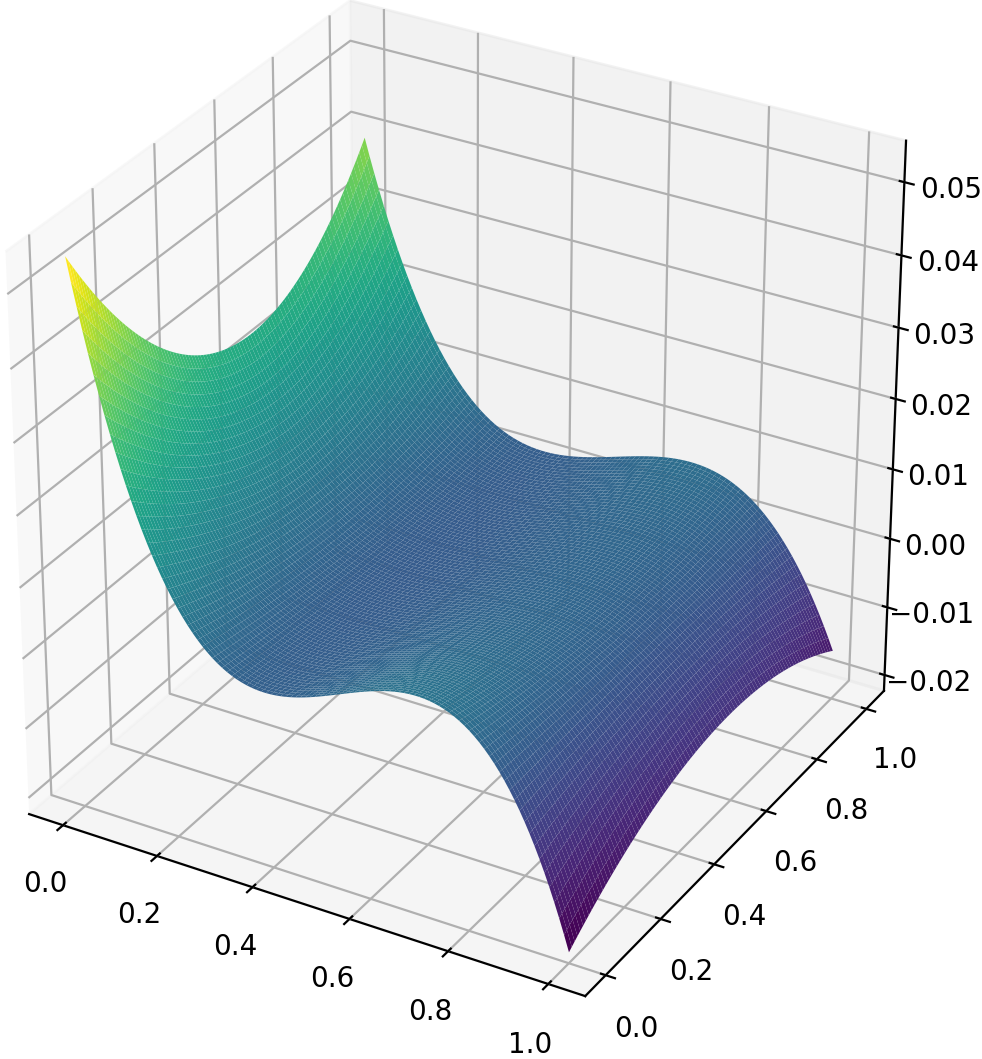}
        \caption{$g_{1,1}^\star(\bx)$}
        \label{fig:multiplicative-true}
    \end{subfigure}
    \begin{subfigure}[b]{0.24\textwidth}
        \centering
        \includegraphics[width=\linewidth]{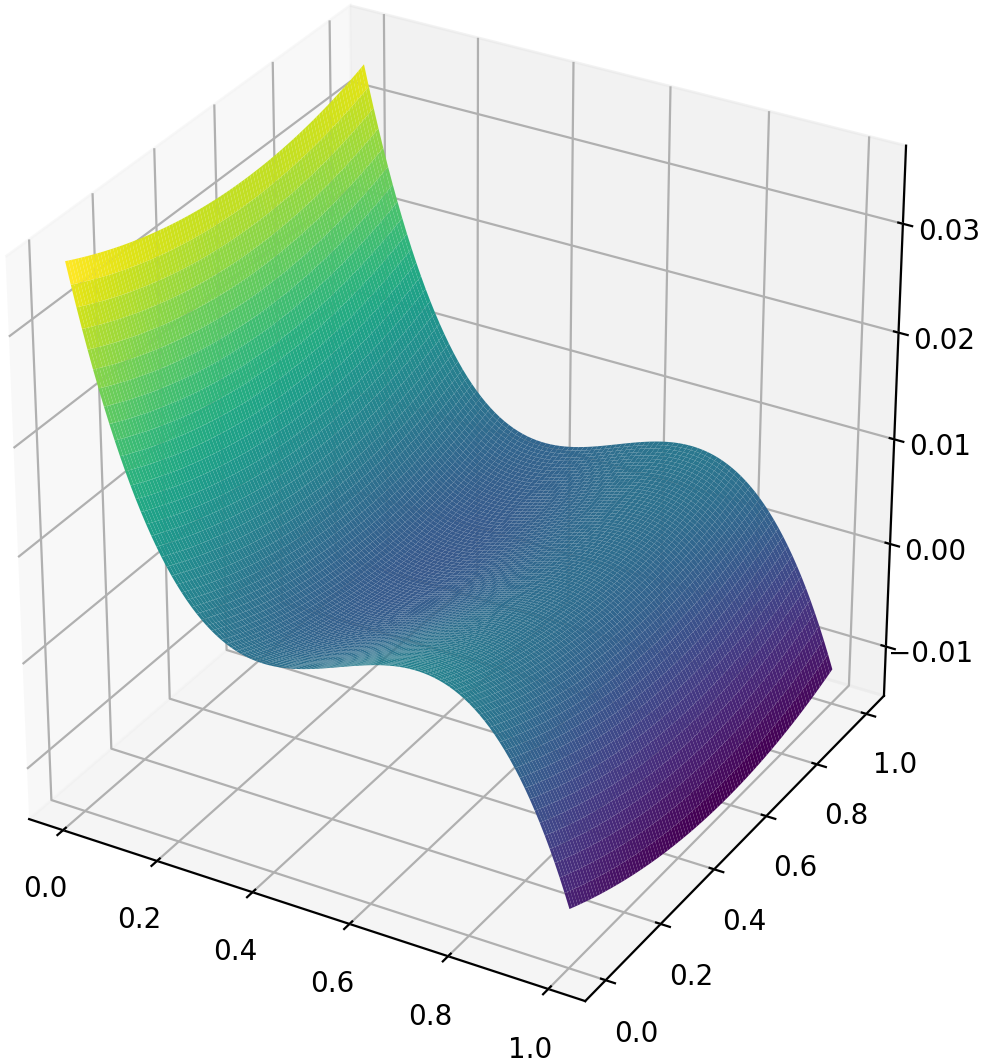}
        \caption{$\bP_1g_{1,1}(\bx)$}
        \label{fig:multiplicative-p}
    \end{subfigure}
    \begin{subfigure}[b]{0.24\textwidth}
        \centering
        \includegraphics[width=\linewidth]{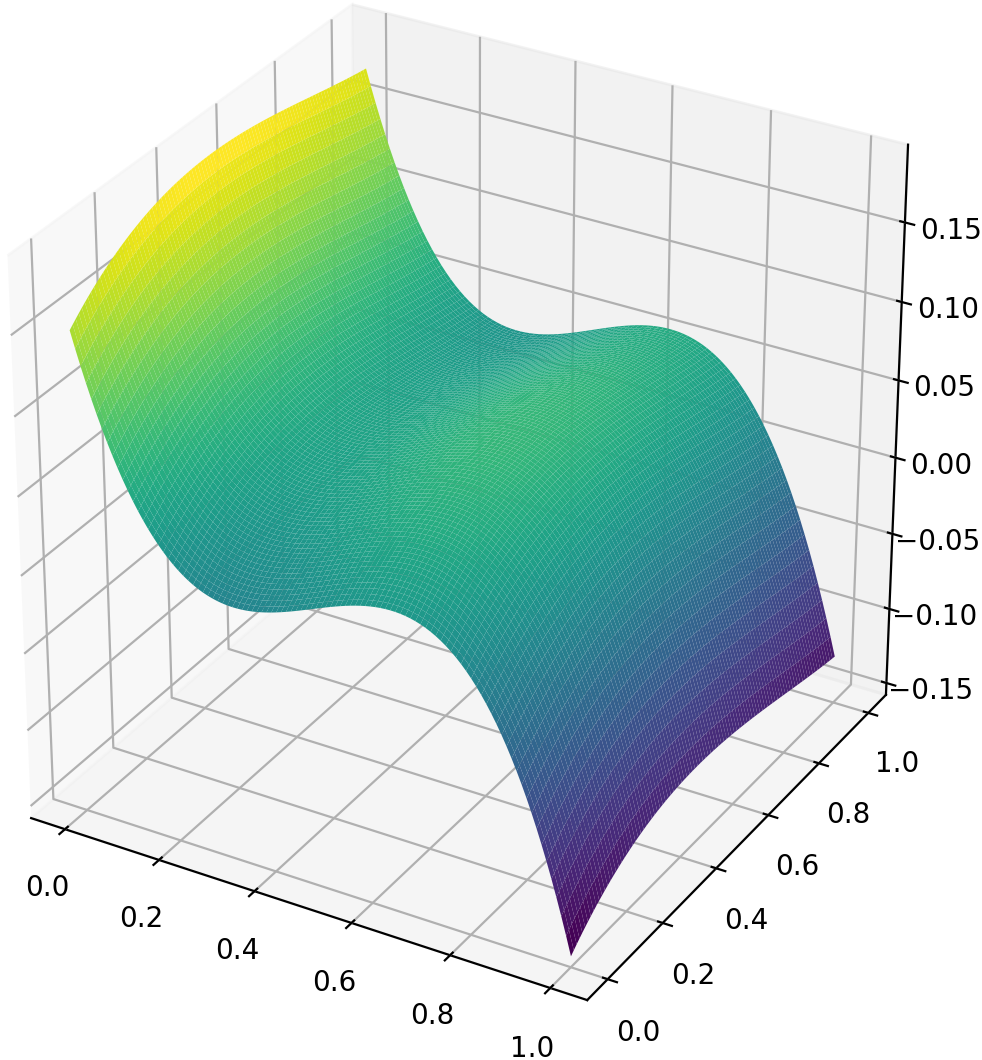}
        \caption{$\hat g_{1,1}(\bx)$}
        \label{fig:multiplicative-hat}
    \end{subfigure}
    \begin{subfigure}[b]{0.24\textwidth}
        \centering
        \includegraphics[width=\linewidth]{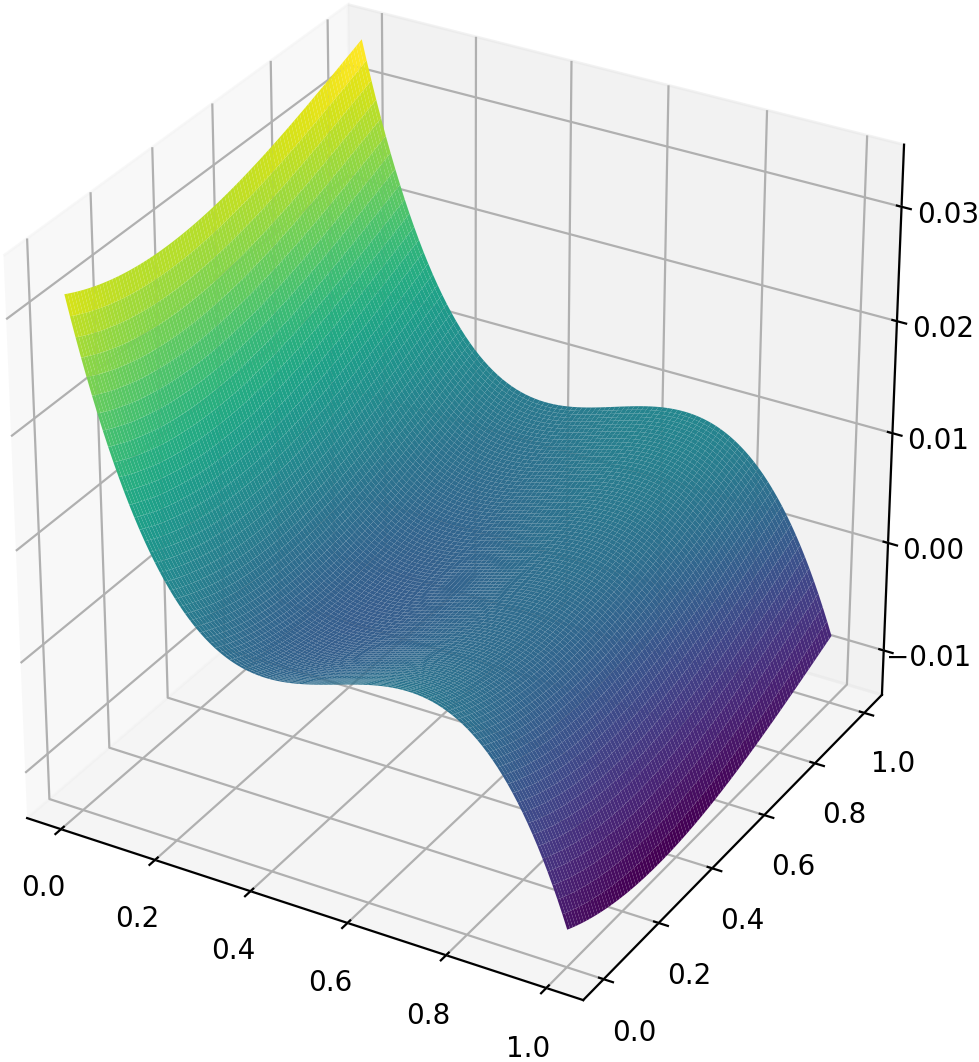}
        \caption{$\sum_{r=1}^{R}a_r\hat g_{1,r}(\bx)$}
        \label{fig:multiplicative-linear}
    \end{subfigure}
    \caption{(a) True function of $g_{1,1}^\star(\bx)$ (b) Projected version of $g_{1,1}^\star(\bx)$ (c) Estimated $g_{1,1}(\bx)$ (d) Best linear combination of $\hat g_{1, r}(\bx)$. }
    \label{fig:multiplicative}
\end{figure}

\begin{table}[htpb!]
\centering
\caption{Under varying dimensions and signal-to-noise ratio. The mean and standard deviation of $\ell_2(\hat\bA_1)$ and ${\rm ReMSE}_{\calY}$ from $100$ replications when the additive loading assumption is replaced with the multiplicative assumption.}
\small
\begin{tabular}{cc|ccc|ccc|ccc}
\hline
&\multicolumn{1}{c|}{$\alpha$} & \multicolumn{3}{c|}{0.1} & \multicolumn{3}{c|}{0.3} & \multicolumn{3}{c}{0.5}\\
&\multicolumn{1}{c|}{$I$}  & 100 & 200 & 300 & 100 & 200 & 300 & 100 & 200 & 300 \\
\hline
\multirow{2}{*}{\rotatebox[origin=c]{90}{IP-SVD}} & $\ell_2(\hat\bA_1)$
& \msd{1.430}{0.103} & \msd{1.450}{0.126} & \msd{1.438}{0.115} & \msd{1.225}{0.165} & \msd{1.182}{0.210} & \msd{1.132}{0.203} &\msd{0.853}{0.212} &\msd{0.796}{0.217} &\msd{0.820}{0.228}\\
& ${\rm ReMSE}_{\calY}$
& \msd{2.596}{0.521} & \msd{2.395}{0.579}  & \msd{2.375}{0.482} & \msd{1.189}{0.210} & \msd{1.095}{0.175}  & \msd{0.984}{0.157} & \msd{0.741}{0.117} & \msd{0.709}{0.117} & \msd{0.704}{0.132}\\
\hline
\multirow{2}{*}{\rotatebox[origin=c]{90}{\parbox[c]{3em}{HOOI}}} & $\ell_2(\hat\bA_1)$
& \msd{1.705}{0.012} & \msd{1.720}{0.006} & \msd{1.723}{0.005}  & \msd{1.705}{0.012}  & \msd{1.720}{0.005} & \msd{1.724}{0.004} & \msd{1.568}{0.213} & \msd{1.653}{0.168} & \msd{1.663}{0.188}\\
& ${\rm ReMSE}_{\calY}$
& \msd{8.092 }{1.686} & \msd{10.108}{2.596} & \msd{12.049}{2.701} & \msd{3.263 }{0.672} & \msd{3.874 }{0.721} & \msd{3.911 }{0.781} & \msd{1.528 }{0.332} & \msd{1.538 }{0.294} & \msd{1.513 }{0.300}\\
\hline
\end{tabular}%

\label{tabl:growing-I-SNR-multiplicative}
\end{table}

\section{Real data applications}  \label{sec:appl}


\subsection{Multi-variate Spatial-Temporal Data} \label{sec:appl1}

In this section, we illustrate the usefulness of the STEFA model and the IP-SVD algorithm on the Comprehensive Climate Dataset (CCDS) -- a collection of climate records of North America.
The dataset was compiled from five federal agencies sources by \cite{lozano2009spatial}\footnote{\url{http://www-bcf.usc.edu/~liu32/data/NA-1990-2002-Monthly.csv}}.
Specifically, we show that we can use the STEFA and IP-SVD to estimate interpretable loading functions, deal better with large noises and make more accurate predictions than the vanilla Tucker decomposition.

\begin{figure}[!hp]
\centering
\includegraphics[width=0.6\textwidth]{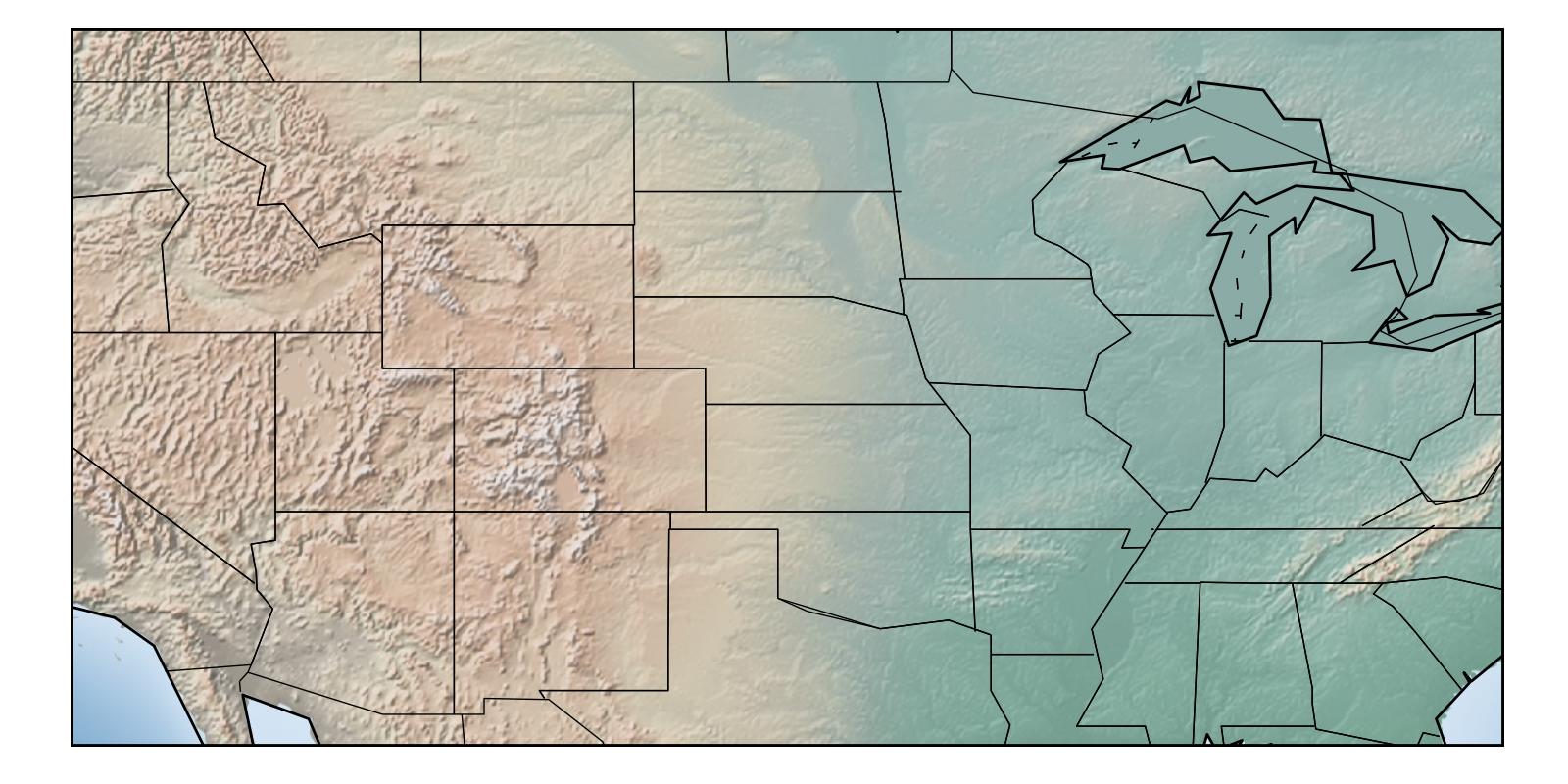}
\caption{The geological region for which the data is collected}\label{fig:climate-map}
\end{figure}

The data contains monthly observations of 17 climate variables from 1990 to 2001 on a $2.5 \times 2.5$ degree grid for latitudes in $(30.475, 50.475)$, and longitudes in $(-119.75, -79.75)$. 
Figure~\ref{fig:climate-map} plots the geological region which covers the majority of the continent of United States and the southern part of Canada.
The total number of observation locations is 125 and the whole time series spans from January, 1990 to December, 2001.
Due to the data quality,
we use only 16 measurements listed in Table \ref{table:CCDS_varlist} at each location and time point.
Thus, the dimensions our our dataset are 125 (locations) $\times$ 16 (variables) $\times$ 156 (time points).
Detailed information about data is given in \cite{lozano2009spatial}.

\begin{table}[htpb!]
    \centering
    \caption{Variables and data sources in the Comprehensive Climate Dataset (CCDS)}
    \label{table:CCDS_varlist}
    \resizebox{0.7\textwidth}{!}{%
        \begin{tabular}{l|c|l|c}
            \hline
            Variables (Short name) & Variable group & \multicolumn{1}{c|}{Type} & Source \\ \hline
            Methane (CH4) & $CH_4$ & \multirow{4}{*}{Greenhouse Gases} & \multirow{4}{*}{NOAA} \\
            Carbon-Dioxide (CO2) & $CO_2$ &  &  \\
            Hydrogen (H2) & $H_2$ &  &  \\
            Carbon-Monoxide (CO) & $CO$ &  &  \\ \hline
            Temperature (TMP) & TMP & \multirow{8}{*}{Climate} & \multirow{8}{*}{CRU} \\
            Temp Min (TMN) & TMP &  &  \\
            Temp Max (TMX) & TMP &  &  \\
            Precipitation (PRE) & PRE &  &  \\
            Vapor (VAP) & VAP &  &  \\
            Cloud Cover (CLD) & CLD &  &  \\
            Wet Days (WET) & WET &  &  \\
            Frost Days (FRS) & FRS &  &  \\ \hline
            Global Horizontal (GLO) & SOL & \multirow{4}{*}{Solar Radiation} & \multirow{4}{*}{NCDC} \\
            Direct Normal (DIR) & SOL &  &  \\
            Global Extraterrestrial (ETR) & SOL &  &  \\
            Direct Extraterrestrial (ETRN) & SOL &  &  \\ \hline
        \end{tabular}%
    }
\end{table}

We first focus on the spatial function structure of this data set.
The covariates $\bX \in \RR^{125\times2}$ of the spacial dimension contain the latitudes and longitudes of all sampling locations, which basically capture the spatial continuity of factor loadings on mode 1.
The semi-parametric form \eqref{eqn:tucker-covari-1} for this application is written as
\begin{equation} \label{eqn:appl1-model}
\calY = \calF \times_1 \paran{\bPhi_1(\bX)\bB_1 + \bR_1(\bX) + \bGamma_1} \times_2 \bA_2 \times_3 \bI + \calE.
\end{equation}
The first mode is the space dimension with loading matrix $\bA_1 = \bPhi_1(\bX)\bB_1 + \bR_1(\bX) + \bGamma_1$.
The second mode is the variable dimension with $\bA_2$ as the variable loading matrix.
The third mode is the time dimension which we do not compress.
So we use the identity matrix $\bI$ in place of $\bA_3$.
This is a matrix-variate factor model similar to  \cite{chen2020statistical} but incorporates covariate effects on the loading matrix in the spatial-mode. 
We normalized each time series to have a unit $\ell_2$ norm.

\paragraph{Climate variable and spatial factors.}
We use $R_1, R_2, R_3 = 6, 6, 156$ where the time mode is not compressed and the other two latent dimensions are chosen according to the literature \citep{lozano2009spatial,bahadori2014fast,chen2020multivariate}.
We use the Legendre basis functions of order 5 for $\bPhi_1(\bX)$ and number of basis $J = 11$. 
The slices of latent tensor factor $\bF_{r_1 : :}$, $r_1\in [6]$, correspond to six spatial factors and the slices $\bF_{: r_2 :}$, $r_2\in [6]$, correspond to the six climate variable factors. 
The meaning of the latent factors can be inferred from their corresponding variable loading matrix $\bA_2$ and spatial loading surfaces in $\bPhi_1(\bX)\bB_1$. 

Figure \ref{fig:appl-loading} (a) shows the heatmap of the varimax-rotated loading matrix $\bA_2$. 
It is clear that the corresponding first climate factor weighted mostly on the four greenhouse gases. 
Thus, the first climate factor can be interpreted as the greenhouse gas factor. 
Interestingly, this greenhouse gas factor also loads heavily on cloud cover (CLD), echoing with a recent scientific research on the observational evidence between greenhouse gas and cloud covers  \citep{ceppi2021observational}.
In a similar way, the second to sixth climate variable factor can be interpreted as temperature, precipitation (wet), frost, solar, and vapor factors, respectively. 
The top six climate factors explain approximately $82.26\%$ , $12.13\%$, $1.48\%$, $0.58\%$, $0.31\%$, $0.26\%$ of the variance along the second (climate variable) mode of the tensor.

\begin{figure}[htpb!]
    \centering
    \begin{subfigure}[b]{0.28\textwidth}
        \includegraphics[width=\textwidth]{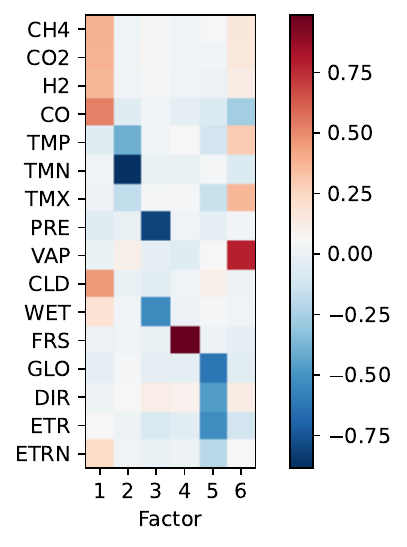}
        \caption{Variable Factors}
    \end{subfigure}
    \begin{subfigure}[b]{0.67\textwidth}
        \includegraphics[width=\textwidth]{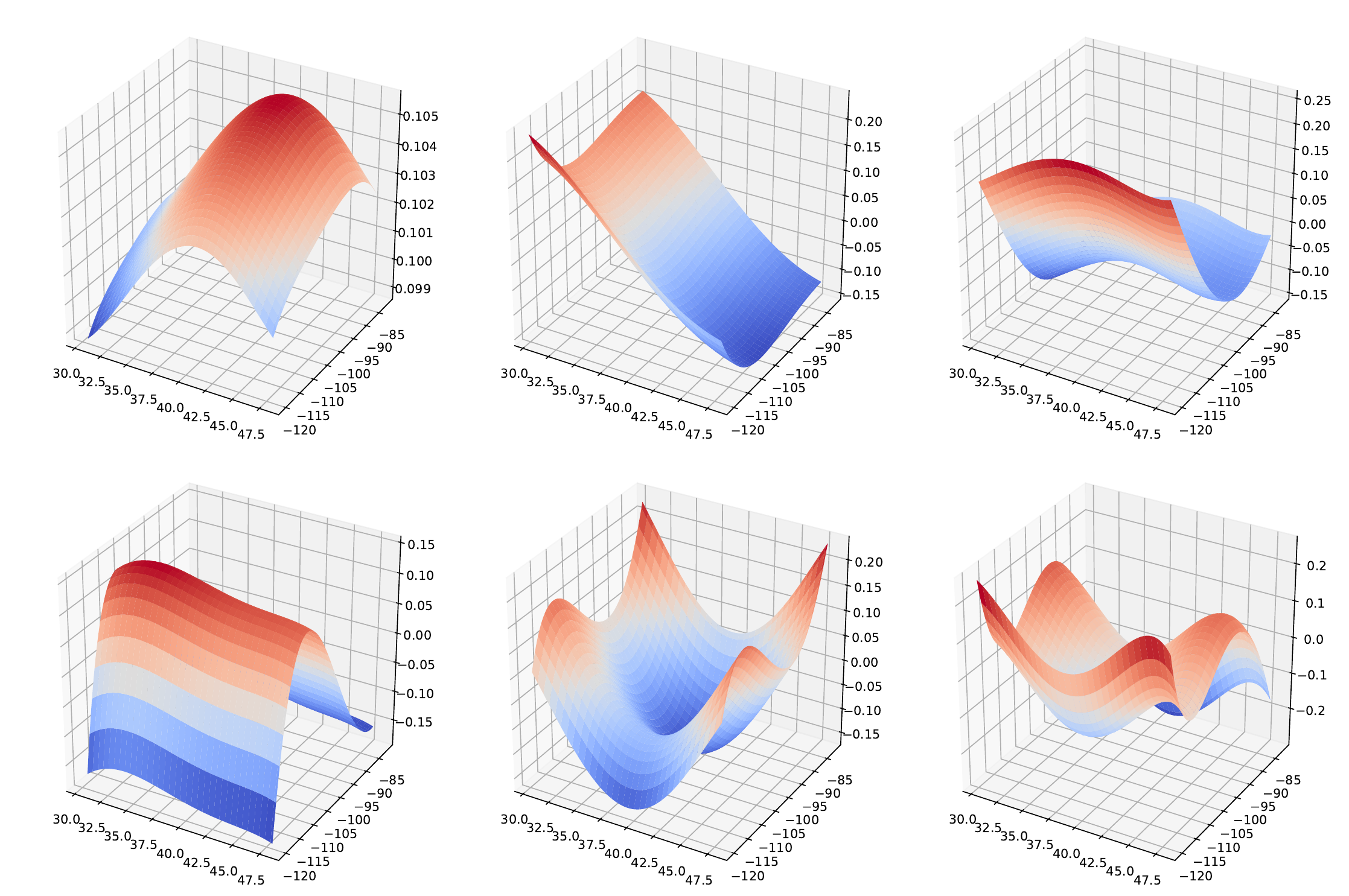}
        \caption{Spatial Factors}
    \end{subfigure}
    \caption{(a) Heat map plots the varimax-rotated $\hat\bA_2$. 
        The first six variable factors explain approximately $82.26\%$ , $12.13\%$, $1.48\%$, $0.58\%$, $0.31\%$, $0.26\%$ of the variance along the second mode of the tensor $\calY$. 
        (b) Six surfaces are the estimated space loading surfaces plotted from six columns of $\hat\bG_1(\bX)=\bPhi_1(\bX)\hat\bB_1$.
        From the top-left to the bottom right sub-figures correspond to the first to the sixth space loading functions with decreasing singular values.
        The coordinates of X and Y axis are aligned with the latitudes and longitudes in Figure~\ref{fig:climate-map}. 
    }
    \label{fig:appl-loading}
\end{figure}

Figure \ref{fig:appl-loading} (b) presents six estimated bi-variate spatial loading surfaces corresponding to the six columns of $\bPhi_1(\bX)\hat\bB_1$. 
The space loading surfaces captures the common spatial variances in 16 environmental variables and they are highly nonlinear.
More insights can be drawn by juxtaposing the discovered loading surfaces with the geological map in Figure~\ref{fig:climate-map} with aligned latitudes and longitudes.
The high value (red) region in the first loading surface corresponds to the Great Lakes region of U.S. and Canada, which was highly-populated and has a well-developed industry in the 90's. 
The second surface represents a south-to-north gradient and a coast-to-inland gradient. 
The third surface has high values in the mountain region of U.S. 
The discovered top three major loading surfaces have their sociological and geological correspondences.  Beyond those, the estimates are very noisy and interpretation gets hard.
The top sixth column of $\bPhi(\bX)_1\hat\bB_1$ explain approximately $93.16\%$ , $2.39\%$, $0.75\%$, $0.42\%$, $0.18\%$, $0.14\%$ of the variance along the first (spatial) mode of the tensor.  


\paragraph{Fitting real data with different noise levels.}
In this section, we compare the vanilla and projected Tucker decomposition by their performances in fitting signal with different levels of noise.
To generate different noise levels, we treat the estimated signal $\hat\calS_v$ and noise $\hat\calE_v$ from \textit{vanilla} Tucker decomposition as the true signal $\calS$ and noise $\calE$ and calibrate the real data with different noise amplifier $\alpha > 0$.
Specifically, the calibrated data is generated as $\calY = \hat\calS_v + \alpha \times \hat\calE_v$.
The setting $\alpha = 1$ corresponds to the original data.
We compare the relative mean square errors (ReMSE) of the signal estimator ${\rm ReMSE}_{\calS}={\|\calS - \hat\calS\|^2_F}/{\norm{\calS}^2_F}$ for vanilla and projected Tucker decomposition in Figure \ref{fig:appl-ReReSS}.
For the vanilla Tucker decomposition, we use the HOOI algorithm. 
For the projected Tucker decomposition, we use the same setting as previously, that is, we use the Legendre basis functions of order 5 for $\bPhi(\bX)$, number of basis $J = 11$ and latent dimensions $R_1, R_2, R_3 = 6, 6, 156$.
Two methods behave the same in the noiseless case where $\alpha=0$.
However, in the noisy setting where $\alpha > 0$, the IP-SVD outperforms the HOOI at all noise levels. 

\begin{figure}[htpb!]
    \centering
    \includegraphics[width=.5\textwidth]{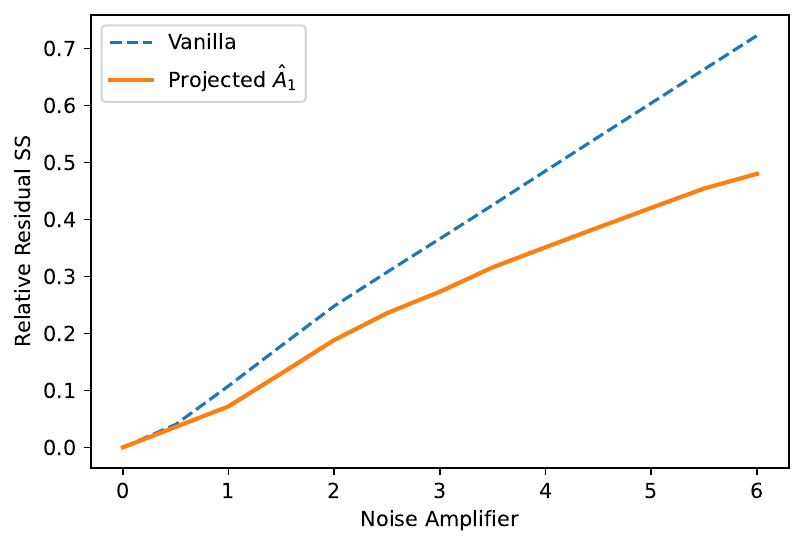}
    \caption{Relative mean square errors (ReMSE) by projected versus vanilla Tucker decomposition with different noise amplifiers.
        The relative residual SS of the signal part is defined as ${\|\calS - \hat\calS\|^2_F}/{\norm{\calS}^2_F}$ where $\hat\calS := \hat\calF\times_1\hat\bA_1\times_2\hat\bA_2$.
        The loading $\bA_1$, $\bA_2$ and factor $\calF$ are estimated by HOOI and IP-SVD, respectively, for vanilla and projected Tucker decomposition. 
        }
    \label{fig:appl-ReReSS}
\end{figure}

\paragraph{Spatial prediction.}
In this section, we compare the prediction performances of the methods based vanilla and projected Tucker decomposition.
The two prediction procedures are presented in Section~\ref{sec:pred}.
We randomly choose the training set to be $50\%$, $67\%$, and $75\%$ of the whole data set.
Table~\ref{table:appl-1-pred} shows the prediction errors, average over cross validations, of the two methods respectively.
It is clear that the STEFA model with projected Tucker decomposition outperforms the vanilla methods.
\begin{table}[htpb!]
    \centering
    \caption{Relative prediction error (averaged value by cross validation).
        For ease of display, the errors for Vanilla and Projected Tucker are reported as $100 \times$ the true value.}
    \label{table:appl-1-pred}
    \resizebox{0.5\textwidth}{!}{%
        \begin{tabular}{c|ccc}
            \hline
            Training set proportion & $50\%$ & $67\%$ & $75\%$ \\ \hline
            Vanilla & 3.52 & 3.48 & 3.05 \\
            Projected & 3.20& 3.23 & 3.01 \\ \hline
            Improvement & 9.0\% & 7.2\% & 1.3\% \\ \hline
        \end{tabular}%
    }
\end{table}

\paragraph{Temporal-mode compression.}
Now we consider fitting the real data with a more complex model where the mode corresponds to time is also compressed:
\begin{equation} \label{eqn:appl1-model2}
    \calY = \calF
    \times_1 \paran{\bPhi_1(\bX)\bB_1 + \bR_1(\bX) + \bGamma_1}
    \times_2 \bA_2
    \times_3 \paran{\bPhi_3(t)\bB_3 + \bR_3(t) + \bGamma_3} + \calE.
\end{equation}
For the space mode, we use the same setting as previously, that is, we use the Legendre basis functions of order 5 for $\bPhi_1(\bX)$ and the number of basis $J_1 = 11$.
For the time mode, we use the sinusoidal basis functions of order 12 for $\bPhi_3(t)$ and the number of basis $J_3 = 13$. 
Figure \ref{fig:time-loading} presents the first two columns of $\bPhi_3(t)\bB_3$ which explains approximately $80.69\%$ and $0.14\%$ of the variance along time mode of the tensor.
Each column of the loading matrix $\bPhi_3(\bt)\hat\bB_3$ can be interpreted from its temporal pattern.
The first time loading corresponds to the temporal mean since it is almost flat over time.
The second time loading corresponds to a linear trend component. 
This trend coincides with the annual greenhouse gas emission data from U.S. environment protection agency\footnote{https://cfpub.epa.gov/ghgdata/inventoryexplorer}, where the greenhouse gas emission has an overall increasing trend from 1992 to 2002 with local peaks around 1995 and 2000.
The other time loading dimensions are less prominent as they account for a small portion of variations. 
As a result, their corresponding interpretations are not obvious and we omit their plots here. 
 
\begin{figure}[htpb!]
    \centering
    \includegraphics[width=.9\textwidth]{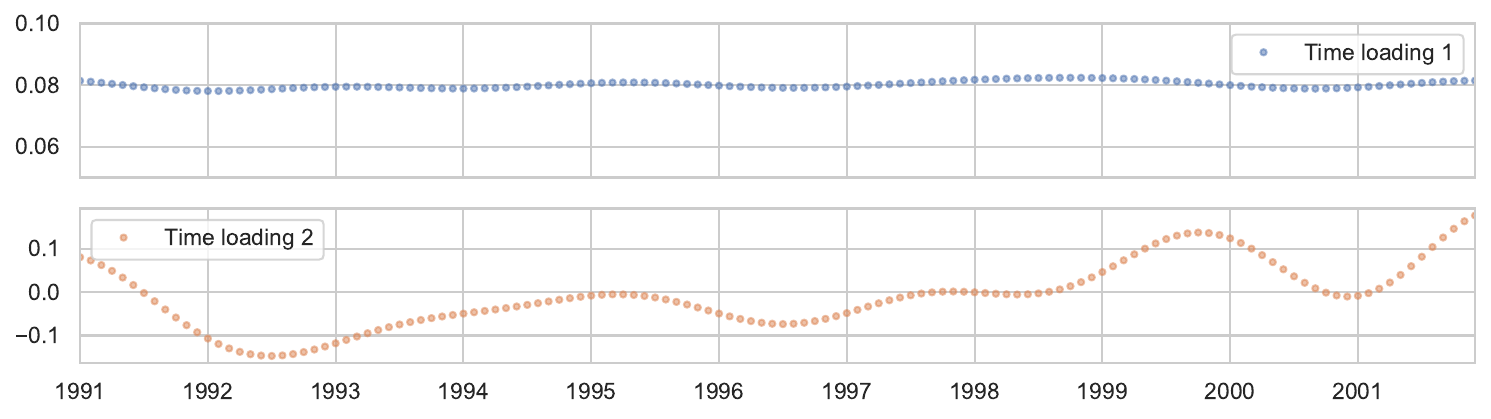}
    \caption{Top two functions of time that corresponds to the first two columns of $\bPhi_3(t)\bB_3$.}
    \label{fig:time-loading}
\end{figure}

\subsection{Human Brain Connection Data} \label{sec:hcp}

We illustrate another application of the STEFA model and IP-SVD to the human brain connection data \citep{desikan2006automated}.
This Human Connectome Project (HCP) dataset consists of brain structural networks collected from 136  individuals.
Each brain network is represented as a $68 \times 68$ binary matrix, each entry of which encodes the presence or absence of a fiber connection between $68$ brain regions.
Thus, the final observation $\calY$ is of dimension $68 \times 68 \times 136$.
Associated with $136$ individuals, there are $573$ features including ages, genders, and various measurements of their brains.
This dataset has been used in \cite{hu2021generalized} for tensor regression and it is available in the \proglang{R} package \pkg{tensorregress}.
We consider the instance of the STEFA model with $\bA_3 = \bPhi_3(\bX) \bB_3 + \bR_3(\bX) + \bGamma_3$.
The covariate $\bX$ contains five features: gender (65 females vs. 71 males), age 22-25 ($n=35$), age 26-30 ($n=58$), and age 31+ ($n=43$).
These categorical variables are coded using sum-to-zero contrasts and the lower rank is set as $(10, 10, 4)$, the same as those set in  \cite{hu2021generalized}.  As an illustration, we choose $\bPhi_3(\bX)$ generated by polynomial basis of order 1, which is a similar linear setting as that in \cite{hu2021generalized} with identity link function. 

\cite{hu2021generalized} consider a similar setting for tensor regression.
However, there are two major differences.
First, they consider a generalized linear model (GLM) with a low-Tucker-rank coefficient tensor that predict $\cal Y$ with given observations $\bX$ and their estimation is based on the maximum likelihood estimation. 
On the contrary, the STEFA model aims to discover the relationships between entries in $\calY$ (with the help of covariates $\bX$ when available) without using any distributional assumption. 
This relationship between entries in $\calY$ can also be used to predict values by interpolation as proposed in Section \ref{sec:pred}. 
Second, the STEFA model is comparable to MMC tensor regression only when the GLM linkage function is linear. 
In such setting, the form of the tensor regression model in \cite{hu2021generalized} is equivalent to assuming $\bA_3 = \bPhi_3(\bX) \bB_3$ with $\bPhi_3(\bX)=\bX$ in the STEFA model. 
It has two limitations: only linear components of $\bX$ is considered and the non-parametric residuals independent with $\bX$ are ignored.
In contrast, the STEFA model and the IP-SVD estimate $\bA_3$ which consists of covariate-relevant $\bPhi_3(\bX) \bB_3$, the residual component $\bR(\bX)$, as well as the orthogonal component $\bGamma_3$. 
We compare the relative mean squared error ${\rm ReMSE}_{\calY} = {\norm{\hat\calY - \calY}_F^2}/{\norm{\calY}_F^2}$  of the STEFA and MMC tensor regression with the identity link in Table \ref{table:appl-2-comp-mmc}. 
The relative mean squared error of the STEFA is smaller under all choices of basis functions. 
In fact, the STEFA model complements the work in \cite{hu2021generalized} since the orthogonal $\bGamma_3$ can also be included in their tensor regression models. 

\begin{table}[htpb!]
    \centering
    \caption{Relative mean squared error ${\rm ReMSE}_{\calY}= {\norm{\hat\calY - \calY}_F^2}/{\norm{\calY}_F^2}$ of the MMC with identity link and the STEFA. 
    Each columns corresponds to a type of basis function with its order in the parentheses.  
    The MMC and the STEFA use the same basis function in $\bPhi(\bX)$. 
    }
    \label{table:appl-2-comp-mmc}
    \resizebox{0.7\textwidth}{!}{%
        \begin{tabular}{c|ccc}
            \hline
             & Polynomials (1) & Polynomials (3)  & Legendre (5) \\ \hline
            MMC (Identity link) & $39.6\%$  & $39.5\%$ &$39.4\%$ \\ \hline
            STEFA& $38.7\%$ & $38.4\%$ &  $38.3\%$  \\ \hline
            Improvement & $2.3\%$ & $2.8\%$ & $2.8\%$  \\ \hline
        \end{tabular}%
    }
\end{table}

The covariate-relevant $\bPhi_3(\bX) \bB_3$ is determined by covariates specified by domain experts. 
Researchers may be curious to identify features affecting brain connectivity other than those already known in the field. 
Now we show that the residual  $\hat\bGamma_3  = \hat\bA_3 - \hat\bPhi_3(\bX)\hat\bB_3$ obtained by STEFA can be used to discover several features other than gender and age.
Analogous to the interpretation that the rows of the $136\times 4$ loading matrix $\hat\bA_3$ can be reviewed as the low-rank representation of $136$ subjects in the latent factor space, 
matrices $\hat\bGamma_3$ and $\hat\bPhi_3(\bX)\hat\bB_3$ can be interpreted, respectively, as covariate-independent and covariate-dependent low-rank representations. 
The idea to identify important features left in the residual component is to first use $\hat\bGamma_3$ in spectral clustering to divide $136$ subjects into four groups, and then use recursive feature elimination (RFE) to identify the top four important features that differentiate these four groups. 
The identified top four important features that are disparate across groups are the average thickness of the right transverse temporal gyri 
which has been shown to be correlated with human acoustic processing \citep{warrier2009relating}, 
the volumn of accumbens area 
which is a key structure in mediating emotional and motivation processing, modulating reward and pleasure processing, and serving a key limbic-motor interface \citep{cohen2009good,salgado2015nucleus},
the unadjusted negative emotion affect related to sadness, fear, and anger, 
and a personality raw score on being active or not. 
Section \ref{append:HCP} provides another illustration of using the STEFA and IP-SVD for explanatory data analysis to partitioning the brain connectivity according to the covariate-relevant loading $\bPhi_3(\bX)\bB_3$. 
These interesting discoveries from explanatory data analysis can be used as good starting points for the following more rigorous scientific researches. 

\section{Discussion} \label{sec:summ}

This paper introduces a high-dimensional Semiparametric TEnsor FActor (STEFA) model with nonparametric loading functions that depend on a few observed covariates.
This model is motivated by the fact that observed variables can partially explain the factor loadings, which helps to increase the accuracy of estimation and the interpretability of results.
We propose a computationally efficient algorithm IP-SVD to estimate the unknown tensor factor, loadings, and the latent dimensions.
The advantages of IP-SVD are two-fold.
First, unlike HOOI which iterates in the ambient dimension, IP-SVD finds the principal components in the covariate-related subspace whose dimension can be significantly smaller.
As a result, IP-SVD requires weaker SNR conditions for convergence.
Secondly, the projection also reduces the effect dimension size of stochastic noise and thus IP-SVD yields an estimate of latent factors with faster convergence rates.

While tensor data is everywhere in the physical world, statistical analysis for tensor data is still challenging. There are several interesting topics for future research.
First, it is important to develop non-parametric tests on whether observed relevant covariates have explaining powers on the loadings and whether they fully explain the loadings.
However, under the tensor decomposition setting, this is more challenging than a straightforward extension from \cite{fan2016projected}.
Second, we mentioned briefly that, when there are multiple observations, one can apply IP-SVD on the sample covariance tensor.
However, a more precise algorithm is needed.
Last but not the least, it is of great need to develop new methods to use STEFA in tensor regression or other tensor data related applications.

\section*{Acknowledgments}
We express our gratitude to the referees and the editors for their invaluable feedback, which greatly enhanced the quality of this paper.

\noindent
\textit{Conflict of interest:} None declared.

\section*{Funding}
Chen's research is supported in part by NSF Grant DMS-1803241.
Xia's research is supported in part by Hong Kong RGC ECS Grant 26302019 and GRF grant 16303320.  Fan's research was partially supported by NSF Grants DMS-1662139, DMS-1712591, DMS-2210833,  and ONR grant N00014-22-1-2340.

\section*{Data, Source Code, and Supplementary Material}
The Comprehensive Climate Dataset (CCDS) is available at {\url{http://www-bcf.usc.edu/~liu32/data/NA-1990-2002-Monthly.csv}.
The source code is available at \url{https://github.com/ElynnCC/STEFA-Code}.
Supplementary material is available online at {\it Journal of the Royal Statistical Society: Series B.}

\bibliographystyle{plainnat}
\bibliography{bibliography}       

\listoffigures

\end{document}


%
%

\begin{appendices}
\begin{center}
{\Large Supplementary Material of ``Semiparametric Tensor Factor Analysis by Iteratively Projected SVD''}

{Elynn Y. Chen$^1$, Dong Xia$^2$, Chencheng Cai$^3$ and Jianqing Fan$^{4,5}$\\ \normalsize
$^1$New York University, 
$^2$Hong Kong University of Science and Technology\\ \normalsize
$^3$Washington State University, $^4$ Fudan University, $^5$Princeton University}
\end{center}

\section{Major Theoretical Proofs}


\begin{proof}[Proof of Lemma~\ref{lem:tfm-ic}]
	Let $\calS = \tilde\calF\times_1\tilde\bA_1\times_2\cdots\times_M\tilde \bA_M$ be a Tucker decomposition of $\calS$ such that $\tilde\calF\in\mathbb R^{R_1\times \cdots\times R_M}$ is the core tensor, and for $m\in[M]$, $\tilde\bA_m\in\mathbb R^{I_m\times R_m}$ is the $m$-mode loading matrix satisfying Assumption~\ref{assum:tfm-ic-1}(i).
	
	Clearly, $\calS = \calF\times_1\bA_1\times_2\cdots\times_M\bA_M$ is a valid Tucker decomposition satisfying Assumption~\ref{assum:tfm-ic-1}(i) if and only if there exist orthogonal matrices $\bH_m\in\mathbb O^{R_m\times R_m}$ for all $m\in[M]$ such that $\bA_m = \tilde\bA_m\bH_m$ and $\calF = \tilde\calF\times_1\bH_1^\top\times_2\cdots\times_m\bH_M^\top$. Moreover, then $\calM_m(\calS)\calM_m(\calS)^\top/I_m$ has the same singular values with $\calM_m(\calF)\calM_m(\calF)^\top$

	Now we are in the place to prove the lemma. Recall that any Tucker decomposition of $\calS$ satisfying Assumption~\ref{assum:tfm-ic-1}(i) are indexed by the set of orthogonal matrices $(\bH_1,\cdots, \bH_M)$ in reference to the decomposition $\calS = \tilde\calF\times_1\tilde\bA_1\times_2\cdots\times_M\tilde \bA_M$. The corresponding core tensor $\calF$ is $\calF = \tilde \calF\times_1\bH_1^\top\times_2\cdots\times_M\bH_M^\top$. The mode-$m$ matricization of $\calF$ is $\calM_m(\calF) = \bH_m^\top\calM_m(\tilde\calF)\left[\bigotimes_{l\in[M], l\neq m}\bH_l^\top\right]$. Suppose for some $\bH_m$, Assumption~\ref{assum:tfm-ic-1}(i) is satisfied such that $\calM_m(\calF)\calM_m(\calF)^\top = \bH_m^\top \calM_m(\tilde\calF)\calM_m(\tilde\calF)^\top\bH_m=\bD_m$ for some diagonal matrix $\bD_m$ with non-zero decreasing diagonal entries. Then the diagonal entries of $\bD_m$ are the eigenvalues of $\calM_m(\tilde\calF)\calM_m(\tilde\calF)^\top$ and the columns in $\bH_m$ are the corresponding eigenvectors, because of the equality
	$\calM_m(\tilde\calF)\calM_m(\tilde\calF)^\top\bH_m=\bH_m\bD_m$.
	Note that $\bD_m$ has the same singular values with $\calM_m(\calS)\calM_m(\calS)^\top$. As a result, when the singular values of $\calM_m(\calS)\calM_m(\calS)^\top$ are distinct, the eigenvalues and eigenvectors of $\calM_m(\tilde\calF)\calM_m(\tilde\calF)^\top$ can be uniquely identified (up to a global sign), resulting in an unique $\bH_m$. Here, the uniqueness is up to a column-wise sign of $\bH_m$.
	
	In conclusion, starting from an arbitrary Tucker decomposition $\calS = \tilde\calF\times_1\tilde\bA_1\times_2\cdots\times_M\tilde \bA_M$ satisfying Assumption~\ref{assum:tfm-ic-1}(i), by choosing the columns of $\bH_m$ to be the eigenvectors of $\calM_m(\tilde\calF)\calM_m(\tilde\calF)^\top$ in descending order of eigenvalues, the Tucker decomposition with $\calF=\tilde\calF\times_1\bH_1^\top\times_2\cdots\times_M\bH_M^\top, \bA_1=\tilde \bA_1\bH_1,\cdots, \bA_M=\tilde \bA_M\bH_M$ is the unique Tucker decomposition satisfying both Assumptions 1(i) and 1(ii) when the eigenvalues of
	$\calM_m(\calS)\calM_m(\calS)^\top$ are distinct for all $m\in[M]$.
\end{proof}

\begin{proof}[Proof of Lemma~\ref{lem:init}] We prove the initialization error and convergence of IP-SVD separately. Without loss of generality, we assume $\EE \eps_{\omega}^2=1$ for all $\omega\in [I_1]\times [I_2]\times [I_3]$.
	
	We begin with the upper bound of the remainder term $\bR_m(\bX_m)$. By Assumption~\ref{assump:sieve}, for each function $m\in[M]$ and $1\leq r_m\leq R_m$, $|g_{m,r_m}(\bx_{i_m})-\bb_{m,r_m}^{\top}\bphi_m(\bx_{i_m})|=O(J_m^{-\tau/2})$ which bounds the $(i_m,r_m)$-th entry of $\bR_m(\bX_m)$. Therefore, a simple fact is
	\begin{equation}\label{eq:RjX}
		\|\bR_m(\bX_m)\|_{\rm F}^2/I_m=O\big(R_m\cdot J_m^{-\tau}\big)
	\end{equation}
	for all $m\in[M]$.
	
	{\bf Initialization error}.
	Without loss of generality, we only prove the upper bound of $\|\tilde \bG_1^{(0)}\tilde \bG_1^{(0)\top}-\bG_1\bG_1^{\top}\|_{\rm F}$. Recall that
	$
	\bG_1=\bPhi_1(\bX_1)\bB_1+\bR_1(\bX_1)
	$
	where, by the definition of $\bPhi_1(\bX_1)$, we have $\bPhi_1(\bX_1)^{\top}\bGamma_1={\bf 0}$. Let $\check{\bG}_1/\sqrt{I_1}$ denotes the top-$R_1$ left singular vectors of $\bPhi_1(\bX_1)\bB_1$. By Condition~(\ref{eq:RjX}) and Davis-Kahan theorem \citep{davis1970rotation} or spectral perturbation formula \citep{xia2019normal},
	\begin{align}\label{eq:check_err}
		\|\check{\bG}_1\check{\bG}_1^{\top}-\bG_1\bG_1^{\top}\|_{\rm F}/I_1=O\big(\sqrt{R_1}\cdot J_1^{-\tau/2}\big)
	\end{align}
	where we used the fact that $\sigma_{R_1}\big(\bPhi_1(\bX_1)\bB_1/\sqrt{I_1}\big)\geq 1-O(\sqrt{R_1}\cdot J_1^{-\tau/2})\geq 1/2$ and $\|\bR_1(\bX_1)\|_{\rm F}/\sqrt{I_1}=O(\sqrt{R_1}\cdot J_1^{-\tau/2})$.
	
	Recall that
	$
	\widetilde\calY=\calF\times_1(\bP_1\bG_1)\times_2 (\bP_2\bG_2)\times_3 (\bP_3\bG_3)+\calE\times_1 \bP_1\times_2 \bP_2\times_3 \bP_3
	$
	and as a result
	$$
	\calM_1(\widetilde\calY)=\bP_1\bG_1\calM_1(\calF)\big((\bP_2\bG_2)\otimes (\bP_3\bG_3)\big)^{\top}+\bP_1\calM_1(\calE)(\bP_2\otimes\bP_3)^{\top}
	$$
	where we denote $\otimes$ the kronecker product.
	The projector matrix $\bP_m\in\RR^{I_m\times I_m}$ with $\rank(\bP_m)=J_m$. Denote the eigen-decomposition of $\bP_m$ by
	$
	\bP_m=\bU_m\bU_m^{\top}\ {\rm where}\  \bU_m^{\top}\bU_m=\bI_{J_m}.
	$
	Therefore, we write
	$$
	\bP_1\calM_1(\calE)(\bP_2\otimes \bP_3)=\bU_1\big(\bU_1^{\top}\calM_1(\calE)(\bU_2\otimes \bU_3)\big)(\bU_2\otimes \bU_3)^{\top}.
	$$
	In addition, we write
	\begin{align*}
		\bP_1\bG_1\calM_1(\calF)\big((\bP_2\bG_2)\otimes (\bP_3\bG_3)\big)^{\top}=\bP_1\bPhi_1(\bX_1)\bB_1\calM_1(\calF)\big((\bP_2\bG_2)\otimes  (\bP_3\bG_3)\big)^{\top}\\
		+\bP_1\bR_1(\bX_1)\calM_1(\calF)\big((\bP_2\bG_2)\otimes (\bP_3\bG_3)\big)^{\top}
	\end{align*}
	where the left singular space of the first matrix is the same to column space of $\check\bG_1$. Denote the top-$R_1$ left singular vectors of $\bP_1\bG_1\calM_1(\calF)\big((\bP_2\bG_2)\otimes (\bP_3\bG_3)\big)^{\top}$ by $\mathring{\bG}_1/\sqrt{I_1}$. Again by Davis-Kahan theorem \citep{davis1970rotation} or spectral perturbation formula \citep{xia2019normal}, we have
	\begin{align}\label{eq:ring-check-err}
		\|\mathring{\bG}_1\mathring{\bG}_1^\top-\check{\bG}_1\check{\bG}_1^{\top}\|_{\rm F}/I_1=O(\kappa_0\sqrt{R_1}\cdot J_1^{-\tau/2})
	\end{align}
	where $\kappa_0$ is $\calF$'s condition number and we used the facts
	\begin{align*}
		\sigma_{R_1}\Big(\bP_1\bPhi_1(\bX_1)\bB_1\calM_1(\calF)\big((\bP_2\bG_2)\otimes (\bP_3\bG_3)\big)^{\top}\Big)\geq& \lambda_{\min}\sqrt{I_2I_3}\cdot \sigma_{R_1}\big(\bPhi_1(\bX_1)\bB_1\big)\\
		&\geq \lambda_{\min}\sqrt{I_1I_2I_3}/2
	\end{align*}
	and
	\begin{align*}
		\Big\|\bP_1\bR_1(\bX_1)\calM_1(\calF)\big((\bP_2\bG_2)\otimes (\bP_3\bG_3)\big)^{\top}\Big\|_{\rm F}=&O\big(\kappa_0\lambda_{\min}\sqrt{I_2I_3}\cdot \|\bR_1(\bX_1)\|_{\rm F}\big)\\
		&=O(\kappa_0\lambda_{\min}\sqrt{R_1I_1I_2I_3}\cdot J_1^{-\tau/2}).
	\end{align*}
	For notational simplicity, we write $\calM_1(\widetilde{\calY})=\bA_1+\bZ_1$ where
	$\bA_1=\bP_1\bG_1\calM_1(\calF)\big((\bP_2\bG_2)\otimes (\bP_3\bG_3)\big)^{\top}$ and $\bZ_1=\bP_1\calM_1(\calE)(\bP_2\otimes\bP_3)^{\top}$.
	Then, $\mathring{\bG}_1/\sqrt{I_1}$ are the top-$R_1$ left singular vectors of $\bA_1$ and are also the top-$R_1$ eigenvectors of $\bA_1\bA_1^{\top}$. Since $\tilde \bG_1^{(0)}/\sqrt{I_1}$ are the top-$R_1$ left singular vectors of $\calM_1(\widetilde{\calY})$, they are also the top-$R_1$ eigenvectors of $\calM_1(\widetilde{\calY})\calM_1^{\top}(\widetilde{\calY})$ whcih can be written as
	$
	\calM_1(\widetilde{\calY})\calM_1^{\top}(\widetilde{\calY})=\bA_1\bA_1^{\top}+\bA_1\bZ_1^{\top}+\bZ_1\bA_1^{\top}+\bZ_1\bZ_1^{\top}.
	$
	Observe that
	$$
	\bZ_1\bZ_1^{\top}=\bU_1 \big(\bU_1^{\top}\calM_1(\calE)(\bU_2\otimes \bU_3)\big)\big(\bU_1^{\top}\calM_1(\calE)(\bU_2\otimes \bU_3)\big)^{\top}\bU_1^{\top}.
	$$
	Then, we can write
	\begin{align*}
		\calM_1(\widetilde{\calY})\calM_1^{\top}(\widetilde{\calY})=\bA_1\bA_1^{\top}+J_2J_3\bP_1+\bA_1\bZ_1^{\top}+\bZ_1\bA_1^{\top}+\bZ_1\bZ_1^{\top}-J_2J_3\bP_1.
	\end{align*}
	By definition, the column space of $\mathring \bG_1$ is a subspace of the column space of $\bP_1$, and $\calP_{\mathring{\bG}_1}\bA_1\bA_1^{\top}\calP_{\mathring{\bG}_1}=\bA_1\bA_1^{\top}$ where
	$\calP_{\mathring{\bG}_1}=\mathring{\bG}_1\mathring{\bG}_1^{\top}/I_1$ denotes the orthogonal projection onto the column space of $\mathring{\bG}_1$.
	
	We write
	\begin{align}\label{eq:M1yM1y}
		\calM_1(\widetilde{\calY})\calM_1^{\top}(\widetilde{\calY})=\underbrace{\big(\bA_1\bA_1+J_2J_3\bP_1\big)}_{\bM}
		+(\bA_1\bZ_1^{\top}+\bZ_1\bA_1^{\top})
		+\bZ_1\bZ_1^{\top}-J_2J_3\bP_1.
	\end{align}
	Clearly, the top-$R_1$ left singular space of $\bM$ is the column space of $\mathring{\bG}_1$ and $\sigma_{R_1}(\bM)-\sigma_{R_1+1}(\bM)=\sigma_{R_1}(\bA_1\bA_1^{\top})\geq \lambda_{\min}^2\cdot I_1I_2I_3$. 
	
	The upper bounds on the spectral norm $\|\bA_1\bZ_1^{\top}+\bZ_1\bA_1^{\top}\|$ and $\|\bZ_1\bZ_1^{\top}-J_2J_3\bP_1\|$ are due to the following lemma whose proof is postponed to Appendix~\ref{sec:app_technical}.
	
	\begin{lemma}\label{lem:AZbound}
		Suppose that Assumption~\ref{assump:noise} holds. There exists an absolute constant $C_1>0$ such that with probability at least $1-2I_1^{-2}$,
		$$
		\|\bA_1\bZ_1^{\top}\|\leq C_4\kappa_0\lambda_{\min}(I_1I_2I_3)^{1/2}\cdot\sqrt{J_1}\log^2 I_1.
		$$
	\end{lemma}
	
	By Lemma~\ref{lem:AZbound}, the following bound holds with probability at least $1-2I_1^{-2}$
	$$
	\|\bA_1\bZ_1^{\top}+\bZ_1\bA_1^{\top}\|\leq C_4\kappa_0\lambda_{\min}(I_1I_2I_3)^{1/2}\cdot\sqrt{J_1}\log^2 I_1.
	$$
	Observe that
	\begin{align*}
		\bZ_1\bZ_1^{\top}&-J_2J_3\bP_1\\
		=&\bU_1\Big(\big(\bU_1^{\top}\calM_1(\calE)(\bU_2\otimes \bU_3)\big)\big(\bU_1^{\top}\calM_1(\calE)(\bU_2\otimes \bU_3)\big)^{\top}-J_2J_3\bI_{J_1}\Big)\bU_1^{\top}.
	\end{align*}
	\begin{lemma}\label{lem:ZZ_bound}
		Suppose that Assumption~\ref{assump:noise} holds. There exist absolute constants $C_3,C_4>0$ so that
		$$
		\PP\big(\|\bZ_1\bZ_1^{\top}-J_2J_3\bP_1\|\geq C_3\sqrt{J_1J_2J_3}\log^4I_1+C_4J_1\log^{5/2}I_1\big)\leq 5I_1^{-2}.
		$$
	\end{lemma}
	By Lemma~\ref{lem:ZZ_bound}, with probability at least $1-5I_1^{-2}$,
	$$
	\big\|\bZ_1\bZ_1^{\top}-J_2J_3\bP_1\big\|\leq C_3\sqrt{J_1J_2J_3}\log^4I_1+C_4J_1\log^{5/2}I_1.
	$$
	We now continue from (\ref{eq:M1yM1y}). By Davis-Kahan theorem \citep{davis1970rotation} and Lemma~\ref{lem:AZbound} and \ref{lem:ZZ_bound}, we get with probability at least $1-7I_1^{-2}$ that
	\begin{align*}
		\big\| \mathring{\bG}_1\mathring{\bG}_1^{\top}-\widetilde{\bG}_1^{(0)}\widetilde{\bG}_1^{(0)\top}\big\|_{\rm F}&\leq C_3\frac{\kappa_0\sqrt{R_1J_1}\log^2I_1}{\lambda_{\min}\sqrt{I_1I_2I_3}}
		&+C_4R_1^{1/2}\frac{(J_1J_2J_3)^{1/2}\log^{4}I_1+J_1\log^{5/2}I_1}{\lambda_{\min}^2I_1I_2I_3}
	\end{align*}
	for some absolute constants $C_3,C_4>0$. Denote the above event by $\mathfrak{E}_1$.
	Together with (\ref{eq:check_err}) and (\ref{eq:ring-check-err}), we get on event $\mathfrak{E}_1$ that
	\begin{align}
		\|\bG_1\bG_1^{\top}-&\widetilde{\bG}_1^{(0)}\widetilde{\bG}_1^{(0)\top}\|_{\rm F}/I_1\notag\\
		\leq&C_3'\left(\frac{\kappa_0\sqrt{R_1J_1}\log^2I_1}{\lambda_{\min}\sqrt{I_1I_2I_3}}+\frac{\sqrt{R_1J_1(J_1\vee J_2J_3)}\log^4I_1}{\lambda_{\min}^2I_1I_2I_3}+\kappa_0\sqrt{R_1}\cdot J_1^{-\tau/2}\right).\label{eq:check-init-err}
	\end{align}
	In view of (\ref{eq:check-init-err}), the initialization $\widetilde{\bG}_1^{(0)}$ is close to the truth as long as
	\begin{align*}
		\lambda_{\min}\sqrt{I_1I_2I_3}\geq C_1'\left(\kappa_0\sqrt{R_1J_1}\log^2I_1+\big(R_1J_1(J_1\vee J_2J_3)\big)^{1/4}\log^2I_1\right)
	\end{align*}
	and $J_1\gg \kappa_0^{2/\tau}R_1^{1/\tau}$.
	The proof is completed by assuming $\kappa_0=O(1)$ and $J_1\asymp J_2\asymp J_3$. Assuming warm initializations and iterates for $\widetilde{\bG}_m^{(t-1)}$, we prove the contraction property for $\widetilde{\bG}_m^{(t)}$.
	
	{\bf IP-SVD iterations}. Without loss of generality, we fix an integer value of $t$ and prove the contraction inequality (\ref{eq:pi_rate}) for $m=1$. For notation simplicity, we denote
	$$
	{\rm Err}_t=\max_{m=1,2,3} \|\tilde \bG_m^{(t)}\tilde \bG_m^{(t)\top}-\bG_m\bG_m^{\top}\|_{\rm F}/I_m.
	$$
	
	By projected power iteration in Section~\ref{sec:est}, the scaled singular vectors $\tilde \bG_1^{(t)}$ is obtained by
	$$
	\tilde \bG_1^{(t)}/\sqrt{I_1}={\rm SVD}_{R_1}\Big(\bP_1\calM_1\Big(\calY\times_2\tilde \bG_2^{(t-1)\top}\times_3 \tilde\bG_3^{(t-1)\top}\Big)\Big)
	$$
	Recall $\bA_m=\bG_m+\bGamma_m$ for $m=1,2,3$ and
	$$
	\calY=\underbrace{\calF\times_1(\bG_1+\bGamma_1)\times_2 (\bG_2+\bGamma_2)\times_3 (\bG_3+\bGamma_3)}_{\calS}+\calE.
	$$
	We then write
	\begin{align}\label{eq:P1Y1G2G3}
		\bP_1\calM_1\Big(\calY&\times_2 \tilde\bG_2^{(t-1)\top}\times_3 \tilde \bG_3^{(t-1)\top}\Big)\notag\\
		=& \bP_1\calM_1\Big(\calS\times_2 \tilde\bG_2^{(t-1)\top}\times_3 \tilde \bG_3^{(t-1)\top}\Big)+\bP_1\calM_1\Big(\calE\times_2 \tilde\bG_2^{(t-1)\top}\times_3 \tilde \bG_3^{(t-1)\top}\Big).
	\end{align}
	By the fact $\bP_1\bGamma_1={\bf 0}$, we obtain
	\begin{align*}
		\bP_1\calM_1\Big(\calS\times_2 \tilde\bG_2^{(t-1)\top}\times_3 \tilde \bG_3^{(t-1)\top}\Big)=\bP_1\bG_1\calM_1(\calF)\big((\bA_2^{\top}\tilde\bG_2^{(t-1)})\otimes (\bA_3^{\top}\tilde\bG_3^{(t-1)})\big).
	\end{align*}
	Observe that $\bP_1\bG_1=\bPhi(\bX_1)\bB_1+\bP_1\bR_1(\bX_1)$. By Condition (\ref{eq:RjX}), we get
	$$
	\sigma_{\min}\big(\bP_1\bG_1/\sqrt{I_1}\big)\geq \sigma_{\min}\big(\bPhi_1(\bX_1)\bB_1/\sqrt{I_1}\big)-O(\sqrt{R_1}\cdot J_1^{-\tau/2})\geq 1-O(\sqrt{R}_1\cdot J_1^{-\tau/2}).
	$$
	Recall that the column space of $\tilde \bG_m^{(t-1)}$ is a subspace of $\bPhi_m(\bX_m)$ for all $m=1,2,3$, implying that
	$
	\bA_m^{\top}\tilde\bG_m^{(t-1)}=\bG_m^{\top}\tilde\bG_m^{(t-1)}
	$
	and as a result
	$$
	\sigma_{\min}\big(\bA_m^{\top}\tilde\bG_m^{(t-1)}\big)=\sigma_{\min}\big(\bG_m^{\top}\tilde\bG_m^{(t-1)}\big)\geq I_m \sqrt{1-\|\tilde \bG_m^{(t-1)}\tilde \bG_m^{(t-1)\top}-\bG_m\bG_m^{\top}\|/I_m}\geq \sqrt{2}I_m/2
	$$
	where the last inequality is due to the fact $\|\tilde \bG_m^{(t-1)}\tilde \bG_m^{(t-1)\top}-\bG_m\bG_m^{\top}\|/I_m\leq 1/2$ which holds as long as the conditions of Lemma~\ref{lem:init} hold and initializations are warm in that $\|\tilde\bG_m^{(0)}\tilde\bG_m^{(0)\top}-\bG_m\bG_m^{\top}\|/I_m\leq 1/2$ for all $m\in[M]$. Therefore, we conclude that
	$$
	\sigma_{\min}\Big(\bP_1\calM_1\big(\calS\times_2 \tilde\bG_2^{(t-1)\top}\times_3 \tilde \bG_3^{(t-1)\top}\big)\Big)\geq \frac{\sqrt{I_1}I_2I_3}{3}\cdot \sigma_{R_1}\big(\calM_1(\calF)\big)\geq \frac{\lambda_{\min}\sqrt{I_1}I_2I_3}{3}.
	$$
	We now bound the operator norm of $\bP_1\calM_1\big(\calE\times_2 \tilde\bG_2^{(t-1)\top}\times_3 \tilde \bG_3^{(t-1)\top}\big)$. We write
	\begin{align}
		\bP_1\calM_1\big(\calE\times_2 &\tilde\bG_2^{(t-1)\top}\times_3 \tilde \bG_3^{(t-1)\top}\big)=\bP_1\calM_1(\calE\times_2(\bG_2\tilde{\bO}_2^{(t-1)})^{\top}\times_3 (\bG_3\tilde{\bO}_3^{(t-1)})^{\top}\big)\notag\\
		+&\bP_1\calM_1\big(\calE\times_2(\bG_2\tilde{\bO}_2^{(t-1)})^{\top}\times_3 (\tilde{\bG}_3^{(t-1)}-\bG_3\tilde{\bO}_3^{(t-1)})^{\top}\big)\notag\\
		&+\bP_1\calM_1\big(\calE\times_2(\tilde{\bG}_2^{(t-1)}-\bG_2\tilde{\bO}_2^{(t-1)})^{\top}\times_3 \tilde{\bG}_3^{(t-1)\top}\big) \label{eq:P1M1Et-1}
	\end{align}
	where $\tilde{\bO}_2^{(t-1)}=\arg\min_{\bO\in\OO^{R_2\times R_2}}\|\tilde{\bG}_2^{(t-1)}-\bG_2\bO\|_{\rm F}$ and $\tilde{\bO}_3^{(t-1)}=\arg\min_{\bO\in\OO^{R_3\times R_3}}\|\tilde{\bG}_3^{(t-1)}-\bG_3\bO\|_{\rm F}$. Clearly,
	\begin{align*}
		\big\|\bP_1\calM_1(\calE\times_2(\bG_2\tilde{\bO}_2^{(t-1)})^{\top}\times_3(\bG_3\tilde{\bO}_3^{(t-1)})^{\top}\big)\big\|=&\|\bP_1\calM_1(\calE)(\bG_2\otimes \bG_3)\|\\
		=&\|\bPhi_1(\bPhi_1^{\top}\bPhi_1)^{-1}\bPhi_1^{\top}\calM_1(\calE)(\bG_2\otimes\bG_3)\|
	\end{align*}
	where we abuse the notation and write $\bPhi_1=\bPhi_1(\bX_1)$.
	Similarly, as the proof of Lemma~\ref{lem:init}, denote $\bU_1$ the eigenvectors of $\bP_1$ so that $\bU_1^{\top}\bU_1=\bI_{J_1}$. Then,
	\begin{align*}
		\big\|\bP_1\calM_1(\calE\times_2(\bG_2\tilde{\bO}_2^{(t-1)})^{\top}\times_3(\bG_3\tilde{\bO}_3^{(t-1)})^{\top}\big)\big\|=\|\bU_1^{\top}\calM_1(\calE)(\bG_2\otimes \bG_3)\|
	\end{align*}
	where the matrix $\bU_1^{\top}\calM_1(\calE)(\bG_2\otimes \bG_3)$ has size $J_1\times (R_2R_3)$.
	
	\begin{lemma}\label{lem:UEG2G3}
		Suppose that Assumption~\ref{assump:noise} holds. There exist absolute constants $C_5,C_6>0$ so that
		$$
		\PP\left(\|\bU_1^{\top}\calM_1(\calE)(\bG_2\otimes \bG_3)\|/\sqrt{I_2I_3}\geq C_5\sqrt{J_1\log I_1}+C_6\sqrt{R_2R_3}\log^{2}I_1\right)\leq 2I_1^{-2}.
		$$
	\end{lemma}

	By Lemma~\ref{lem:UEG2G3}, we get that with probability at least $1-2I_1^{-2}$ that
	\begin{align}\label{eq:P1M1EG2G3}
		\big\|\bP_1\calM_1(\calE\times_2(\bG_2\tilde{\bO}_2^{(t-1)})^{\top}\times_3(\bG_3\tilde{\bO}_3^{(t-1)})^{\top}\big)\big\|/\sqrt{I_2I_3}=O\left(\sqrt{J_1\vee (R_2R_3)}\log^{2}I_1\right).
	\end{align}
	We now bound the second and third terms on RHS of (\ref{eq:P1M1Et-1}). Write
	\begin{align*}
		\big\|\bP_1\calM_1\big(\calE&\times_2(\bG_2\tilde{\bO}_2^{(t-1)})^{\top}\times_3 (\tilde{\bG}_3^{(t-1)}-\bG_3\tilde{\bO}_3^{(t-1)})^{\top}\big)\big\|\\
		=&\big\|\bU_1^{\top}\calM_1(\calE)\big(\bG_2\otimes(\tilde{\bG}_3^{(t-1)}-\bG_3\tilde{\bO}_3^{(t-1)})\big)\big\|.
	\end{align*}
	Recall $\bG_3=\bPhi_3\bB_3+\bR_3$ where we again abused the notation and dropped their dependences on $\bX_3$. Therefore,
	\begin{align}\label{eq:pi_eq1}
		\Big|\|\tilde{\bG}_3^{(t-1)}-\bG_3\tilde{\bO}_3^{(t-1)} \|_{\rm F}-\|\tilde\bG_3^{t-1}-\bPhi_3\bB_3\tilde \bO_3^{(t-1)}\|_{\rm F}\Big|/\sqrt{I_3}=O(\sqrt{R_3}\cdot J_3^{-\tau/2}).
	\end{align}
	We obtain
	\begin{align*}
		\big\|\bU_1^{\top}\calM_1(\calE)&\big(\bG_2\otimes(\tilde{\bG}_3^{(t-1)}-\bG_3\tilde{\bO}_3^{(t-1)})\big)\big\|\\
		\leq&\big\|\bU_1^{\top}\calM_1(\calE)\big(\bG_2\otimes(\tilde{\bG}_3^{(t-1)}-\bPhi_3\bB_3\tilde{\bO}_3^{(t-1)})\big)\big\|+\big\|\bU_1^{\top}\calM_1(\calE)\big(\bG_2\otimes \bR_3\big)\big\|.
	\end{align*}
	Note that the column space of $\tilde \bG_3^{(t-1)}$ belongs to the column space of $\bPhi_3$. Denote $\bU_3$ the left singular vectors of $\bPhi_3\bB_3$. Then,
	\begin{align}\label{eq:U1E1G2_DeltaG3}
		\big\|\bU_1^{\top}\calM_1(\calE)\big(\bG_2&\otimes(\tilde{\bG}_3^{(t-1)}-\bPhi_3\bB_3\tilde{\bO}_3^{(t-1)})\big)\big\|\notag\\
		&\leq \|\tilde\bG_3^{(t-1)}-\bPhi_3\bB_3\tilde \bO_3^{(t-1)}\|_{\rm F}\cdot \sup_{\bA\in \RR^{J_3\times R_3}, \|\bA\|_{\rm F}\leq 1}\|\bU_1^{\top}\calM_1(\calE)(\bG_2\otimes \bU_3 \bA)\|\notag\\
		=&O\big(\|\tilde\bG_3^{(t-1)}-\bPhi_3\bB_3\tilde \bO_3^{(t-1)}\|_{\rm F}\sqrt{I_2}\cdot \sqrt{J_1+J_3R_3+R_2R_3}\log^{3/2}I_1\big),
	\end{align}
	where the last inequality holds with probability at least $1-4I_1^{-2}$ and is due to Lemma~\ref{lem:UEsupGR}.
	
	\begin{lemma}\label{lem:UEsupGR}
		Suppose that $\calE$ has i.i.d entries satisfying Assumption~\ref{assump:noise}. Define $\mathfrak{B}(d_1,d_2):=\{\bA\in\RR^{d_1\times d_2}, \|\bA\|_{\rm F}\leq 1\}$. There exist absolute constants $C_1>0$ such that
		\begin{align}\label{eq:supUEGU3A}
			\PP\Big(\sup_{\bA\in \mathfrak{B}(J_3,R_3)} \big\|\bU_1^{\top}\calM_1(\calE)(\frac{\bG_2}{\sqrt{I_2}}\otimes \bU_3\bA) \big\|\geq C_1\sqrt{J_1+J_3R_3+R_2R_3}\log^{3/2}I_1\Big)\leq 4I_1^{-2}
		\end{align}
		and
		\begin{align}\label{eq:supUEU2AU3B}
			\PP\bigg(\sup_{\substack{\bA\in \mathfrak{B}(J_2,R_2)\\ \bB\in\mathfrak{B}(J_3,R_3)}} \big\|\bU_1^{\top}\calM_1(\calE)(\bU_2\bA\otimes \bU_3\bB) \big\|\geq C_2\sqrt{J_1+J_2R_2+J_3R_3+R_2R_3}\log^{3/2}I_1\bigg)\leq 4I_1^{-2}.
		\end{align}
	\end{lemma}
	Recall that $\bU_1, \bG_2, \bR_3$ are deterministic matrices. Following the same treatment as in the proof of Lemma~\ref{lem:AZbound}, we get with probability at least $1-2I_1^{-2}$,
	\begin{align}\label{eq:U1E1G2R3}
		\big\|\bU_1^{\top}\calM_1(\calE)\big(\bG_2\otimes \bR_3\big)\big\|\leq (R_3I_2I_3)^{1/2}J_3^{-\tau/2}\cdot\big(C_3\sqrt{J_1\log I_1}+C_4\sqrt{R_2R_3}\log^2I_1\big)
	\end{align}
	for some absolute constants $C_3,C_4>0$.
	
	Putting together (\ref{eq:U1E1G2_DeltaG3}) and (\ref{eq:U1E1G2R3}), we get with probability at least $1-6I_1^{-2}$ that
	\begin{align}\label{eq:P1M1EG2G3_delta}
		\big\|\bU_1^{\top}\calM_1(\calE)\big(\bG_2&\otimes(\tilde{\bG}_3^{(t-1)}-\bG_3\tilde{\bO}_3^{(t-1)})\big)\big\|\notag\\
		\leq&C_5({\rm Err}_{t-1}+\sqrt{R_3}\cdot J_3^{-\tau/2})\sqrt{I_2I_3}\cdot \sqrt{J_1+J_3R_3+R_2R_3}\log^{3/2}I_1.
	\end{align}
	Similarly, we can show that with probability at least $1-8I_1^{-2}$,
	\begin{align}\label{eq:P1M1EG2delta_G3}
		\|\bP_1\calM_1\big(\calE\times_2&(\tilde\bG_2^{(t-1)}-\bG_2\tilde{\bO}_2^{(t-1)})^{\top}\times_3 \tilde\bG_3^{(t-1)\top}\big)\|\notag\\
		\leq&C_6({\rm Err}_{t-1}+\sqrt{R_2}\cdot J_2^{-\tau/2})\sqrt{I_2I_3}\cdot \sqrt{J_1+J_3R_3+R_2R_3}\log^{3/2}I_1,
	\end{align}
	where we used the fact $\|\widetilde{\bG}_3^{(t-1)}\|_{\rm F}\leq \sqrt{R_3I_3}$ and the fact that the column space of $\widetilde{\bG}_3^{(t-1)}$ is a subspace of the column space of $\bU_3$.
	
	Therefore, by (\ref{eq:P1M1EG2G3}), (\ref{eq:P1M1EG2G3_delta}) and (\ref{eq:P1M1EG2delta_G3}), we conclude that with probability at least $1-16I_1^{-2}$ that
	\begin{align*}
		\big\|\bP_1\calM_1(\calE)(&\tilde\bG_2^{(t-1)}\otimes\tilde \bG_3^{(t-1)})\big\|\leq C_3\sqrt{I_2I_3}\cdot\sqrt{J_1+R_2R_3}\log^{2}I_1\\
		&+C_6({\rm Err}_{t-1}+\sqrt{R_2}J_2^{\tau/2}+\sqrt{R_3} J_3^{-\tau/2})\sqrt{I_2I_3}\cdot \sqrt{J_1+J_3R_3+R_2R_3}\log^{3/2}I_1.
	\end{align*}
	Now, we continue from (\ref{eq:P1Y1G2G3}). Recall that we denote $\mathring{\bG}_1/\sqrt{I_1}$ the top-$R_1$ left singular vectors of $\bP_1\bG_1$. As shown in the proof of initialization, we have $\|\mathring{\bG}_1\mathring{\bG}_1^{\top}-\bG_1\bG_1^{\top}\|_{\rm F}/I_1\leq 2\sqrt{R_1} J_1^{-\tau/2}$. Applying Daivs-Kahan theorem to (\ref{eq:P1Y1G2G3}), we get with probability at least  $1-16I_1^{-2}$ that
	\begin{align*}
		\|\tilde\bG_1^{(t)}&\tilde\bG_1^{(t)\top}-\mathring{\bG}_1\mathring{\bG}_1^{\top}\|_{\rm F}/I_1
		\leq C_4\frac{\sqrt{I_2I_3}\cdot\sqrt{J_1R_1+R_1R_2R_3}\log^{2}I_1}{\lambda_{\min}\sqrt{I_1}I_2I_3}\\
		&+C_5\frac{({\rm Err}_{t-1}+\sqrt{R_2}J_2^{-\tau/2}+\sqrt{R_3}J_3^{-\tau/2})\sqrt{I_2I_3}\cdot \sqrt{J_1R_1+J_3R_1R_3+R_1R_2R_3}\log^{3/2}I_1}{\lambda_{\min}\sqrt{I_1}I_2I_3}
	\end{align*}
	for some absolute constants $C_4,C_5>0$.
	
	Therefore, as long as $\lambda_{\min}\sqrt{I_1I_2I_3}\geq C_5'\sqrt{J_1R_1+J_3R_1R_3+R_1R_2R_3}\log^{3/2}I_1$ for a large enough absolute constant $C_5'>0$, we get with probability at least $1-16I_1^{-2}$ that
	\begin{align*}
		\|\tilde\bG^{(t)}_1&\tilde\bG^{(t)\top}_1-\bG_1\bG_1^{\top}\|_{\rm F}/I_1\\
		\leq& \frac{{\rm Err}_{t-1}}{2}+\frac{2\sqrt{R_1}J_1^{-\tau/2}+\sqrt{R_2} J_2^{-\tau/2}+\sqrt{R_3}J_3^{-\tau/2}}{2}+C_4'\frac{\sqrt{J_1R_1+R_1R_2R_3}\log^{2}I_1}{\lambda_{\min}\sqrt{I_1I_2I_3}}.
	\end{align*}
	In the same fashion, we can prove similar bounds of $\|\tilde\bG_m^{(t)}\tilde\bG_m^{(t)\top}-\bG_m\bG_m^{\top}\|_{\rm F}$ for all $m=1,2,3$. Therefore, with probability at least $1-48I_1^{-2}$,
	\begin{align}\label{eq:iter_cont}
		{\rm Err}_t\leq \frac{{\rm Err}_{t-1}}{2}+(\sqrt{R_1}J_1^{-\tau/2}+\sqrt{R_2}J_2^{-\tau/2}+\sqrt{R_3}J_3^{-\tau/2})+C_4'\frac{\sqrt{J_1R_1+R_1R_2R_3}\log^{2}I_1}{\lambda_{\min}\sqrt{I_1I_2I_3}},
	\end{align}
	which proves the first claim of Lemma~\ref{lem:init}. The same properties can be proved for all iterations and all hold on the same event.
	
	By the above contraction inequality in (\ref{eq:iter_cont}), after
	$$
	t_{\max}=O(\log(\lambda_{\min}\sqrt{I_1I_2I_3/J_1})+\tau\cdot \log (J_1)+1)
	$$
	iterations, we obtain
	$$
	{\rm Err}_{t_{\max}}=C_4''\frac{\sqrt{J_1R_1+R_1R_2R_3}\log^{2}I_1}{\lambda_{\min}\sqrt{I_1I_2I_3}}+\sqrt{R_1}J_1^{-\tau/2}+\sqrt{R_2}J_2^{-\tau/2}+\sqrt{R_3}J_3^{-\tau/2}
	$$
	which holds with probability at least $1-48I_1^{-2}$. The proof is concluded by noting that $J_1\asymp J_2\asymp J_3$ and $J_1\geq J_2\geq J_3$.
\end{proof}

In order to prove Theorem~\ref{thm:F}, we begin with proving the following result.
\begin{lemma}\label{lem:F}(Factor tensor)
	Suppose that  conditions of Lemma~\ref{lem:init} hold. Then, with probability at least $1-49I_1^{-2}$ that
	\begin{align*}
		\|\tilde{\calF}-\calF&\times_1\tilde\bO_1^{\top}\times_2 \tilde\bO_2^{\top}\times_3 \tilde \bO_3^{\top}\|_{\rm F}\\
		\leq&C_1\Big(\frac{\kappa_0\cdot\sqrt{J_1R_1+R_1R_2R_3}\log^{2}I_1}{\sqrt{I_1I_2I_3}}\Big)+2\kappa_0\lambda_{\min}\sqrt{R_1}J_1^{-\tau/2}.
	\end{align*}
	where $\tilde\bO_m$ is an orthogonal matrix which realizes the minimium $\min_{\bO} \|\tilde\bG_m-\bG_m\bO\|_{\rm F}$ and $C_1>0$ is an absolute constant.
\end{lemma}

\begin{proof}[Proof of Lemma~\ref{lem:F}]
	Recall that $\tilde\calF=(I_1I_2I_3)^{-1}\cdot\calY\times_1 \tilde \bG_1^{\top}\times_2 \tilde\bG_2^{\top}\times_3 \tilde \bG_3^{\top}$ and so that
	\begin{align*}
		\tilde\calF=(I_1I_2I_3)^{-1}\cdot\calF\times_1(\tilde \bG_1^{\top}\bG_1)\times_2(\tilde \bG_2^{\top}\bG_2)\times_3(\tilde \bG_3^{\top}\bG_3)+(I_1I_2I_3)^{-1}\cdot\calE\times_1 \tilde\bG_1^{\top}\times_2 \tilde\bG_2^{\top}\times_3 \tilde\bG_3^{\top}
	\end{align*}
	where we used the fact $\tilde\bG_m^{\top} \bGamma_m={\bf 0}$ since the column space of $\tilde\bG_m$ is a subspace of the column space of $\bPhi_m(\bX_m)$. Recall that
	$$
	\bG_m^{\top}\big(\tilde\bG_m\tilde\bG_m^{\top}-\bG_m\bG_m^{\top}\big)\bG_m/I_m^2=\bG_m^{\top}\tilde\bG_m (\bG_m^{\top}\tilde\bG_m)^{\top}/I_m^2-\bI_{R_m}
	$$
	where $\bG_m^{\top}\tilde\bG_m$ is an $R_m\times R_m$ matrix. Then, by Lemma~\ref{lem:init}, with probability at least $1-48I_1^{-2}$ that
	$$
	\big\|\bG_m^{\top}\tilde\bG_m (\bG_m^{\top}\tilde\bG_m)^{\top}/I_m^2-\bI_{R_m} \big\|_{\rm F}\leq C_5\frac{\sqrt{J_1R_1+R_1R_2R_3}\log^{2}I_1}{\lambda_{\min}\sqrt{I_1I_2I_3}}+2\sqrt{R_1}J_1^{-\tau/2}.
	$$
	It implies that for all $m=1,2,3$, there exists an orthonormal matrix $\tilde \bO_m\in\OO^{R_m\times R_m}$ so that
	$$
	\|\tilde\bG_m^{\top}\bG_m/I_m-\tilde\bO_m^{\top}\|_{\rm F}=C_5\frac{\sqrt{J_1R_1+R_1R_2R_3}\log^{2}I_1}{\lambda_{\min}\sqrt{I_1I_2I_3}}+2\sqrt{R_1}J_1^{-\tau/2},
	$$
	which holds with the same probability. Therefore,
	\begin{align}\label{eq:tildeF-F-2terms}
		\tilde\calF-\calF\times_1\tilde\bO_1^{\top}\times_2 \tilde\bO_2^{\top}\times_3& \tilde \bO_3^{\top}=(I_1I_2I_3)^{-1}\cdot\calE\times_1 \tilde\bG_1^{\top}\times_2 \tilde\bG_2^{\top}\times_3 \tilde\bG_3^{\top}\notag\\
		+&\calF\big(\times_1(\tilde \bG_1^{\top}\bG_1/I_1)\times_2(\tilde \bG_2^{\top}\bG_2/I_2)\times_3(\tilde \bG_3^{\top}\bG_3/I_3)-\times_1 \tilde \bO_1^{\top}\times_2\tilde\bO_2^{\top}\times_3\tilde\bO_3^{\top}\big).
	\end{align}
	Observe that
	\begin{align*}
		\big\|\calF&\times_1(\tilde\bG_1^{\top}\bG_1/I_1-\tilde\bO_1^{\top})\times_2(\tilde \bG_2^{\top}\bG_2/I_2)\times_3(\tilde \bG_3^{\top}\bG_3/I_3)\big\|_{\rm F}\\
		&\leq \big\|(\tilde\bG_1^{\top}\bG_1/I_1-\tilde\bO_1^{\top})\calM_1(\calF)\big((\tilde \bG_2^{\top}\bG_2/I_2)\otimes(\tilde \bG_3^{\top}\bG_3/I_3)\big)\big\|_{\rm F}\\
		&\leq \|\calM_1(\calF)\|\cdot \|\tilde\bG_1^{\top}\bG_1/I_1-\tilde\bO_1^{\top}\|_{\rm F}\\
		&\leq C_5\frac{\kappa_0\cdot\sqrt{J_1R_1+R_1R_2R_3}\log^{2}I_1}{\sqrt{I_1I_2I_3}}+2\kappa_0\lambda_{\min}\sqrt{R_1}J_1^{-\tau/2}.
	\end{align*}
	As a result, we can show with probability at least $1-48I_1^{-2}$ that
	\begin{align}\label{eq:FtildeGG-O}
		\big\|\calF\big(&\times_1(\tilde \bG_1^{\top}\bG_1/I_1)\times_2(\tilde \bG_2^{\top}\bG_2/I_2)\times_3(\tilde \bG_3^{\top}\bG_3/I_3)-\times_1 \tilde \bO_1^{\top}\times_2\tilde\bO_2^{\top}\times_3\tilde\bO_3^{\top}\big)\big\|\notag\\
		\leq&C_5'\frac{\kappa_0\cdot\sqrt{J_1R_1+R_1R_2R_3}\log^{2}I_1}{\sqrt{I_1I_2I_3}}+6\kappa_0\lambda_{\min}\sqrt{R_1} J_1^{-\tau/2}.
	\end{align}
	Observe that the rank of $\calM_1\big(\calE\times_1 \tilde\bG_1^{\top}\times_2 \tilde\bG_2^{\top}\times_3 \tilde\bG_3^{\top}\big)$ is bounded by $R_1$. Similarly,
	\begin{align*}
		\big\|\calE\times_1 \tilde\bG_1^{\top}\times_2 \tilde\bG_2^{\top}\times_3 \tilde\bG_3^{\top}\big\|_{\rm F}&=\big\|\calM_1\big(\calE\times_1 \tilde\bG_1^{\top}\times_2 \tilde\bG_2^{\top}\times_3 \tilde\bG_3^{\top}\big) \big\|_{\rm F}\\
		\leq&\sqrt{R_1}\cdot \big\|\calM_1\big(\calE\times_1 \tilde\bG_1^{\top}\times_2 \tilde\bG_2^{\top}\times_3 \tilde\bG_3^{\top}\big) \big\|.
	\end{align*}
	\begin{lemma}\label{lem:EG1G2G3_tilde}
		Suppose that Assumption~\ref{assump:noise} holds and assume $J_1\asymp J_2\asymp J_3$ and $J_1\geq R_1\geq R_2\geq R_3$. There exist absolute constants $C_7>0$ so that,
		\begin{align*}
			\PP\Big(\big\|\calM_1\big(\calE\times_1 \tilde\bG_1^{\top}\times_2 \tilde\bG_2^{\top}\times_3 \tilde\bG_3^{\top}\big) \big\|/\sqrt{I_1I_2I_3}\geq C_7\sqrt{J_1R_1}\log^{3/2} I_1\Big)\leq I_1^{-2}.
		\end{align*}
	\end{lemma}
	By Lemma~\ref{lem:EG1G2G3_tilde}, with probability at least $1-I_1^{-2}$,
	\begin{align}\label{eq:MG1G2G3}
		\big\|\calM_1\big(\calE\times_1 \tilde\bG_1^{\top}\times_2 \tilde\bG_2^{\top}\times_3 \tilde\bG_3^{\top}\big) \big\|/\sqrt{I_1I_2I_3}\leq  C_1'\sqrt{J_1R_1}\log^{3/2}I_1.
	\end{align}
	for some absolute constant $C_1'>0$. 
	
	Putting together (\ref{eq:tildeF-F-2terms}), (\ref{eq:FtildeGG-O}) and (\ref{eq:MG1G2G3}), we conclude that with probability at least $1-49I_1^{-2}$ that
	\begin{align*}
		\|\tilde\calF-\calF\times_1\tilde\bO_1^{\top}&\times_2 \tilde\bO_2^{\top}\times_3 \tilde \bO_3^{\top}\|_{\rm F}\\
		\leq&C_8\frac{\kappa_0\cdot\sqrt{J_1R_1+R_1R_2R_3}\log^{2}I_1}{\sqrt{I_1I_2I_3}}+6\kappa_0\lambda_{\min}\sqrt{R_1}J_1^{-\tau/2},
	\end{align*}
	which proves Lemma~\ref{lem:F}.
\end{proof}

\begin{proof}[Proof of Theorem~\ref{thm:F}]
	Let $\hat\bO_m$ denote the left singular vectors of $\calM_m(\tilde\calF)$ for all $m\in[M]$, and $\tilde\bD_m$ denote the singular values of $\calM_m(\tilde\calF)$.  Similarly, denote $\bD_m$ the singular values of $\calM_m(\calF)$. Lemma~\ref{lem:F} implies, with probability at least $1-49I_1^{-2}$, that
	$$
	\|\bD_m-\tilde\bD_m\|\leq C_8\frac{\kappa_0\cdot\sqrt{J_1R_1+R_1R_2R_3}\log^{2}I_1}{\sqrt{I_1I_2I_3}}+6\kappa_0\lambda_{\min}\sqrt{R_1}J_1^{-\tau/2}
	$$
	and
	\begin{align*}
		\|\tilde\bO_1\calM_1(\tilde\calF)&\calM_1(\tilde\calF)^{\top}\tilde\bO_1^{\top}-\calM_1(\calF)\calM_1(\calF)^{\top}\|\\
		=&\|\calM_1(\tilde\calF)\calM_1(\tilde\calF)^{\top}-\tilde\bO_1^{\top}\calM_1(\calF)\calM_1(\calF)^{\top}\tilde\bO_1\|.
	\end{align*}
	Denote $\tilde\bH_1=\tilde\bO_1\hat\bO_1$ so that $\tilde\bO_1\calM_1(\tilde\calF)\calM_1(\tilde\calF)^{\top}\tilde\bO_1^{\top}-\calM_1(\calF)\calM_1^{\top}(\calF)=\tilde\bH_1\tilde\bD_1^2\tilde\bH_1^{\top}-\bD_1^2$ where we used Assumption~\ref{assum:tfm-ic-1}.
	Denote $\varepsilon_{\alpha}=C_8\kappa_0^2\lambda_{\min}\sqrt{J_1R_1+R_1R_2R_3}\log^{2}(I_1)/\sqrt{I_1I_2I_3}+6\kappa_0^2\lambda^2_{\min}\sqrt{R_1}J_1^{-\tau/2}$.
	Then, by Lemma~\ref{lem:F}, we obtain with probability at least $1-49I_1^{-2}$ that
	\begin{align}\label{eq:tildeD-D}
		\|\tilde\bH_1\tilde\bD_1^2\tilde\bH_1^{\top}-\bD_1^2\|_{\rm F}\leq \varepsilon_{\alpha}.
	\end{align}
	Note that for each $j=1,\cdots,R_1$, we obtain $\sigma_j(\bD_1^2)-\sigma_{j+1}(\bD_1^2)\geq \lambda_{\min}\cdot {\rm Egap}(\calF)$. Under the conditions of Theorem~\ref{thm:F}, it follows with probability at least $1-49I_1^{-2}$ that
	$$
	\min_{1\leq j\leq R_1}\sigma_j(\bD_1^2)-\sigma_{j+1}(\bD_1^2)\geq C_1\kappa_0^2\sqrt{R_1}\cdot \varepsilon_{\alpha}
	$$
	for a large enough constant $C_1>1$ implying that the order of eigenvalues of $\bD_1^2$ will be maintained in view of (\ref{eq:tildeD-D}). By applying the Davis-Kahan theorem to each isolated eigenvector of $\tilde\bH_1\tilde\bD_1^2\tilde\bH_1$, we can conclude that $\|\tilde\bh_j\tilde\bh_j^{\top}-\be_j\be_j^{\top}\|\leq 1/(2\kappa_0^2\sqrt{R_1})$ which holds for all $j=1,\cdots,R_1$ where $\tilde\bh_j$ denotes the $j$-th column of $\tilde\bH_1$ and $\be_j$ denotes the $j$-th canonical basis vector. Indeed, it holds as long as the ${\rm Egap}(\calF)$ is large enough as stated in the conditions of Theorem~\ref{thm:F}.
	It implies that there exists a $\tilde{s}_j\in\{\pm1\}$ so that $\|\tilde \bh_j \tilde{s}_j-\be_j\|\leq 1/\sqrt{2\kappa_0^4R_1}$ for each $j$. Denote $\tilde{\bS}_1={\rm diag}(\tilde{s}_1,\cdots,\tilde{s}_{R_1})$ so that
	$$
	\|\tilde\bH_1\tilde{\bS}_1-\bI_{R_1}\|_{\rm F}\leq \Big(\sum_{j=1}^{R_1}\|\tilde\bh_j\tilde{s}_j-\be_j\|^2\Big)^{1/2}\leq 1/(\sqrt{2}\kappa_0^2).
	$$
	Note that, on the same event, $\|\tilde\bH_1\tilde\bD_1^2\tilde\bH_1^{\top}-\tilde\bD_1^2\|_{\rm F}\leq \varepsilon_{\alpha}+\|\tilde\bD^2_1-\bD_1^2\|_{\rm F}\leq 2\varepsilon_{\alpha}$ where the last bound is due to Lemma~\ref{lem:F}. Since $\tilde\bD_1$ is a diagonal matrix, $\|\tilde\bH_1\tilde\bD_1^2\tilde\bH_1^{\top}-\tilde\bS_1\tilde\bD_1^2\tilde\bS_1\|_{\rm F}\leq 2\varepsilon_{\alpha}$. Write
	\begin{align*}
		\|\tilde\bH_1\tilde\bD_1^2\tilde\bH_1^{\top}-\tilde\bS_1\tilde\bD_1^2\tilde\bS_1^{\top}\|_{\rm F}\geq \|(\tilde\bH_1-\tilde\bS_1)\tilde\bD_1^2\tilde\bS_1^{\top}+\tilde\bS_1\tilde\bD_1^2(\tilde\bH_1-\tilde\bS_1)^{\top}\|_{\rm F}\\
		-\|(\tilde\bH_1-\tilde\bS_1)\tilde\bD_1^2(\tilde\bH_1-\tilde\bS_1)^{\top}\|_{\rm F}\geq 2\|\tilde\bH_1-\tilde\bS_1\|_{\rm F}\sigma_{\min}(\tilde\bD_1^2)-O\big(\|\tilde\bH_1-\tilde\bS_1\|_{\rm F}^2\sigma_{\max}(\tilde\bD_1^2)\big)\\
		\geq 2\|\tilde\bH_1-\tilde\bS_1\|_{\rm F}\sigma_{\min}(\tilde\bD_1^2)-O\big(\kappa_0^2\|\tilde\bH_1-\tilde\bS_1\|_{\rm F}^2\sigma_{\min}(\tilde\bD_1^2)\big)
	\end{align*}
	where the last inequality holds with probability at least $1-49I_1^{-2}$ as long as $\|\bD_1-\tilde\bD_1\|\leq \lambda_{\min}/4$ which is guaranteed by the lower bound on $\lambda_{\min}$. It implies that
	$$
	\|\tilde\bH_1\tilde\bD_1^2\tilde\bH_1^{\top}-\tilde\bS_1\tilde\bD_1^2\tilde\bS_1^{\top}\|_{\rm F}\geq (2-\sqrt{2})\|\tilde\bH_1-\tilde\bS_1\|_{\rm F}\sigma_{\min}(\tilde\bD_1^2)\geq \|\tilde\bH_1-\tilde\bS_1\|_{\rm F}\lambda_{\min}^2/5.
	$$
	Therefore, we conclude with probability at least $1-49I_1^{-2}$ that
	$$
	\|\tilde\bH_1-\tilde\bS_1\|_{\rm F}\leq 10\varepsilon_{\alpha}/\lambda_{\min}^2\leq C_7\frac{\kappa_0^2\sqrt{JR_1+R_1R_2R_3}\log^{2}I_1}{\lambda_{\min}\sqrt{I_1I_2I_3}}+C_8\kappa_0^2\sqrt{R_1}J_1^{-\tau/2}.
	$$
	As a result, $\hat\bG_1=\tilde\bG_1\hat\bO_1$ and then with probability at least $1-49I_1^{-2}$,
	\begin{align*}
		\|\hat\bG_1-\bG_1\tilde\bS_1\|_{\rm F}/\sqrt{I_1}\leq& \|\hat\bG_1-\bG_1\tilde\bH_1\|_{\rm F}/\sqrt{I_1}+\|\tilde\bH_1-\tilde\bS_1\|_{\rm F}\\
		=& \|\hat\bG_1-\bG_1\tilde\bO_1\hat\bO_1\|_{\rm F}/\sqrt{I_1}+\|\tilde\bH_1-\tilde\bS_1\|_{\rm F}\\
		=&\|\tilde\bG_1-\bG_1\tilde\bO_1\|_{\rm F}/\sqrt{I_1}+\|\tilde\bH_1-\tilde\bS_1\|_{\rm F}\\
		\leq&C_7\frac{\kappa_0^2\sqrt{JR_1+R_1R_2R_3}\log^{2}I_1}{\lambda_{\min}\sqrt{I_1I_2I_3}}+C_8\kappa_0^2\sqrt{R_1}J_1^{-\tau/2}
	\end{align*}
	where the last inequality is due to that $\tilde\bO_1$ realizes the minimum of $\min_{\bO}\|\tilde\bG_1-\bG_1\bO\|_{\rm F}$. Clearly, the bounds can be proved identically for all $\|\hat\bG_m-\bG_m\bS_m\|_{\rm F}/\sqrt{I_m}$.
	
	At last, recall that $\hat\calF=\tilde\calF\times_1 \hat\bO_1^{\top}\times_2 \hat\bO_2^{\top}\times_3 \hat\bO_3^{\top}$. We conclude that with probability at least $1-49I_1^{-2}$,
	\begin{align*}
		\|\hat\calF-\calF\times_1\bS_1&\times_2\bS_2\times_3\bS_3\|_{\rm F}=\|\hat\calF-\calF\times_1\tilde\bH_1^{\top}\times_2\tilde\bH_2^{\top}\times_3\tilde\bH_3^{\top}\|_{\rm F}+O(\kappa_0\varepsilon_{\alpha}/\lambda_{\min})\\
		=&\|\tilde\calF\times_1 \hat\bO_1^{\top}\times_2 \hat\bO_2^{\top}\times_3 \hat\bO_3^{\top}-\calF\times_1\tilde\bH_1^{\top}\times_2\tilde\bH_2^{\top}\times_3\tilde\bH_3^{\top}\|_{\rm F}+O(\kappa_0\varepsilon_{\alpha}/\lambda_{\min})\\
		=&\|\tilde\calF-\calF\times_1(\hat\bO_1\tilde\bH_1^{\top})\times_2(\hat\bO_2\tilde\bH_2^{\top})\times_3(\hat\bO_3\tilde\bH_3^{\top})\|_{\rm F}+O(\kappa_0\varepsilon_{\alpha}/\lambda_{\min})\\
		=&\|\tilde\calF-\calF\times_1\tilde\bO_1^{\top}\times_2\tilde\bO_2^{\top}\times_3\tilde\bO_3^{\top}\|_{\rm F}+O(\kappa_0\varepsilon_{\alpha}/\lambda_{\min})=O((\kappa_0\varepsilon_{\alpha}/\lambda_{\min})\\
		=&O\Big(\frac{\kappa_0^3\sqrt{JR_1+R_1R_2R_3}\log^2I_1}{\sqrt{I_1I_2I_3}}+\kappa_0^3\lambda_{\min}\sqrt{R_1}J_1^{-\tau/2}\Big),
	\end{align*}
	which proves Theorem~\ref{thm:F}.
\end{proof}

\begin{proof}[Proof of Theorem~\ref{thm:Gamma}]
	Without loss of generality, we prove the bound for $m=1$.
	Recall by definition that
	$$
	\hat\bGamma_1=\bP_1^{\perp}\hat\bA_1=\bP_{1}^{\perp}\calM_1(\calY\times_2\bP_2\times_2\bP_3)\calM_1(\hat{\calF}\times_2\hat\bG_2\times_3\hat\bG_3)^{\top}\big(\calM_1(\hat\calF)\calM_1^{\top}(\hat\calF)\big)^{-1}/(I_2I_3).
	$$
	Since the column space of $\hat\bG_m$ is a subspace of the column space of $\bP_m$ so that $\hat\bG_m^{\top}\bGamma_m={\bf 0}$, we can write
	\begin{align*}
		\calM_1(\calY\times_2\bP_2\times_3&\bP_3)\calM_1(\hat{\calF}\times_2\hat\bG_2\times_3\hat\bG_3)^{\top}=\calM_1(\calY)(\bP_2\otimes \bP_3)(\hat\bG_2\otimes\hat\bG_3)\calM_1(\hat\calF)^{\top}\\
		&=\calM_1(\calY)(\hat\bG_2\otimes\hat\bG_3)\calM_1(\hat\calF)^{\top}\\
		=&(\bG_1+\bGamma_1)\calM_1(\calF)\big((\bG_2^{\top}\hat\bG_2)\otimes (\bG_3^{\top}\hat\bG_3)\big)\calM_1(\hat\calF)^{\top}+\calM_1(\calE)(\hat\bG_2\otimes\hat\bG_3)\calM_1(\hat\calF)^{\top}
	\end{align*}
	and as a result
	\begin{align*}
		\hat\bGamma_1
		=&\bP_{1}^{\perp}\bGamma_1\calM_1(\calF)\big((\bG_2^{\top}\hat\bG_2)\otimes (\bG_3^{\top}\hat\bG_3)\big)\calM_1(\hat\calF)^{\top}\big(\calM_1(\hat\calF)\calM_1^{\top}(\hat\calF)\big)^{-1}/(I_2I_3)\\
		&+\bP_{1}^{\perp}\calM_1(\calE)(\hat\bG_2\otimes\hat\bG_3)\calM_1(\hat\calF)^{\top}\big(\calM_1(\hat\calF)\calM_1^{\top}(\hat\calF)\big)^{-1}/(I_2I_3).
	\end{align*}
	Under the conditions of Lemma~\ref{lem:init} and by Theorem~\ref{thm:F}, we conclude with probability at least $1-49I_1^{-2}$ that $\sigma_{\min}\big(\calM_1(\hat\calF)\big)\geq \lambda_{\min}/2$. Now, it suffices to bound the spectral norm $\bP_{1}^{\perp}\calM_1(\calE)(\hat\bG_2\otimes\hat\bG_3)\calM_1(\hat\calF)^{\top}\big(\calM_1(\hat\calF)\calM_1^{\top}(\hat\calF)\big)^{-1}$. Since the column spaces of $\hat\bG_2$ and $\hat \bG_3$ are the subspaces of column spaces of $\bPhi_2(\bX_2)$ and $\bPhi_3(\bX_3)$, respectively, we have
	\begin{align*}
		\big\| &\bP_{1}^{\perp}\calM_1(\calE)(\hat\bG_2\otimes\hat\bG_3)\calM_1(\hat\calF)^{\top}\big(\calM_1(\hat\calF)\calM_1^{\top}(\hat\calF)\big)^{-1}/(I_2I_3)\big\|\\
		\leq&\big\| \calM_1(\hat\calF)^{\top}\big(\calM_1(\hat\calF)\calM_1^{\top}(\hat\calF)\big)^{-1}\big\|/(I_2I_3)\cdot\sup_{\substack{\bA_2\in\mathfrak{B}(J_2,R_2), \bA_3\in\mathfrak{B}(J_3,R_3)\\ \bB\in\mathfrak{B}(R_2R_3,R_1)}}\|\bP_1^{\perp}\calM_1(\calE)\big((\bU_2\bA_2)\otimes(\bU_3\bA_3)\big)\bB\|
	\end{align*}
	where $\bU_m$ are the left singular vectors of $\bPhi_m(\bX_m)$ and $\mathfrak{B}(d_1,d_2)=\{\bB\in\RR^{d_1\times d_2}: \|\bB\|\leq 1\}$. The following lemma is needed whose proof is reproducible by the proof of Lemma~\ref{lem:EG1G2G3_tilde}.
	
	\begin{lemma}\label{lem:P1E1supA}
		Suppose that Assumption~\ref{assump:noise} holds and assume $J_1\asymp J_2\asymp J_3$ and $J_1\geq R_1\geq R_2\geq R_3$. There exist an absolute constant $C_9>0$ so that with probability at least $1-I_1^{-2}$,
		$$
		\sup_{\substack{\bA_2\in\mathfrak{B}(J_2,R_2), \bA_3\in\mathfrak{B}(J_3,R_3)\\ \bB\in\mathfrak{B}(R_2R_3,R_1)}}\|\bP_1^{\perp}\calM_1(\calE)\big((\bU_2\bA_2)\otimes(\bU_3\bA_3)\big)\bB\|\leq C_9\sqrt{I_1+J_1R_1+R_1R_2R_3}\log^{3/2} I_1.
		$$
	\end{lemma}
	By Lemma~\ref{lem:P1E1supA}, we get with probability at least $1-I_1^{-2}$ that
	\begin{align*}
		\Big\|\bP_{1}^{\perp}&\calM_1(\calE)(\hat\bG_2\otimes\hat\bG_3)\calM_1(\hat\calF)^{\top}\big(\calM_1(\hat\calF)\calM_1^{\top}(\hat\calF)\big)^{-1}/(I_2I_3) \Big\|_{\rm F}\\
		&\leq C_9'\frac{\sqrt{I_1R_1+J_1R_1^2+R_1^2R_2R_3}\log^{3/2} I_1}{\lambda_{\min}\sqrt{I_2I_3}}
	\end{align*}
	where we used the fact $\big\|\calM_1(\hat\calF)^{\top}\big(\calM_1(\hat\calF)\calM_1^{\top}(\hat\calF)\big)^{-1}\big\|\leq C_1'\lambda_{\min}^{-1}$ by Theorem~\ref{thm:F} and conditions of Lemma~\ref{lem:init}.
	
	Since $\bP_{1}^{\perp}\bGamma_1=\bGamma_1$, we get
	\begin{align*}
		\bP_{1}^{\perp}\bGamma_1&\calM_1(\calF)\big((\bG_2^{\top}\hat\bG_2)\otimes (\bG_3^{\top}\hat\bG_3)\big)\calM_1(\hat\calF)^{\top}\big(\calM_1(\hat\calF)\calM_1^{\top}(\hat\calF)\big)^{-1}/(I_2I_3)\\
		=&\bGamma_1\calM_1(\calF)\big((\bG_2^{\top}\hat\bG_2)\otimes (\bG_3^{\top}\hat\bG_3)\big)\calM_1(\hat\calF)^{\top}\big(\calM_1(\hat\calF)\calM_1^{\top}(\hat\calF)\big)^{-1}/(I_2I_3)\\
		=&\bGamma_1\calM_1(\calF)\big(\bS_2\otimes \bS_3\big)\calM_1(\hat\calF)^{\top}\big(\calM_1(\hat\calF)\calM_1^{\top}(\hat\calF)\big)^{-1}\\
		&+O\Big(\kappa_0\|\bGamma_1\|\cdot (\|\hat\bG_2^{\top}\bG_2/I_2-\bS_2\|_{\rm F}+\|\hat\bG_3^{\top}\bG_3/I_3-\bS_3\|_{\rm F})\Big)
	\end{align*}
	where the last term is bounded in terms of Frobenius norm and $\bS_2, \bS_3$ are defined as in Theorem~\ref{thm:F}. Meanwhile,
	$$
	\big\|\bS_1^{\top}\calM_1(\hat\calF)-\calM_1(\calF)\big(\bS_2^{\top}\otimes\bS_3^{\top}\big)\big\|_{\rm F}\leq \|\hat\calF-\calF\times_1\bS_1\times_2\bS_2\times_3\bS_3\|_{\rm F}.
	$$
	Therefore, by Theorem~\ref{thm:F}, with probability at least $1-49I_1^{-2}$ that
	\begin{align*}
		\Big\|\bGamma_1\calM_1(\calF)\big(\bS_2^{\top}\otimes \bS_3^{\top}\big)\calM_1(\hat\calF)^{\top}\big(\calM_1(\hat\calF)\calM_1^{\top}(\hat\calF)\big)^{-1}-\bGamma_1\bS_1^{\top}\Big\|_{\rm F} \\
		=O\big(\lambda_{\min}^{-1}\|\bGamma_1\|\cdot \|\hat\calF-\calF\times_1\bS_1\times_2\bS_2\times_3\bS_3\|_{\rm F}\big)\\
		=O\Big(\|\bGamma_1\|\cdot \frac{\kappa_0^3\sqrt{J_1R_1+R_1R_2R_3}\log^{2}I_1}{\lambda_{\min}\sqrt{I_1I_2I_3}}\Big)+O\big(\|\bGamma_1\|\cdot \kappa_0^3\sqrt{R_1}J_1^{-\tau/2}\big).
	\end{align*}
	Finally, we get with probability at least $1-50I_1^{-2}$ that
	\begin{align*}
		\|\hat\bGamma_1&-\bGamma_1\bS_1^{\top}\|_{\rm F}\\
		=&O\bigg(\|\bGamma_1\|\cdot\Big( \frac{\kappa_0^3\sqrt{J_1R_1+R_1R_2R_3}\log^{2}I_1}{\lambda_{\min}\sqrt{I_1I_2I_3}}+\kappa_0^3\sqrt{R_1}J_1^{-\tau/2}\Big)\bigg)
		+O\Big(\frac{\sqrt{R_1I_1+J_1R_1^2+R_1^2R_2R_3}\log^{3/2} I_1}{\lambda_{\min}\sqrt{I_2I_3}}\Big)
	\end{align*}
	which concludes the proof of Theorem~\ref{thm:Gamma} in view of $J_m\leq I_m$.
\end{proof}

\section{Number of Basis Functions} \label{append:n-basis}

Determination of the number of basis functions is an important task in non-parametric and semi-parametric estimations. 
It is more challenging in the STEFA model. 
According to our analysis in ``Effect of the number of fitting basis'' in the simulation section, 
the interactions between the true number of basis, the working number of basis, the signal-to-noise ratio, and the relative mean squared errors is not straightforward. 
Specifically, Table~\ref{tabl:J-effect} shows that increasing the sieve order $J$ does not always improve the performance and $J=16$ does not achieve the best performance among all choices of $J$, even though the data is simulated with order 16. 

To start a formal investigation of this challenging problem, we can first take the perspective of a regression problem: 
$$\calM_m(\calY) = \bPhi_m\bB_m + \bE_m,$$
where $\bPhi_m$ is an ensemble of basis functions, $\bB_m$ is the coefficients and $\bE_m$ is the residual. 
An potential data-driven way to determine the sieve degree is to construct an F-test based on the statistics $(\|P_{\Phi_m}\calM_m(\calY)\|_F^2 - \|P_{\Phi_m'}\calM_m(\calY)\|_F^2)/\|P^\perp_{\Phi_m}\calM_m(\calY)\|_F^2$ when comparing two choices of sieve degrees (corresponding to $\Phi_m$ and the reduced one $\Phi'_m$). However, strictly speaking, the residual $\bE_m$ is not of multivariate Gaussian and the coefficient matrix $\bB_m$ is restricted to a certain low rank structure due to the other modes in tensor $\calY$. 
The proper test for this sieve determination needs further investigation and is beyond the current main streamline of this paper. 

In the simulation and real data analysis sections, we choose the sieve degree $J$ in an ad-hoc way in the simulation. 
For instance the degree used in polynomial basis and B-spline basis are chosen to accommodate one's expectation on the smoothness of the function. 
The impact of such ad-hoc choices of sieve degree $J$ was investigated in Table~3 in the simulation section, where an obvious bias-variance trade-off was observed. 
In real data analyses, we choose the $J$ that minimizes the relative mean squared errors. 

Here, we present some additional empirical results of selecting the number of basis through relative mean squared errors (ReMSE). 
In this simulation, we consider a three-way tensor with fixed dimensions $I_1=I_2=I_3=100$, whose signal part can be decomposed to a Tucker decomposition with rank $R=(3, 3, 3)$.  
We fix the signal-to-noise ratio $\alpha=1.5$, and simulate the parametric part of loading within the manifold space of Legendre function of a two-dimensional $\bX_m, m=1,2,3$. 
The magnitude of $\bGamma_m$ is controlled by $\mu$ as in Section 5. 
We consider three different magnitudes of $\bGamma_m$'s and four different numbers of true basis $J$. 
For each combination of $(\mu, J)$, we simulate for $100$ times and report the average number of selected basis in Table~\ref{tabl:order-selection} for four different methods, which minimize 
(a) in-sample ${\rm ReMSE}_{\calY}$;
(b) in-sample ${\rm ReMSE}_{\calS}$;
(c) out-of-sample ${\rm ReMSE}_{\calY}$; and 
(d) out-of-sample ${\rm ReMSE}_{\calS}$, respectively. 
We note that 
the ReMSE with respect to signal ${\rm ReMSE}_{\calS}$ is only available in simulation environment and the ReMSE with respect to observed tensors ${\rm ReMSE}_{\calY}$ is more practical in real data analysis. 

Table~\ref{tabl:order-selection} shows that selecting number of basis by minimizing in-sample ReMSE usually leads to an over-estimation of $J$. 
However, selecting number of basis by minimizing out-of-sample ReMSE usually produces more accurate results. 
Comparing the last two columns, we notice that, as long as we use the out-of-sample ReMSE, it does not really matter whether we use ReMSE with respect to observed noisy $\calY$ or the true signal $\calS$ for all different combinations of magnitudes of $\bGamma_m$ and number of true basis $J$. 
This observation provides empirical support to using out-of-sample ${\rm ReMSE}_{\calY}$ to select $J$ in real applications where ${\rm ReMSE}_{\calS}$ can not be calculated. 
Moreover, by comparing estimated $\hat J$ by minimizing out-of-sample ReMSE across different true $(\mu, J)$ combinations, we observe that out-of-sample $\hat J$ tends to give an under-estimation of the true number of basis for the purpose of robustness when $J$ is large.  

\begin{table}[htpb!]
	\centering
		\begin{tabular}{cc|cc|cc}
			\hline
			\multicolumn{2}{c|}{Truth} &  \multicolumn{2}{c|}{Average $\hat J$ by minimizing In-Sample} &  \multicolumn{2}{c}{Average $\hat J$ by minimizing Out-of-Sample}  \\
			\hline
			$\mu$ & $J$ & ${\rm ReMSE_{\calY}}$ & ${\rm ReMSE_{\calS}}$  & ${\rm ReMSE_{\calY}}$  & ${\rm ReMSE_{\calS}}$  \\
			\hline 
			0.1 & 2 & 3.97 & 3.31 & 2.00  &2.00 \\
			0.1 & 4 & 8.00  & 6.52 & 4.00  & 4.00 \\
			0.1 & 8 & 16.00 &  13.70 &  8.02 & 8.02\\
			0.1 & 16 & 32.00 &  28.32 & 11.38 & 11.38\\
			\hline
			0.3 & 2 & 3.97 & 3.76 & 2.00  & 2.00\\
			0.3 & 4 & 7.99 & 7.62 & 4.00  & 4.00\\
			0.3 & 8 & 16.00 & 15.60  & 7.81 & 7.81\\
			0.3& 16 & 32.00 & 31.60  & 8.35 & 8.34\\
			\hline
			0.5 & 2 & 3.99 & 3.96 & 2.00  & 2.00\\
			0.5&4 & 8.00  &7.90 &  4.00  &4.00\\
			0.5 & 8 & 15.99& 15.93 & 6.79 & 6.78\\
			0.5&16& 32.00  & 31.86 & 6.88 & 6.88\\
			\hline
		\end{tabular}
	\caption{Average number of $\hat J$ of selected basis by the four methods over 100 repetitions.}
	\label{tabl:order-selection}
\end{table}

\section{Kernel Smoothing with Tensor Factor Model}

In this section, we derive the kernel smoothing formula \eqref{eqn:kernel-smoothing} under the vanilla tensor factor model.
Under this setting, the relevant covariates $\bX_1$ is still available for the 1-st mode and we would like to predict a new tensor $\calY^{new}\in\RR^{I_1^{new}\times I_2\times I_3}$ with new covariates $\bX_1^{new}\in\RR^{I_1^{new}\times D_1}$.
However, we do not use the STEFA model to incorporate $\bX_1$ in the model.
Instead, we use an algorithm for solving noisy Tucker decomposition \eqref{eqn:tensor-factor-model} and obtain an estimator of the signal part $\hat\calS =  \hat\calF \times_1 \hat\bA_1 \times_2 \hat\bA_2 \times_3 \hat\bA_3$.
The informative covariates $\bX_1$ and $\bX_1^{new}$ are used non-parametrically.

Recall that we defined the kernel weight matrix $\bW \in \RR^{I_1^{new} \times I_1}$ with entry
\[
w_{ij} = \frac{ K_h(dist(\bx_{1,i\cdot}^{new},\bx_{1,j\cdot})) }{ \sum_{i = 1}^{I_1} K_h(dist(\bx_{1,i\cdot}^{new},\bx_{1,j\cdot})) }, \quad i \in [I_1^{new}] \text{ and } j \in [I_1].
\]
where $K_h(\cdot)$ is the kernel function, $dist(\cdot, \cdot)$ is a pre-defined distance function such as the Euclidean distance, and $\bx_{1,i\cdot}$ is the $i$-th row of $\bX_1$.


For each row of $\bX_1^{new}$, we will predict a tensor slice $\bY_i^{new}\in\RR^{I_2\times I_3}$.
Let $\by_i^{new} = \vec(\bY_i^{new})$,  $\bY \defeq \calM_1(\calY)^\top = \brackets{\by_1 \cdots \by_{I_1}}$, consisting of the signal part $\bS \defeq \calM_1(\calS)^\top = \brackets{\bs_1 \cdots \bs_{I_1}}$ and the noise part $\bE \defeq \calM_1(\calE)^\top = \brackets{\be_1 \cdots \be_{I_1}}$.
Define $\bSigma_y \defeq \E{\bY^\top\bY}$ and $\bSigma^{new}_i \defeq \E{\bY^\top\by^{new}_i}$.
The best linear predictor for $\hat\by^{new}_i$ based on $\bY$ is
\begin{equation} \label{eqn:hat-y-0-ls}
	\hat\by^{new}_i = \bY\cdot\bSigma_y^{-1} \bSigma^{new}_i.
\end{equation}
With knowledge of covariates $\bX_1$, it is possible to estimate $\bSigma_y^{-1}$ and $\bSigma^{new}_i$ from $\bX_1$ and $\bx^{new}_{i\cdot}$.
However, in practice it involves inverting a $I_1\times I_1$ matrix $\bSigma_y$ which may be computational costly when $I_1$ is large.
The computational burden can be relieved by taking advantage of the Tucker low-rank structure.

To estimate $\bSigma^{new}_i$, we note that $\bSigma^{new}_i = \E{(I_2I_3)^{-1}\bY^\top\bs^{new}_i}$ where $\bs^{new}_i$ is the signal part of $\by^{new}_i$.
Thus, it can be estimated by $\hat\bSigma^{new}_i = (I_2I_3)^{-1}\bY^\top\hat\bs^{new}_i$.
We use kernel predictors for $\hat\bs^{new}_i$, that is,
\begin{equation} \label{eqn:s0-est}
	\hat\bs^{new}_i = \frac{\sum_{j = 1}^{I_1}\hat\bs_{j} K_h(dist(\bx_{1,i\cdot}^{new},\bx_{1,j\cdot}))}{\sum_{i = 1}^{I_1} K_h(dist(\bx_{1,i\cdot}^{new},\bx_{1,j\cdot}))} = \sum_{j = 1}^{I_1} w_{ij}\hat\bs_{j}.
\end{equation}
With careful calculation, we have a simpler expression for $\hat\by^{new}_i$.
First, we have $\hat\bs^{new}_i=\hat\bS\bw_{i\cdot}$ and
\begin{equation}  \label{eqn:y-0-hat-simple}
	\hat\by^{new}_i = \bY\hat\bSigma_y^{-1}\hat\bSigma^{new}_i =
	\bY\hat\bSigma_y^{-1} \cdot \bY^\top\hat\bS\bw_{i\cdot}
	= \bY\hat\bSigma_y^{-1}\bY^\top\bY\hat\bA_1\hat\bA_1^\top\bw_{i\cdot} = \bY\hat\bA_1\hat\bA_1^\top\bw = \hat\bS\bw_{i\cdot}.
\end{equation}
Equation \eqref{eqn:y-0-hat-simple} shows that, under the tensor factor model, we do not need to actually calculate $\hat\bSigma_y^{-1}$ to obtain the best linear predictor \eqref{eqn:hat-y-0-ls}.
Kernel smoothing formula \eqref{eqn:kernel-smoothing} is obtained by applying \eqref{eqn:y-0-hat-simple} to each $i$-th row of $\bX^{new}_1$ and stacking the resulting tensor slices $\hat\bY^{new}_i$ along the first mode for $i \in I^{new}_1$.

\section{More Simulation Results}  
\subsection{Inequal Dimensions}
In this section, we consider the setting where tensor $\calY$ has different dimensions, that is, $I_1, I_2, I_3$ are not equal.
We fix $R=3$, $I_1 = 100$ but vary $\alpha$ and $I_2, I_3 \ge I_1$.
The ReMSE of estimating the loading matrices $\bA_m$ and the tensor $\cal Y$ are reported in Table \ref{tabl:unbalanced}.

Although the dimensions for the three modes are artificially designed to be different in this simulation, no significant difference between $\ell_2(\hat\bA_m)$, $m\in[3]$ is observed. The error in estimating the loading matrices of the three modes appears to be symmetric. With a fixed signal-to-noise ratio coefficient $\alpha$ and a fixed $I_{min}=I_1=100$, the performance of both projected Tucker and vanilla Tucker decomposition is not sensitive to the other two dimensions $I_2, I_3$.

\begin{table}[htpb!]
	\centering
	\resizebox{\textwidth}{!}{%
		\begin{tabular}{ccc|c|c|c|c|c|c|c|c}
			\hline
			\multirow{2}{*}{$\paran{I_1,I_2,I_3}$} & \multirow{2}{*}{$R$} & \multirow{2}{*}{$\alpha$} & \multicolumn{4}{c|}{IP-SVD} & \multicolumn{4}{c}{HOOI} \\ \cline{4-11}
			&  &  & $\ell_2(\hat\bA_1)$ & $\ell_2(\hat\bA_2)$ & $\ell_2(\hat\bA_3)$ & $\mathrm{ReMSE}_\calY$ & \multicolumn{1}{c|}{$\ell_2(\hat\bA_1)$} & \multicolumn{1}{c|}{$\ell_2(\hat\bA_2)$} & \multicolumn{1}{c|}{$\ell_2(\hat\bA_3)$} & $\mathrm{ReMSE}_\calY$ \\ \hline
			(100,100,200) & 3 & 0.3 & \msd{0.805}{0.260} & \msd{0.805}{0.242} & \msd{0.820}{0.253} & \msd{0.885}{0.283} & \msd{1.703}{0.014} & \msd{1.703}{0.016} & \msd{1.718}{0.007} & \msd{3.647}{0.782}\\
			(100,100,400) & 3 & 0.3 & \msd{0.824}{0.227} & \msd{0.850}{0.219} & \msd{0.859}{0.234} & \msd{0.930}{0.284} & \msd{1.704}{0.012} & \msd{1.703}{0.015} & \msd{1.725}{0.004} & \msd{4.329}{0.958}\\
			(100,200,200) & 3 & 0.3 & \msd{0.840}{0.223} & \msd{0.828}{0.213} & \msd{0.782}{0.208} & \msd{0.903}{0.264} & \msd{1.706}{0.014} & \msd{1.718}{0.006} & \msd{1.719}{0.006} & \msd{4.072}{0.802}\\
			(100,200,400) & 3 & 0.3 & \msd{0.840}{0.222} & \msd{0.857}{0.239} & \msd{0.853}{0.221} & \msd{0.935}{0.259} & \msd{1.705}{0.011} & \msd{1.718}{0.005} & \msd{1.725}{0.003} & \msd{4.711}{0.910}\\ \hline
			(100,100,200) & 3 & 0.5 & \msd{0.264}{0.073} & \msd{0.278}{0.076} & \msd{0.274}{0.067} & \msd{0.279}{0.071} & \msd{1.641}{0.172} & \msd{1.635}{0.177} & \msd{1.655}{0.167} & \msd{1.715}{0.348}\\
			(100,100,400) & 3 & 0.5 & \msd{0.274}{0.065} & \msd{0.282}{0.068} & \msd{0.271}{0.071} & \msd{0.280}{0.060} & \msd{1.695}{0.048} & \msd{1.688}{0.047} & \msd{1.715}{0.039} & \msd{1.981}{0.288}\\
			(100,200,200) & 3 & 0.5 & \msd{0.258}{0.061} & \msd{0.277}{0.062} & \msd{0.262}{0.068} & \msd{0.268}{0.063} & \msd{1.677}{0.095} & \msd{1.686}{0.117} & \msd{1.685}{0.124} & \msd{1.825}{0.323}\\
			(100,200,400) & 3 & 0.5 & \msd{0.273}{0.078} & \msd{0.270}{0.069} & \msd{0.262}{0.068} & \msd{0.271}{0.066} & \msd{1.692}{0.074} & \msd{1.704}{0.071} & \msd{1.712}{0.068} & \msd{2.063}{0.373}\\ \hline
		\end{tabular}%
	}
	\caption{Unbalanced tensor dimensions. The average spectral and Frobenius Schatten q-$\sin\Theta$ loss ($q=2$) for $\hat\bA_m$, $m\in[3]$ and average Frobenius loss for $\calY$ under various settings.}
	\label{tabl:unbalanced}
\end{table}

\subsection{Comparison to the MMC Linear Tensor Regression}   \label{sec:compare-xu-2019}

In this section, we compare our approach (IP-SVD) to the MMC tensor regression method of \citet{xu2019generalized} on a linear tensor model. The $100\times100\times 100$ observed tensor $\calY$ is generated in the same way as in Section~\ref{sec:simu} with a core tensor of $1\times 1\times 1$. That is we set $I_1=I_2=I_3=100$ and $R_1=R_2=R_3=1$. Covariates $\bX_m, m=1,2,3$ of $100\times 1$ are randomly sampled from a uniform distribution on $[0, 1]$ i.i.d.. The parametric loading matrix $\bG_m$ is quadratic with respect to $\bX_m$ such that $[\bG_m]_i \propto 1 + [\bX_m]_i + [\bX_m]_i^2$. The non-parametric loading matrix $\bGamma_m$ is added in a similar way to Section~\ref{sec:simu-nonparam-loading}. 
In summary, the observed tensor is generated according to 
\begin{equation}
	\calY = \calF \times_1\bA_1\times_2\bA_2\times\bA_3 +\calE,\label{eq:additional-simu-model}
\end{equation}
where 
\begin{align*}
	\bA_m &= \bG_m + \mu\bGamma_m,\\
	\bG_m &= \bZ_m / \|\bZ_m\|,\\
	\bZ_m &= 1 + \bX_m + \bX_m^2,
\end{align*}
and $\calE$ is a $I_1\times I_2\times I_3$ tensor with i.i.d. standard Gaussian entries, $\calF$ is a $1\times1\times1$ (a scalar) of value $I^\alpha$ and $\bX_m$ has i.i.d. Uniform(0,1) entries.
Again, we use $\alpha$ to control the signal-to-noise ratio and use $\mu$ to control the relative strength of non-parametric loading parts. 

Three models are used to fit $\calY$. 
IP-SVD(NP) model denotes the IP-SVD approach that we fit $\calY$ with correctly specified ranks and sieve orders and with the {\it non-parametric (NP)} loading parts as in \eqref{eq:additional-simu-model}. 
IP-SVD(P) model is a similar model of IP-SVD(NP) except that we ignore the non-parametirc part $\bGamma$ and only fits the {\it parametric (P)} part $\bA_m=\bG_m$. 
MMC-LTR model stands for the multiple-mode-covariate linear regression model from \citet{xu2019generalized} where each $\bA_m$ is assumed to be linear in $\bX_m$ (and therefore, misspecifies the model with low sieve ranks). 
We report the relative MSE (ReMSE, averaged over 100 repetitions) of the three methods for different signal-to-noise ratios (varying $\alpha$) and for different non-parametric components (varying $\mu$)  in Table~\ref{tabl:compare-to-LTR}. 

\begin{table}[htpb!]
	\centering
	\resizebox{.78\textwidth}{!}{%
		\begin{tabular}{c|ccc|ccc}
			\hline
			$\alpha$ & \multicolumn{3}{c|}{2}& \multicolumn{3}{c}{1}\\
			\hline
			$\mu$ & 1 & 0.1 & 0 & 1 & 0.1 & 0 \\
			\hline
			IP-SVD(NP) 
			&\msd{0.004}{1.5e-4}  & \msd{0.002}{6.7e-5} & \msd{0.002}{6.6e-5} 
			&\msd{0.343}{0.019} & \msd{0.174}{0.007} & \msd{0.172}{0.007} 
			\\
			\hline
			IP-SVD(P) 
			&\msd{0.931}{0.002}  &\msd{0.170}{0.001} &\msd{2.6e-4}{6.7e-5} 
			&\msd{0.932}{0.002}  &\msd{0.172}{0.001} & \msd{0.026}{0.007} 
			\\
			\hline
			MMC-LTR 
			& \msd{0.937}{0.008}  & \msd{0.271}{0.112} & \msd{0.199}{0.138}
			&\msd{0.937}{0.008} &\msd{0.272}{0.112} &\msd{0.201}{0.138} 
			\\
			\hline
		\end{tabular}%
	}
	\caption{Mean and standard deviation of the Relative MSE for the three approaches under different signal-to-noise ratios and different strength of non-parametric parts in 100 repetitions.}
	\label{tabl:compare-to-LTR}
\end{table}

When the relative strength of the non-parametric loading $\bGamma_m$ is strong ($\mu = 1$),  we have IP-SVD(NP) $>$ IP-SVD(P) $>$ MMC-LTR, where $>$ means ``performs better than", for both settings of moderate and strong signal-to-noise ratio. 
When the relative strength of the non-parametric loading $\bGamma_m$ is weak ($\mu=0.1$) or non-existing ($\mu=0$), the IP-SVD still performs better than MMC-LTR. 
The IP-SVD(NP) has a disadvantage relative to IP-SVD(P) under this setting, especially when the model is fully parametric ($\mu=0$). 
However, when the signal-to-noise ratio is strong $\alpha = 2$, IP-SVD(NP) still performs the best in face of the weak relative strength of the non-parametric loading $\bGamma_m$  ($\mu=0.1$). 

We conclude at least for the specific setting with linear linkage function, nonlinear loading factors and non-parametric loading parts, IP-SVD outperforms the tensor regression method due to IP-SVD's capability as a semiparametric model with sieve expansions. 
At the same time, we acknowledge that the tensor regression method can handle more complicated linkage functions such as logistic model and probit model. 
In general, the STEFA model as a unsupervised method is a complement to the supervised MMC tensor regression model \citep{xu2019generalized}.

\section{Proofs of Technical Lemmas}\label{sec:app_technical}

\begin{proof}[Proof of Lemma~\ref{lem:AZbound}]
	By the definitions of $\bA_1$ and $\bZ_1$, we write
	\begin{align}\label{eq:A1Z1_bound}
		\bA_1\bZ_1^{\top}=\bP_1\bG_1\calM_1(\calF)\big((\bP_2\bG_2)\otimes (\bP_3\bG_3)\big)^{\top}(\bU_2\otimes \bU_3)\big((\bU_2\otimes \bU_3)^{\top}\calM_1^{\top}(\calE)\bU_1\bU_1^{\top}\big)
	\end{align}
	where $\bU_m\bU_m^{\top}=\bP_m$, $\bU_m\in\RR^{I_m\times J_m}$ and $\bU_m^{\top}\bU_m=\bI_{J_m}$. It suffices to prove the upper bound of $\|\bB_1(\bU_2\otimes \bU_3)^{\top}\calM_1^{\top}(\calE)\bU_1\|$ where $\bB_1=\bA_1(\bU_2\otimes \bU_3)$ is an $I_1\times (J_2J_3)$ deterministic matrix.
	
	Denote $\bE_1=\calM_1^{\top}(\calE)\bU_1\in\RR^{(I_2I_3)\times J_1}$. By Assumption~\ref{assump:noise}, $\bE_1=(\be_{1,1},\cdots, \be_{1,I_2I_3})^{\top}$ has i.i.d. rows and each row is a $J_1$-dimensional centered sub-exponential random vector in that
	\begin{align}\label{eq:e_subexp}
		\sup\nolimits_{\|\bv\|\leq 1}\PP\big(|\langle \bv, \be_{1,j}\rangle|\geq t\big)\leq \exp(-Ct)\quad \textrm{for all } j\in[I_2I_3]
	\end{align}
	for any $t>1$. Meanwhile, $\EE(\be_{1,j}\be_{1,j}^{\top})=\bI_{J_1}$ for all $j\in[I_2I_3]$. By (\ref{eq:e_subexp}), there exists an absolute constant $C_1>0$ so that
	\begin{align}\label{eq:small_prob}
		\PP\Big(\max_{j\in [I_2I_3], k\in[J_1]} \big|e_{1,j}(k)\big|\geq C_1\log(I_2I_3J_1)\Big)\leq \frac{1}{(I_2I_3J_1)^4}.
	\end{align}
	Denote the above event by $\mathfrak{E}_0$. We obtain $\max_j\|\be_{1,j}\|\leq C_1\sqrt{J_1}\log(I_2I_3J_1)$ on $\mathfrak{E}_0$. Now, we denote $\delta_1= C_1\sqrt{J_1}\log(I_2I_3J_1)$.

	Denote $\{\tilde{\bu}_{23,j}\}_{j=1}^{I_2I_3}$ the columns of $(\bU_2\otimes\bU_3)^{\top}$. We write
	\begin{align*}
		\bB_1(\bU_2\otimes \bU_3)^{\top}\bE_1=&\sum_{j=1}^{I_2I_3} \bB_1 \tilde{\bu}_{23,j}\be_{1,j}^{\top}\\
		=&\sum_{j=1}^{I_2I_3} \bB_1 \tilde{\bu}_{23,j}\be_{1,j}^{\top}\mathbbm{1}(\|\be_{1,j}\|\leq \delta_1)+\sum_{j=1}^{I_2I_3} \bB_1 \tilde{\bu}_{23,j}\be_{1,j}^{\top}\mathbbm{1}(\|\be_{1,j}\|> \delta_1),
	\end{align*}
	where the second term on the RHS is simply $0$ on event $\mathfrak{E}_0$. It suffices to bound the first term, which is a sum of independent random matrices. Write
	\begin{align}
		\Big\|\sum_{j=1}^{I_2I_3} \bB_1 \tilde{\bu}_{23,j}\be_{1,j}^{\top}\mathbbm{1}(\|\be_{1,j}\|\leq \delta_1) \Big\|\leq& \Big\|\sum_{j=1}^{I_2I_3} \bB_1 \tilde{\bu}_{23,j}\Big(\be_{1,j}^{\top}\mathbbm{1}(\|\be_{1,j}\|\leq \delta_1)-\EE \be_{1,j}^{\top}\mathbbm{1}(\|\be_{1,j}\|\leq \delta_1)\Big) \Big\|\nonumber\\
		&+\Big\|\sum_{j=1}^{I_2I_3} \bB_1\tilde{\bu}_{23,j}\EE \be_{1,j}^{\top}\mathbbm{1}(\|\be_{1,j}\|\leq \delta_1) \Big\|\label{eq:con_opt1}.
	\end{align}
	Since $\EE \be_{1,j}=0$ and $\|\bB_1\|\leq \|\bA_1\|\leq \kappa_0\lambda_{\min}\sqrt{I_1I_2I_3}$, together with (\ref{eq:small_prob}), we get
	\begin{align*}
		\Big\|\sum_{j=1}^{I_2I_3} \bB_1\tilde{\bu}_{23,j}&\EE \be_{1,j}^{\top}\mathbbm{1}(\|\be_{1,j}\|\leq \delta_1) \Big\|=\Big\|\sum_{j=1}^{I_2I_3} \bB_1\tilde{\bu}_{23,j}\EE \be_{1,j}^{\top}\mathbbm{1}(\|\be_{1,j}\|> \delta_1) \Big\|\\
		\leq& \|\bB_1\|\cdot \sum_{j=1}^{I_2I_3}\EE\|\be_{1,j}\|\mathbbm{1}(\|\be_{1,j}\|> \delta_1) =I_2I_3\|\bB_1\|\cdot \EE\|\be_{1,1}\|\mathbbm{1}(\|\be_{1,1}\|> \delta_1)\\
		\leq& I_2I_3\|\bB_1\|\cdot \EE^{1/2}\big(\|\be_{1,1}\|^2\big)\cdot \PP^{1/2}\big((\|\be_{1,1}\|> \delta_1\big)\\
		\leq& I_2I_3\kappa_0\lambda_{\min}\sqrt{I_1I_2I_3}\cdot C_2\sqrt{J_1}\cdot \frac{1}{(I_2I_3J_1)^2}\leq C_2\frac{\kappa_0\lambda_{\min}\sqrt{I_1I_2I_3}}{I_2I_3J_1}.
	\end{align*}
	Now it suffices to prove the upper bound of first term in RHS of (\ref{eq:con_opt1}), which is the spectral norm of the sum of independent random matrices. By definition,
	$$
	\Big\|\bB_1 \tilde{\bu}_{23,j}\Big(\be_{1,j}^{\top}\mathbbm{1}(\|\be_{1,j}\|\leq \delta_1)-\EE \be_{1,j}^{\top}\mathbbm{1}(\|\be_{1,j}\|\leq \delta_1)\Big)\Big\|\leq 2\delta_1\kappa_0\lambda_{\min}\sqrt{I_1I_2I_3},\quad \forall j\in[I_2I_3]
	$$
	and
	\begin{align*}
		\Big\|\sum_{j=1}^{I_2I_3}\EE&\bB_1\tilde{\bu}_{23,j}\be_{1,j}^{\top}\be_{1,j}\tilde{\bu}_{23,j}^{\top}\bB_1^{\top} \cdot  \mathbbm{1}(\|\be_{1,j}\|\leq \delta_1)\Big\|\\
		\leq & \Big\|\sum_{j=1}^{I_2I_3}\EE\bB_1\tilde{\bu}_{23,j}\be_{1,j}^{\top}\be_{1,j}\tilde{\bu}_{23,j}^{\top}\bB_1^{\top} \Big\|+\Big\|\sum_{j=1}^{I_2I_3}\EE\bB_1\tilde{\bu}_{23,j}\be_{1,j}^{\top}\be_{1,j}\tilde{\bu}_{23,j}^{\top}\bB_1^{\top}\cdot  \mathbbm{1}(\|\be_{1,j}\|> \delta_1)\Big\|\\
		\leq & J_1I_1I_2I_3\kappa_0^2\lambda_{\min}^2+\sum_{j=1}^{I_2I_3}\|\tilde{\bu}_{23,j}\|^2 \kappa_0^2\lambda_{\min}^2I_1I_2I_3\cdot \EE \|\be_{1,j}\|^2 \mathbbm{1}(\|\be_{1,j}\|> \delta_1)\\
		\leq& J_1\kappa_0^2\lambda_{\min}^2I_1I_2I_3+J_2J_3\kappa_0^2\lambda_{\min}^2I_1I_2I_3\cdot C_2J_1\cdot \frac{1}{(I_2I_3J_1)^2}\leq 2J_1\kappa_0^2\lambda_{\min}^2I_1I_2I_3,
	\end{align*}
	where the last inequality holds since $I_m\geq J_m$. Similarly, we have
	\begin{align*}
		\Big\|\sum_{j=1}^{I_2I_3}\EE&\be_{1,j}\tilde{\bu}_{23,j}^{\top}\bB_1^{\top}\bB_1\tilde{\bu}_{23,j}\be_{1,j}^{\top} \cdot  \mathbbm{1}(\|\be_{1,j}\|\leq \delta_1)\Big\|\\
		\leq&\Big\|\sum_{j=1}^{I_2I_3}\EE\be_{1,j}\tilde{\bu}_{23,j}^{\top}\bB_1^{\top}\bB_1\tilde{\bu}_{23,j}\be_{1,j}^{\top} \Big\|+\Big\|\sum_{j=1}^{I_2I_3}\EE\be_{1,j}\tilde{\bu}_{23,j}^{\top}\bB_1^{\top}\bB_1\tilde{\bu}_{23,j}\be_{1,j}^{\top}\cdot  \mathbbm{1}(\|\be_{1,j}\|>\delta_1)\Big\|\\
		\leq& \|\bB_1\|_{\rm F}^2+\Big\|\sum_{j=1}^{I_2I_3}\EE\be_{1,j}\tilde{\bu}_{23,j}^{\top}\bB_1^{\top}\bB_1\tilde{\bu}_{23,j}\be_{1,j}^{\top}\cdot  \mathbbm{1}(\|\be_{1,j}\|>\delta_1)\Big\|\\
		\leq& R_1\kappa_0^2\lambda_{\min}^2I_1I_2I_3+\sum_{j=1}^{I_2I_3}\tilde{\bu}_{23,j}^{\top}\bB_1^{\top}\bB_1\tilde{\bu}_{23,j}\EE \langle\be_{1,1},\bv \rangle^2 \cdot  \mathbbm{1}(\|\be_{1,1}\|>\delta_1)\big]\leq 2R_1\kappa_0^2\lambda_{\min}^2I_1I_2I_3,
	\end{align*}
	where $\bv$ is any fixed vector in $\RR^{J_1}$ with unit norm.
	
	Then, by matrix Bernstein inequality (\cite{tropp2012user}), with probability at least $1-I_1^{-2}$,
	\begin{align*}
		\Big\|\sum_{j=1}^{I_1} \bB_1 \Big(\be_{1,j}\mathbbm{1}(\|\be_{1,j}\|\leq \delta_1)-\EE \be_{1,j}\mathbbm{1}(&\|\be_{1,j}\|\leq \delta_1)\Big)\tilde{\bu}_j^{\top} \Big\|\\
		&\leq C_1\kappa_0\lambda_{\min}(I_1I_2I_3)^{1/2}\big(\sqrt{J_1\log I_1}+\delta_1\log I_1\big),
	\end{align*}
	where we assumed $R_m\leq J_m$.
	Since $\delta_1\asymp \sqrt{J_1}\log(I_1)$ and $I_1>J_1$, we get with probability at least $1-2I_1^{-2}$ that
	\begin{align*}
		\big\|\bB_1(\bU_2\otimes \bU_3)^{\top}\bE_1\bU_1\bU_1^{\top} \big\|\leq C_3\kappa_0\lambda_{\min}(I_1I_2I_3)^{1/2}\cdot\sqrt{J_1}\log^2 I_1
	\end{align*}
	for some absolute constant $C_3>0$.
	
	Therefore, by (\ref{eq:A1Z1_bound}), we get with probability  at least $1-2I_1^{-2}$ that
	$$
	\|\bA_1\bZ_1^{\top}\|\leq C_4\kappa_0\lambda_{\min}(I_1I_2I_3)^{1/2}\cdot\sqrt{J_1}\log^2 I_1
	$$
	for some absolute constant $C_4>0$, which proves Lemma~\ref{lem:AZbound}.
\end{proof}

\vspace{1cm}

\begin{proof}[Proof of Lemma~\ref{lem:ZZ_bound}]
	Similarly as the proof of Lemma~\ref{lem:AZbound}, we denote $\bE_2=\calM_1(\calE)(\bU_2\otimes \bU_3)$ the $I_1\times (J_2J_3)$ matrix with independent rows $\big(\be_{2,i}^{\top}\big)_{i=1}^{I_1}$. Here, $\be_{2,i}$ is a sub-exponential random vector with $\EE \be_{2,i}\be_{2,i}^{\top}=\bI_{J_2J_3}$ and $\EE \be_{2,i}={\bf 0}$.
	
	Similarly, we denote $\{\tilde{\bu}_i\}_{i=1}^{I_1}$ the columns of $\bU_1^{\top}$. Then, we write
	\begin{align*}
		\bU_1^{\top}\calM_1(\calE)(\bU_2\otimes \bU_3)=\sum_{i=1}^{I_1} \tilde{\bu}_i\be_{2,i}^{\top}
	\end{align*}
	and as a result
	\begin{align}
		\big(\bU_1^{\top}\calM_1&(\calE)(\bU_2\otimes \bU_3)\big)\big(\bU_1^{\top}\calM_1(\calE)(\bU_2\otimes \bU_3)\big)^{\top}-J_2J_3\bI_{J_1}\notag\\
		=&\Big(\sum_{i=1}^{I_1} \tilde{\bu}_i\be_{2,i}^{\top}\Big)\Big(\sum_{i=1}^{I_1} \tilde{\bu}_i\be_{2,i}^{\top}\Big)^{\top}-J_2J_3\bI_{J_1}\notag\\
		=&\Big(\sum_{i=1}^{I_1}\tilde{\bu}_i\be_{2,i}^{\top}\be_{2,i}\tilde{\bu}_i^{\top}-J_2J_3\bI_{J_1}\Big)+\sum_{1\leq i_1\neq i_2\leq I_1}\tilde{\bu}_{i_1}\be_{2,i_1}^{\top}\be_{2,i_2}\tilde{\bu}_{i_2}^{\top}. \label{eq:ZZbound_2terms}
	\end{align}
	Observe that $\EE\sum_{i=1}^{I_1}\tilde{\bu}_i\be_{2,i}^{\top}\be_{2,i}\tilde{\bu}_i^{\top}=J_2J_3\bI_{J_1}$. Note that
	\begin{align*}
		\sum_{i=1}^{I_1}\tilde{\bu}_i\be_{2,i}^{\top}\be_{2,i}\tilde{\bu}_i^{\top}-J_2J_3\bI_{J_1}=\sum_{i=1}^{I_1}\|\be_{2,i}\|^2 \tilde{\bu}_i\tilde{\bu}_i^{\top}-J_2J_3\bI_{J_1}.
	\end{align*}
	Similarly as proof of Lemma~\ref{lem:AZbound}, denote $\delta_1=C_1\sqrt{J_2J_3}\log I_1$ so that $\PP(\max_{i}\|\be_{2,i}\|\geq \delta_1)\leq (I_1J_2J_3)^{-3}$. Following the same treatment there, we can show that
	\begin{align}\label{eq:ue2e2u}
		\PP\Big(\Big\| \sum_{i=1}^{I_1}\tilde{\bu}_i\be_{2,i}^{\top}\be_{2,i}\tilde{\bu}_i^{\top}-J_2J_3\bI_{J_1}\Big\|\geq (J_2J_3)^{1/2}\big(C_2\sqrt{J_1\log I_1}+C_3\log^2 I_1\big)\Big)\leq I_1^{-2}
	\end{align}
	with $C_2,C_3>0$ being absolute constants, where we used the facts (needed in matrix Bernstein inequality), with $\be_1$ being the first column of $\calM_1(\calE)^{\top}$,
	\begin{align*}
		\Big\|\sum_{i=1}^{I_1}\EE\big(&\|\be_{2,i}\|^2-J_2J_3\big)^2\|\tilde{\bu}_i\|^2\tilde{\bu}_i\tilde{\bu}_i^{\top} \Big\|=\EE\big(\|\be_{2,1}\|^2-J_2J_3\big)^2\Big\|\sum_{i=1}^{I_1} \|\tilde{\bu}_i\|^2\tilde{\bu}_i\tilde{\bu}_i^{\top}\Big\|\\
		\leq& \EE\big(\|\be_{2,1}\|^2-J_2J_3\big)^2=\EE \|\be_{2,1}\|^4-(J_2J_3)^2\stackrel{\bU_{23}=\bU_2\otimes \bU_3}{=}\EE\big(\be_1^{\top}\bU_{23}\bU_{23}^{\top}\be_1\big)^2-(J_2J_3)^2\\
		\stackrel{\bU_{23} = (\bu_{23,j})_{j=1}^{J_2J_3}}{=}&\EE\Big(\sum_{j=1}^{J_2J_3}\langle \be_1, \bu_{23,j}\rangle^2\Big)^2-(J_2J_3)^2\\
		=&\sum_{j=1}^{J_2J_3}\EE\langle \be_1, \bu_{23,j}\rangle^4+\sum_{1\leq j\neq j'\leq J_2J_3} \EE \langle \be_1, \bu_{23,j}\rangle^2\langle \be_1, \bu_{23,j'}\rangle^2-(J_2J_3)^2\\
		=&\sum_{j=1}^{J_2J_3}\EE\langle \be_1, \bu_{23,j}\rangle^4+\sum_{1\leq j\neq j'\leq J_2J_3} \|\bu_{23,j}\|^2\| \bu_{23,j'}\|^2-(J_2J_3)^2\\
		=&\sum_{j=1}^{J_2J_3}\Big(\EE\langle \be_1, \bu_{23,j}\rangle^4-\|\bu_{23,j}\|^4\Big)+\|\bU_{23}\|_{\rm F}^2\|\bU_{23}\|_{\rm F}^2-(J_2J_3)^2\\
		=&\sum_{j=1}^{J_2J_3}\Big(\EE\langle \be_1, \bu_{23,j}\rangle^4-\|\bu_{23,j}\|^4\Big)=O(J_2J_3).
	\end{align*}
	We now deal with the second term in RHS of (\ref{eq:ZZbound_2terms}). Observe that
	$$
	g(\bE_2)=\sum_{1\leq i_1\neq i_2\leq I_1}\tilde{\bu}_{i_1}\tilde{\bu}_{i_2}^{\top}\langle \be_{2,i_1}, \be_{2,i_2}\rangle
	$$
	is a generalized U-statistic. Let $\{\tilde \be_{2,i}\}_{i=1}^{I_1}$ be an independent copy of $\{\be_{2,i}\}_{i=1}^{I_1}$. By the tail probability of decoupling of U-statistics (\cite[Theorem 3.4.1]{de2012decoupling}), for all $t>0$,
	\begin{align}\label{eq:Ustat_decouple}
		\PP\Big(\Big\| \sum_{1\leq i_1\neq i_2\leq I_1}\tilde{\bu}_{i_1}\tilde{\bu}_{i_2}^{\top}\langle \be_{2,i_1}, \be_{2,i_2}\rangle\Big\|\geq t\Big)\leq C_1\cdot\PP\Big(\Big\| \sum_{1\leq i_1\neq i_2\leq I_1}\tilde{\bu}_{i_1}\tilde{\bu}_{i_2}^{\top}\langle \be_{2,i_1}, \tilde\be_{2,i_2}\rangle\Big\|\geq C_2t\Big)
	\end{align}
	for some absolute constants $C_1,C_2>0$. Clearly,
	\begin{align}\label{eq:i1i2bound_decouple}
		\Big\|  \sum_{1\leq i_1\neq i_2\leq I_1}\tilde{\bu}_{i_1}\tilde{\bu}_{i_2}^{\top}\langle \be_{2,i_1}, \tilde\be_{2,i_2}\rangle\Big\|\leq \big\| \big(\bU_1^{\top}\bE_2\big)\big(\bU_1^{\top}\tilde{\bE}_2\big)^{\top} \big\|+\Big\|  \sum_{i=1}^{I_1}\tilde{\bu}_{i}\tilde{\bu}_{i}^{\top}\langle \be_{2,i}, \tilde\be_{2,i}\rangle \Big\|.
	\end{align}
	By a similar treatment as in the proof of Lemma~\ref{lem:AZbound}, we get
	$$
	\PP\big(\|\bU_1^{\top}\bE_2\|\geq C_3\sqrt{J_1\log I_1}+C_4\sqrt{J_2J_3}\log^2 I_1\big)\leq I_1^{-2}.
	$$
	Denote $\mathfrak{E}_1$ the above event. To get a sharp upper bound for $\|(\bU_1^{\top}\bE_2)(\bU_1^{\top}\tilde{\bE}_2)^{\top}\|$, we fix $\tilde{\bE_2}$ and recall the definition $\bU_1^{\top}\bE_2=\bU_1^{\top}\calM_1(\calE)(\bU_2\otimes \bU_3)$. Define $\bE_1=\bU_1^{\top}\calM_1(\calE)\in\RR^{J_1\times (I_2I_3)}$, which has i.i.d. columns $\{\be_{1,i}\}_{i=1}^{I_2I_3}$. Denote $\{\tilde{\bu}_{23,i}^{\top}\}_{i=1}^{I_2I_3}$ the rows of $\bU_2\otimes \bU_3$. Then, we write
	\begin{align*}
		(\bU_1^{\top}\bE_2)(\bU_1^{\top}\tilde{\bE}_2)^{\top}=\sum_{i=1}^{I_2I_3} \be_{1,i}\tilde{\bu}_{23,i}^{\top}(\bU_1^{\top}\tilde{\bE}_2)^{\top}.
	\end{align*}
	Similarly by the treatment of the proof of Lemma~\ref{lem:AZbound}, conditioned on $\tilde\bE_2$, we have
	\begin{align*}
		\PP\big(\|(\bU_1^{\top}\bE_2)(\bU_1^{\top}\tilde{\bE}_2)^{\top}\|\geq \|\bU_1^{\top}\tilde{\bE}_2\|\cdot \big(C_3\sqrt{J_1\log I_1}+C_4\sqrt{J_1}\log^2I_1\big) \big| \tilde{\bE}_2\big)\leq I_1^{-2}.
	\end{align*}
	Together with the event $\mathfrak{E}_1$, we conclude that with probability at least $1-2I_1^{-2}$,
	\begin{align}\label{eq:UEUEtop}
		\|(\bU_1^{\top}\bE_2)(\bU_1^{\top}\tilde{\bE}_2)^{\top}\| \leq C_3'J_1\log^{5/2}I_1+C_4'\sqrt{J_1J_2J_3}\log^4I_1
	\end{align}
	for some absolute constants $C_3',C_4'>0$.
	
	For the second term in (\ref{eq:i1i2bound_decouple}), we still apply the truncation treatment as in the proof of Lemma~\ref{lem:AZbound}. In this case, note that there exists an event $\mathfrak{E}_2$ with $\PP(\mathfrak{E}_2)\geq 1-I_1^{-2}$ such that $\max_{i\in[I_1]} \|\tilde{\be}_{2,i}\|\leq C_0\sqrt{J_2J_3}\log I_1$. Conditioned on $\tilde{\be}_{2,i}$, we have $\PP(|\langle \be_{2,i},\tilde{\be}_{2,i}\rangle|\geq C_1'\|\tilde{\be}_{2,i}\|\log I_1)\leq I_1^{-4}$ for some absolute constant $C_1'>0$. By a similar proof, we can obtain
	\begin{align}\label{eq:iuue2i}
		\PP\Big(\Big\|  \sum_{i=1}^{I_1}\tilde{\bu}_{i}\tilde{\bu}_{i}^{\top}\langle \be_{2,i}, \tilde\be_{2,i}\rangle\Big\|\geq C_3\sqrt{J_2J_3\log I_1}
		&+C_4\sqrt{J_2J_3}\log^{3}I_1\Big)\leq 2I_1^{-2}.
	\end{align}
	Putting (\ref{eq:iuue2i}), (\ref{eq:UEUEtop}), (\ref{eq:i1i2bound_decouple}), (\ref{eq:Ustat_decouple}) together, we get
	\begin{align*}
		\PP\Big(\Big\|  \sum_{1\leq i_1\neq i_2\leq I_1}\tilde{\bu}_{i_1}\tilde{\bu}_{i_2}^{\top}\langle \be_{2,i_1}, \be_{2,i_2}\rangle\Big\|\geq C_3'J_1\log^{5/2}I_1+C_4'\sqrt{J_1J_2J_3}\log^4I_1\Big)\leq 4I_1^{-2}.
	\end{align*}
	Then, together with (\ref{eq:ue2e2u}) and (\ref{eq:ZZbound_2terms}), we get with probability at least $1-5I_1^{-2}$ that
	$$
	\|\bZ_1\bZ_1^{\top}-J_2J_3\bP_1\|\leq C_3\sqrt{J_1J_2J_3}\log^4I_1+C_4J_1\log^{5/2}I_1,
	$$
	which proves Lemma~\ref{lem:ZZ_bound}.
\end{proof}

\vspace{1cm}

\begin{proof}[Proof of Lemma~\ref{lem:UEG2G3}]
	The strategy is similar to that in the proof of Lemma~\ref{lem:AZbound}. Denote $\bE_1=\calM_1(\calE)(\bG_2\otimes\bG_3)/\sqrt{I_2I_3}$ the $I_1\times (R_2R_3)$ random matrices with i.i.d rows $\{\be_{1,i}^{\top}\}_{i=1}^{I_1}$ satisfying $\EE \be_{1,i}={\bf 0}$ and $\EE \be_{1,i}\be_{1,i}^{\top}=\bI_{R_2R_3}$. Then,
	\begin{align*}
		\bU_1^{\top}\calM_1(\calE)(\bG_2\otimes \bG_3)=\sqrt{I_2I_3}\bU_1^{\top}\bE_1=\sqrt{I_2I_3}\cdot\sum_{i=1}^{I_1}\tilde{\bu}_i \be_{1,i}^{\top},
	\end{align*}
	where $\{\tilde{\bu}_i^{\top}\}_{i=1}^{I_1}$ denotes the rows of $\bU_1$. By a similar treatment, we have
	$$
	\PP\Big(\max_{1\leq i\leq I_1}\|\be_{1,i}\|^2\geq C_1R_2R_3\log^2 I_1\Big)\leq I_1^{-4}.
	$$
	Denote the above event by $\mathfrak{E}_0$. Define $\delta_1=C_1'\sqrt{R_2R_3}\log I_1$ so that $\PP(\max_i \|\be_{1,i}\|\geq \delta_1)\leq I_1^{-4}$. Then, we write
	\begin{align}\label{eq:U1EG2G3_2terms}
		\bU_1^{\top}\calM_1(\calE)(\bG_2\otimes \bG_3)=\sqrt{I_2I_3}\cdot \sum_{i=1}^{I_1}\tilde{\bu}_i\be_{1,i}^{\top}\mathbbm{1}(\|\be_{1,i}\|\leq \delta_1)+\sqrt{I_2I_3}\cdot \sum_{i=1}^{I_1}\tilde{\bu}_i\be_{1,i}^{\top}\mathbbm{1}(\|\be_{1,i}\|> \delta_1).
	\end{align}
	The second term in RHS of (\ref{eq:U1EG2G3_2terms}) is simply $0$ on event $\mathfrak{E}_0$. It suffices to investigate the first term on RHS of (\ref{eq:U1EG2G3_2terms}). We will apply the matrix Bernstein inequality as in the proof of Lemma~\ref{lem:AZbound}. Indeed, we can show that
	$$
	\PP\bigg(\Big\| \sum_{i=1}^{I_1}\tilde{\bu}_i\be_{1,i}^{\top}\mathbbm{1}(\|\be_{1,i}\|\leq \delta_1)\Big\|\geq C_1\sqrt{J_1\log I_1}+C_2\sqrt{R_2R_3}\log^{2}I_1\bigg)\leq I_1^{-2}
	$$
	for some absolute constants $C_1, C_2>0$. Since the proof is identical to the proof of Lemma~\ref{lem:AZbound}, we skip it here.
\end{proof}

\vspace{1cm}

\begin{proof}[Proof of Lemma~\ref{lem:UEsupGR}]
	We only prove (\ref{eq:supUEGU3A}) since the proof of (\ref{eq:supUEU2AU3B}) is similar. We begin with a standard discretization step, see \cite[Lemma~5]{zhang2018tensor}. There exists a subset $\mathfrak{D}_{1/3}\subset \mathfrak{B}(J_3,R_3)$ such that $\log {\rm Card}(\mathfrak{D}_{1/3})\leq c_1J_3R_3$ for some absolute constant $c_1>0$ and for any $\bB\in \mathfrak{B}(J_3,R_3)$,
	$$
	\min_{\bD\in\mathfrak{D}_{1/3}} \|\bB-\bD\|_{\rm F}\leq 1/3.
	$$
	It is easy to show that
	\begin{align}\label{eq:suptomax}
		\sup_{\bA\in \mathfrak{B}(J_3,R_3)} \big\|\bU_1^{\top}\calM_1(\calE)(\bG_2\otimes \bU_3\bA) \big\|/\sqrt{I_2}\leq \frac{3}{2}\cdot \max_{\bD\in \mathfrak{D}_{1/3}} \big\|\bU_1^{\top}\calM_1(\calE)(\bG_2\otimes \bU_3\bD) \big\|/\sqrt{I_2}.
	\end{align}
	It suffices to prove the upper bound in RHS of (\ref{eq:suptomax}). Under Assumption~\ref{assump:noise}, there exists an event $\mathfrak{E}_0$ with $\PP(\mathfrak{E}_0)\geq 1-I_1^{-5}$ in which
	$$
	\max_{\omega\in[I_1]\times [I_2]\times [I_3]} |e_{\omega}|\leq C_0\log I_1
	$$
	for some absolute constant $C_0>0$. Denote $\delta_0=C_0\log I_1$ and $\|\calE\|_{\infty}=\max_{\omega}|e_{\omega}|$. We write
	\begin{align}\label{eq:maxU1_2terms}
		\max_{\bD\in \mathfrak{D}_{1/3}} \big\|\bU_1^{\top}\calM_1(\calE)&(\bG_2\otimes \bU_3\bD) \big\|/\sqrt{I_2}
		\leq \max_{\bD\in \mathfrak{D}_{1/3}} \big\|\bU_1^{\top}\calM_1(\calE)(\bG_2\otimes \bU_3\bD)\mathbbm{1}(\|\calE\|_{\infty}\leq \delta_0) \big\|/\sqrt{I_2}\notag\\
		&+\max_{\bD\in \mathfrak{D}_{1/3}} \big\|\bU_1^{\top}\calM_1(\calE)(\bG_2\otimes \bU_3\bD)\mathbbm{1}(\|\calE\|_{\infty}>\delta_0) \big\|/\sqrt{I_2}.
	\end{align}
	Conditioned on event $\mathfrak{E}_0$, the second term in RHS of (\ref{eq:maxU1_2terms}) is simply $0$. It suffices to prove the upper bound of first term in RHS of (\ref{eq:maxU1_2terms}). Write
	\begin{align}\label{eq:U1EG2U3D_delta}
		\max_{\bD\in \mathfrak{D}_{1/3}} \big\|\bU_1^{\top}&\calM_1(\calE)(\bG_2\otimes \bU_3\bD)\mathbbm{1}(\|\calE\|_{\infty}\leq \delta_0) \big\|/\sqrt{I_2}\notag\\
		\leq& \max_{\bD\in \mathfrak{D}_{1/3}} \big\|\bU_1^{\top}\EE\big[\calM_1(\calE)\mathbbm{1}(\|\calE\|_{\infty}\leq \delta_0)\big](\bG_2\otimes \bU_3\bD) \big\|/\sqrt{I_2}\notag\\
		+ \max_{\bD\in \mathfrak{D}_{1/3}} &\big\|\bU_1^{\top}\Big[\calM_1(\calE)\mathbbm{1}(\|\calE\|_{\infty}\leq \delta_0)-\EE\calM_1(\calE)\mathbbm{1}(\|\calE\|_{\infty}\leq \delta_0)\Big](\bG_2\otimes \bU_3\bD) \big\|/\sqrt{I_2}.
	\end{align}
	To upper bound the first term in RHS of (\ref{eq:U1EG2U3D_delta}), note that $\forall \omega\in[I_1]\times[I_2]\times[I_3]$,
	$$
	\big|\EE e_{\omega}\mathbbm{1}(|e_{\omega}|\leq \delta_0)\big|=\big|\EE e_{\omega}\mathbbm{1}(|e_{\omega}|>\delta_0)\big|\stackrel{\tiny \textrm{Cauchy-Scwharz}}{\leq} \PP^{1/2}(|e_{\omega}|> \delta_0).
	$$
	Then, the first term in RHS of (\ref{eq:U1EG2U3D_delta}) is bounded as
	\begin{align}\label{eq:U1Ein_exp}
		\max_{\bD\in \mathfrak{D}_{1/3}} \big\|&\bU_1^{\top}\EE\big[\calM_1(\calE)\mathbbm{1}(\|\calE\|_{\infty}\leq \delta_0)\big](\bG_2\otimes \bU_3\bD) \big\|/\sqrt{I_2}\notag\\
		\leq&\|\EE\calM_1(\calE)\mathbbm{1}(\|\calE\|_{\infty}\leq \delta_0)\|_{\rm F}=\left(\sum_{\omega\in[I_1]\times[I_2]\times[I_3]} \PP(|e_{\omega}|>\delta_0)\right)^{1/2}\leq (I_1I_2I_3I_1^{-5})^{1/2}=O(I_1^{-1}).
	\end{align}
	We now continue the upper bound the second term in RHS of (\ref{eq:U1EG2U3D_delta}). Similarly, let $\mathfrak{R}_{1/3}(J_1)$ and $\mathfrak{R}_{1/3}(R_2R_3)$ the $1/3$-net of $\mathfrak{B}(J_1,1)$ and $\mathfrak{B}(R_2R_3,1)$, respectively.  Then, given any vector $\bu\in\mathfrak{B}(J_1,1)$ and $\bw\in\mathfrak{B}(R_2R_3,1)$, we have
	$$
	\min_{\bv\in \mathfrak{R}_{1/3}(J_1)}\|\bu-\bv\|\leq 1/3\quad {\rm and}\quad \min_{\bv\in\mathfrak{R}_{1/3}(R_2R_3)}\|\bv-\bw\|\leq 1/3
	$$
	where $\|\cdot\|$ represents $\ell_2$-norm for vectors. Meanwhile, $\log {\rm Card}(\mathfrak{R}_{1/3}(J_1))\leq c_1J_1$ and $\log{\rm Card}(\mathfrak{R}_{1/3}(R_2R_3))\leq c_2R_2R_3$ for some absolute constants $c_1, c_2>0$. Then,
	\begin{align*}
		\max_{\bD\in \mathfrak{D}_{1/3}} &\big\|\bU_1^{\top}\Big[\calM_1(\calE)\mathbbm{1}(\|\calE\|_{\infty}\leq \delta_0)-\EE\calM_1(\calE)\mathbbm{1}(\|\calE\|_{\infty}\leq \delta_0)\Big](\bG_2\otimes \bU_3\bD) \big\|/\sqrt{I_2}\\
		\leq&  \frac{9}{2}\max_{\bD\in \mathfrak{D}_{1/3}}\max_{\substack{\bu\in\mathfrak{R}_{1/3}(J_1)\\ \bw\in \mathfrak{R}_{1/3}(R_2R_3)}} \bu^{\top}\bU_1^{\top}\Big[\calM_1(\calE)\mathbbm{1}(\|\calE\|_{\infty}\leq \delta_0)-\EE\calM_1(\calE)\mathbbm{1}(\|\calE\|_{\infty}\leq \delta_0)\Big](\bG_2\otimes \bU_3\bD) \bw/\sqrt{I_2}.
	\end{align*}
	Now, we fix $\bu,\bw$ and $\bD$, and prove a concentration inequality for the above RHS, then we finish the proof by a union bound. Denote $\tilde{\bu}=\bU_1\bu\in\RR^{I_1}$ and $\tilde{\bv}=(\bG_2\otimes \bU_3\bD)\bw/\sqrt{I_2}\in\RR^{I_2I_3}$. Clearly, we have $\max\big\{\|\tilde\bu\|,\|\tilde\bv\|\big\}\leq 1$. Now, we write
	\begin{align}\label{eq:tildeuM-Etildev}
		\tilde{\bu}^{\top}\Big[\calM_1(\calE)\mathbbm{1}&(\|\calE\|_{\infty}\leq \delta_0)-\EE\calM_1(\calE)\mathbbm{1}(\|\calE\|_{\infty}\leq \delta_0)\Big] \tilde{\bv}\notag\\
		=&\sum_{\omega=(\omega_1,\omega_2,\omega_3)\in[I_1]\times [I_2]\times [I_3]}\tilde u(\omega_1)\tilde v(\omega_2\omega_3)\big[e_{\omega}\mathbbm{1}(|e_{\omega}|\leq \delta_0)-\EE e_{\omega}\mathbbm{1}(|e_{\omega}|\leq \delta_0)\big],
	\end{align}
	which is a sum of independent centered random variables. Observe that
	$$
	\Big|\tilde u(\omega_1)\tilde v(\omega_2\omega_3)\big[e_{\omega}\mathbbm{1}(|e_{\omega}|\leq \delta_0)-\EE e_{\omega}\mathbbm{1}(|e_{\omega}|\leq \delta_0)\big]\Big|\leq 2\delta_0|\tilde{u}(\omega_1)||\tilde{v}(\omega_2\omega_3)|,\quad \forall \omega
	$$
	implying that each term in RHS of (\ref{eq:tildeuM-Etildev}) is a bounded random variable. By applying Hoeffiding's inequality (\cite{hoeffding1994probability}) to (\ref{eq:tildeuM-Etildev}), we get
	\begin{align*}
		\PP\Big(\Big|\tilde{\bu}^{\top}\big[\calM_1(\calE)\mathbbm{1}(\|\calE\|_{\infty}\leq \delta_0)-\EE\calM_1(\calE)\mathbbm{1}(\|\calE\|_{\infty}\leq \delta_0)\big] \tilde{\bv}\Big|\geq t\Big)\leq& 2\exp\Big(\frac{-2t^2}{4\delta_0^2\sum_{\omega} \tilde{u}(\omega_1)^2 \tilde{v}(\omega_2\omega_3)^2}\Big)\\
		\leq& 2\exp(-t^2/(2\delta_0^2)).
	\end{align*}
	Since the cardinality of the product set of $\mathfrak{D}_{1/3}, \mathfrak{R}_{1/3}(J_1), \mathfrak{R}_{1/3}(R_2R_3)$ is bounded by $3^{C_0'(J_3R_3+J_1+R_2R_3)}$ for some absolute constant $C_0'$, we apply a union bound and get
	\begin{align*}
		\PP\Big(\max_{\bD\in \mathfrak{D}_{1/3}}&\max_{\substack{\bu\in\mathfrak{R}_{1/3}(J_1)\\ \bw\in \mathfrak{R}_{1/3}(R_2R_3)}}\Big| \bu^{\top}\bU_1^{\top}\big[\calM_1(\calE)\mathbbm{1}(\|\calE\|_{\infty}\leq \delta_0)-\EE\calM_1(\calE)\mathbbm{1}(\|\calE\|_{\infty}\leq \delta_0)\big](\frac{\bG_2}{\sqrt{I_2}}\otimes \bU_3\bD) \bw\Big|\geq t\Big)\\
		\leq& 2\exp\left(-\frac{t^2}{2\delta_0^2}+C_1(J_3R_3+J_1+R_2R_3)\right)\stackrel{t=2\sqrt{C_1(J_1+J_3R_3+R_2R_3)}\delta_0\log^{1/2}I_1}{\leq } 2I_1^{-2}
	\end{align*}
	implying that with probability at least $1-2I_1^{-2}$,
	\begin{align}\label{eq:U1Edelta_0}
		\max_{\bD\in \mathfrak{D}_{1/3}} &\big\|\bU_1^{\top}\big[\calM_1(\calE)\mathbbm{1}(\|\calE\|_{\infty}\leq \delta_0)-\EE\calM_1(\calE)\mathbbm{1}(\|\calE\|_{\infty}\leq \delta_0)\big](\bG_2\otimes \bU_3\bD) \big\|/\sqrt{I_2}\notag\\
		&\leq C_1\sqrt{J_1+J_3R_3+R_2R_3}\log^{3/2}I_1.
	\end{align}
	Putting together (\ref{eq:U1Edelta_0}), (\ref{eq:U1Ein_exp}), (\ref{eq:U1EG2U3D_delta}), (\ref{eq:maxU1_2terms}) and (\ref{eq:suptomax}), we conclude that with probability at least $1-4I_1^{-2}$,
	$$
	\sup_{\bA\in \mathfrak{B}(J_3,R_3)} \big\|\bU_1^{\top}\calM_1(\calE)(\bG_2\otimes \bU_3\bA) \big\|/\sqrt{I_2}\leq C_3\sqrt{J_1+J_3R_3+R_2R_3}\log^{3/2}I_1
	$$
	for some absolute constant $C_3>0$. This proves (\ref{eq:supUEGU3A}) of Lemma~\ref{lem:UEsupGR}.
\end{proof}

\vspace{1cm}

\begin{proof}[Proof of Lemma~\ref{lem:EG1G2G3_tilde}]
	Recall by definition that $\tilde\bG_m^{\top}\tilde{\bG}/I_m={\bI}_{R_m}$. Moreover, the column space of $\tilde \bG_m$ is a subspace of column space of $\bPhi_m(\bX_m)$. Denote $\bU_m\in\RR^{I_m\times J_m}$ the left singular vectors of $\bPhi_m(\bX_m)$. Then, there exists a $\tilde \bB_m\in\RR^{J_m\times R_m}$ so that $\tilde\bG_m=\bU_m\tilde\bB_m$ and $\tilde \bB_m^{\top}\tilde \bB_m/I_m=\bI_{R_m}$ where $\tilde\bB_m$ is dependent with $\calE$ while $\bU_m$ is independent with $\calE$. Denote $\mathcal{G}(d_1,d_2):=\{\bB\in\RR^{d_1\times d_2}, \bB^{\top}\bB=\bI_{d_2}\}$.
	Then,
	$$
	\frac{\big\|\calM_1\big(\calE\times_1 \tilde\bG_1^{\top}\times_2 \tilde\bG_2^{\top}\times_3 \tilde\bG_3^{\top}\big) \big\|}{\sqrt{I_1I_2I_3}}\leq \sup_{\tilde{\bB}_m\in\calG(J_m,R_m)} \big\|\calM_1\big(\calE\times_1 (\bU_1\tilde\bB_1)^{\top}\times_2 (\bU_2\tilde\bB_2)^{\top}\times_3 (\bU_3\tilde\bB_3)^{\top}\big) \big\|.
	$$
	By choosing $1/5$-nets of $\calG(J_1,R_1), \calG(J_2, R_2)$ and $\calG(J_3,R_3)$ (e.g., by the covering number of Grassmannians in \cite{koltchinskii2015optimal}), respectively, we can reproduce the proof of Lemma~\ref{lem:UEsupGR} and show that with probability at least $1-I_1^{-2}$,
	$$
	\big\|\calM_1\big(\calE\times_1 \tilde\bG_1^{\top}\times_2 \tilde\bG_2^{\top}\times_3 \tilde\bG_3^{\top}\big) \big\|\leq C_1\sqrt{J_1R_1+J_2R_2+J_3R_3+R_1+R_2R_3}\log^{3/2} I_1
	$$
	for some absolute constant $C_1>0$. Since $J_1\asymp J_2\asymp J_3$ and $R_1\geq R_2\geq R_3$, we can simplify the upper bound to $C_1'\sqrt{J_1R_1}\log^{3/2} I_1$ and finish the proof of Lemma~\ref{lem:EG1G2G3_tilde}.
\end{proof}

\section{More real data applications}

%

\subsection{Human Brain Connection Data} \label{append:HCP}

As an additional illustration of using the STEFA and IP-SVD for explanatory data analysis, 
we consider partitioning the brain connectivity according to the $136 \times 4$ covariate-relevant loading matrix $\bPhi_3(\bX)\bB_3$. 
As mentioned in the main text, we choose $\bPhi_3(\bX)$ generated by polynomial basis of order 1, which is a similar linear setting as that in \cite{hu2021generalized} with identity link function. 
Each column of $\bPhi_3(\bX)\bB_3$  corresponds to one of the four directions with largest variance among individual features in $\bX$. 
The meaning of each latent direction can be inferred from $\bB_3$. 
Figure \ref{fig:B2} presents the heat map of $\bB_3$, showing that the four factors weight mostly on the four columns of $\bPhi_3(\bX)$, namely, all-ones vector, female variable, and Age 22-25 variable and 31+ variable, respectively.
Then each factor is interpreted as the effects associated with global average, female, Age 22-25 variable and 31+. 

\begin{figure}[htpb!]
	\centering
	\includegraphics[width=.5\textwidth]{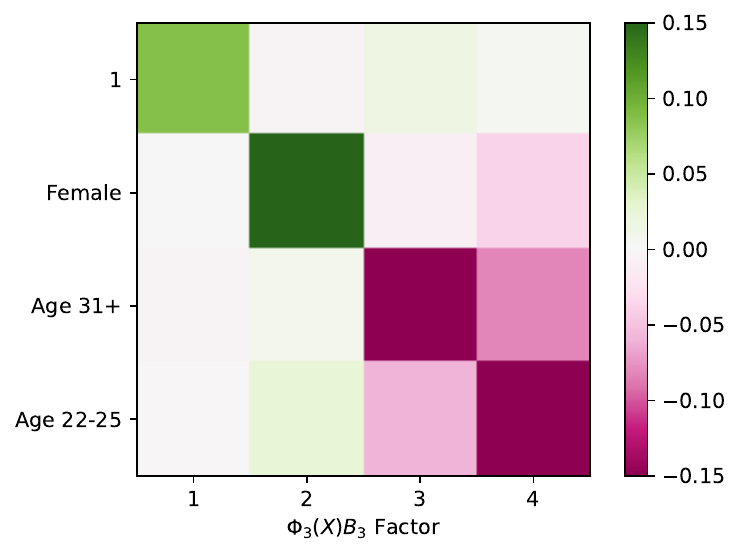}
	\caption{Heat map of the coefficient $\bB_3$ in the covariate-relevant loading with polynomial basis of order one.}  \label{fig:B2}
\end{figure}

We obtained a $68 \times 68 \times 4$ connectivity tensor by projecting the original connectivity tensor on the column space of $\bPhi_3(\bX)\bB_3$. 
As a result, the four slices along the third mode corresponds to the connectivity matrix for each of the global average, female, Age 22-25 variable and 31+ effects. 
We divide the $68$ regions of the brain into two clusters based on the connectivity matrix for each of the effects. 
The clustering results are plotted in Figure \ref{fig:brain-region-clusters}. 
The connection within the same cluster is stronger than that between clusters.
Some connectivity patterns can be observed. 
For example, the global connection exhibits clear left and right spatial separation and the age 22 - 25 group shows additional inter-hemispheric connectivity. 
While such explanatory analysis can provide some interesting observation, more rigorous methods, such as statistical testing procedures, need to be developed to support any scientific claim.
These are important directions for future statistics researches.  

\begin{figure}[htpb!]
	\centering
	\begin{subfigure}[b]{0.45\textwidth}
		\centering
		\includegraphics[width=\textwidth]{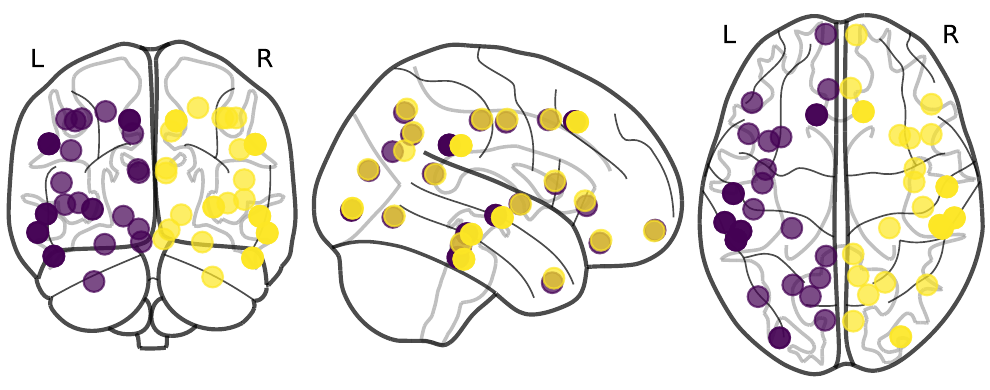}
		\caption{Global effect}
	\end{subfigure}
	\hspace{2em}
	\begin{subfigure}[b]{0.45\textwidth}
		\centering
		\includegraphics[width=\textwidth]{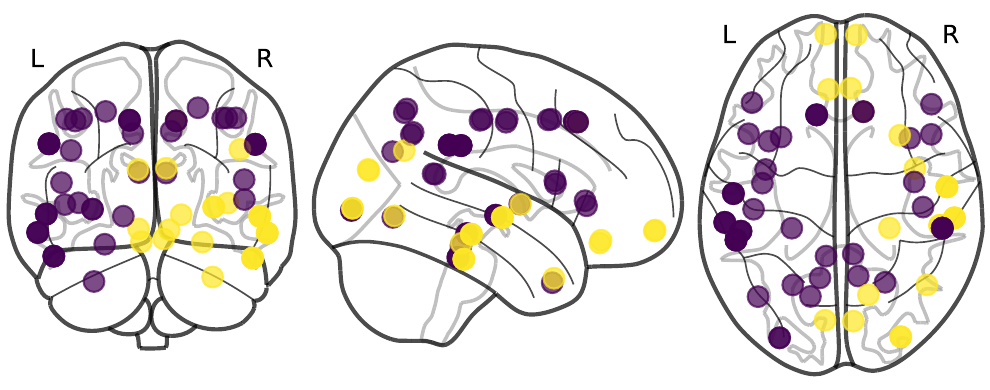}
		\caption{Female effect}
	\end{subfigure}
	\par\bigskip
	\begin{subfigure}[b]{0.45\textwidth}
		\centering
		\includegraphics[width=\textwidth]{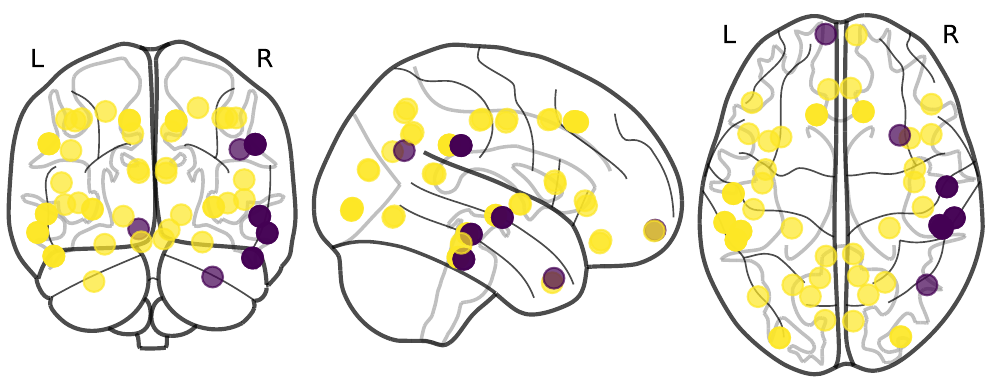}
		\caption{Age 22-25}
	\end{subfigure}
	\hspace{2em}
	\begin{subfigure}[b]{0.45\textwidth}
		\centering
		\includegraphics[width=\textwidth]{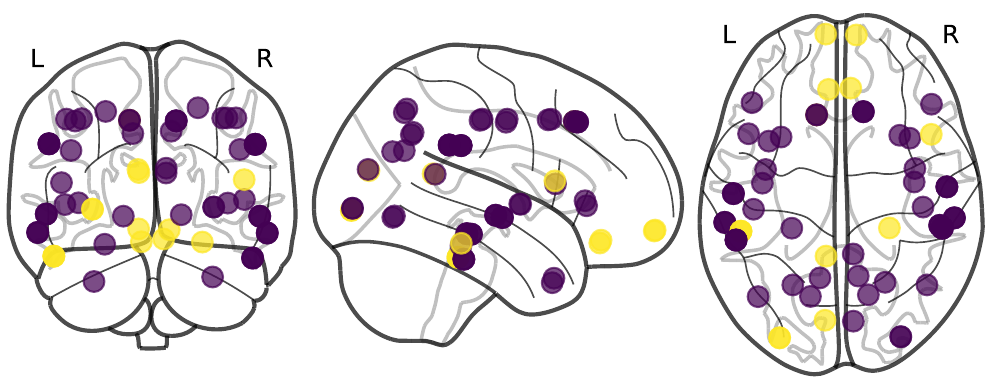}
		\caption{Age 31+}
	\end{subfigure}
	\caption{Partition of brain regions for each latent dimension corresponding to the covariate-relevant loading.}
	\label{fig:brain-region-clusters}
\end{figure}

\end{appendices}